\documentclass[ejsv2,noshowframe,preprint]{imsart}
\RequirePackage[numbers]{natbib}
\RequirePackage{comment}
\RequirePackage[colorlinks,citecolor=blue,urlcolor=blue]{hyperref}
\RequirePackage{graphicx}

\arxiv{2010.00000}
\startlocaldefs
\theoremstyle{plain}

\newtheorem{theorem}{Theorem}[section]
\newtheorem{lem}[theorem]{Lemma}
\newtheorem{corollary}[theorem]{Corollary}
\newtheorem{proposition}[theorem]{Proposition}

\theoremstyle{definition}
\newtheorem{definition}[theorem]{Definition}

\theoremstyle{remark}

\newcommand{\CDF}{\mathrm{CDF}}
\newcommand{\BK}{\mathrm{BK}}
\newcommand{\MZW}{\mathrm{MZW}}
\newcommand{\dom}{\mathrm{Dom}}

\DeclareMathOperator{\Var}{Var}
\DeclareMathOperator{\Cov}{Cov}

\newcommand{\MC}{\mathrm{MC}}
\renewcommand{\d}{\mathrm{d}}

\endlocaldefs

\usepackage{xcolor}

\usepackage{amsmath,amsthm,mathrsfs,systeme,lipsum,framed,texshade,listings,stmaryrd,dsfont,soul,float}

\newcommand{\rev}[1]{\textcolor{black}{#1}}

\begin{document}

\setstcolor{blue}
\setulcolor{blue}
\begin{frontmatter}
\title{Simulation of warping processes with applications to temperature data}
%\title{A sample article title with some additional note\thanksref{t1}}
\runtitle{Simulation of warping processes with applications to temperature data}
%\thankstext{T1}{A sample additional note to the title.}

\begin{aug}
%%%%%%%%%%%%%%%%%%%%%%%%%%%%%%%%%%%%%%%%%%%%%%%
%% Only one address is permitted per author. %%
%% Only division, organization and e-mail is %%
%% included in the address.                  %%
%% Additional information can be included in %%

\author[A]{\fnms{Nolwenn}~\snm{Le Méhauté}\ead[label=e1]{nolwenn.le-mehaute@univ-grenoble-alpes.fr}},
\author[A]{\fnms{Jean-François}~\snm{Coeurjolly}\ead[label=e2]{ jean-francois.coeurjolly@univ-grenoble-alpes.fr}}
\and
\author[B]{\fnms{Marie-Hélène}~\snm{Descary}\ead[label=e3]{descary.marie-helene@uqam.ca}}

%%%%%%%%%%%%%%%%%%%%%%%%%%%%%%%%%%%%%%%%%%%%%%
%% Addresses                                %%
%%%%%%%%%%%%%%%%%%%%%%%%%%%%%%%%%%%%%%%%%%%%%%
\address[A]{Univ. Grenoble Alpes, CNRS, Grenoble INP\footnote{Institute of Engineering Univ. Grenoble Alpes}, LJK, 38000 Grenoble, France\printead[presep={,\ }]{e1,e2}}

\address[B]{Department of Mathematics,
Université du Québec à Montréal, 
Canada\printead[presep={,\ }]{e3}}

\runauthor{N. Le Méhauté et al.}
\end{aug}

\begin{abstract}
Curve registration plays a major role in functional data analysis by separating amplitude and phase variation through warping functions and the accurate simulation of warping processes is essential for developing statistical methods that properly account for phase variability in functional data. In this paper, we focus on the simulation of continuous warping processes with a prescribed expectation and a controllable variance. We study and compare three procedures, including two existing methods and a new algorithm based on randomized empirical cumulative distribution functions. For each approach, we provide an operational description and establish theoretical results for the first two moments of the simulated processes. A numerical study illustrates the theoretical findings and highlights the respective merits of the three methods. \rev{Finally, we present an application to the analysis of hourly temperature data spatially averaged over Canada, based on simulated realizations from a warping process estimated using spatially averaged temperature quantile functions.}

\end{abstract}



\end{frontmatter}
%%%%%%%%%%%%%%%%%%%%%%%%%%%%%%%%%%%%%%%%%%%%%%
%% Please use \tableofcontents for articles %%
%% with 50 pages and more                   %%
%%%%%%%%%%%%%%%%%%%%%%%%%%%%%%%%%%%%%%%%%%%%%%
\tableofcontents

\section{Introduction}

Functional data analysis focuses on the study of observations that take the form of functions, such as curves or images. Improvements in measurement accuracy and in data storage have led to the emergence of functional data analysis in a wide range of fields such as biomechanics (e.g \citep{ryan2006functional}), medicine (e.g \citep{orozco2025functional} for a review) or climate science (e.g \citep{leeming2025functional}).

A common example where functional data occur is given by curves resulting from the replication of the same random experiment. Such a sample generally highlights lateral displacements, referred to as phase variability, that originate from measurement errors or from intrinsic characteristic of the underlying phenomenon. The Berkeley growth data provide an example which highlights that ignoring phase variability leads to analyzing a mean curve that does not conserve the common features of the curves and thus results in an erroneous analysis \citep{ramsay2005functional}. Thus, phase variability plays a major role in statistical analysis of functional datasets. 
In the presence of phase variability, the standard approach is to account for it by solving a curve registration problem. From a mathematical perspective, this problem is expressed as follows. Given $f_1, \dots, f_n,$ an $n$-sample of curves belonging to a subset of $L^2([0,1]),$ the associated curves without phase variation, denoted by $g_1, \dots, g_n$, are given by the set of equations $g_i = f_i \circ w_i,$ for all $i \in \{1, \dots, n\}.$ In these expressions, the functions $w_1, \dots, w_n$ correspond to $n$ realizations of a same stochastic process $w,$ and are respectively called warping functions and warping process.
\rev {
Depending on the setting or the application under study, warping functions may belong to different classes of functions (for examples see \citet{marron2015functional}). In the sequel, focusing on preserving the order of events and roaming the line segment $[0,1]$ continuously, we assume that warping functions belong to the set $\Gamma$ defined as follows:
\begin{equation} \label{expression_Gamma_set}
\Gamma =  \{ \gamma:[0,1] \to [0,1] \mid \gamma(0) = 0, \gamma(1) =1, \gamma \text{ is a strictly increasing continuous function}\}. 
\end{equation}
In particular, unless other assumption specified, we do not impose the warping functions to be differentiable.
}

The estimation of warping functions associated with given datasets has been extensively studied. A 2015 survey by \citet{marron2015functional} identified four main classes of estimation procedures: landmark registration, equivalence classes approach, variational methods and Fisher-Rao metric minimization. More recent approaches rely on a Bayesian framework (e.g \citep{lu2017bayesian,kim2025sequential}).
In contrast, the literature on simulation procedures for paths of warping processes is rather limited. To the best of our knowledge, only two key papers \citep{bharath2020distribution, ma2024stochastic} address this question, each proposing a distinct algorithm. Given the important role of phase variation in functional data analysis, this gap in the literature is a limiting factor in the development of models that rely on simulated warping process. 

One key contribution of this work is the proposal of a third algorithm, motivated by the fact that the cumulative distribution function of a continuous random variable is itself a warping function. The objectives of this paper are threefold: (i) to provide a detailed description of the three algorithms, making them operational and reproducible to generate a warping process in $\Gamma$; (ii) to establish the main theoretical properties of the simulated processes; and (iii) to compare them from a numerical point of view. In particular, we focus on the following expected properties: (C1) the simulated process has a target expectation; and (C2) the user can control its variability. \rev{Here, both requirements are understood pointwise: for a warping process $w$, we consider $\mathbb E[w(t)]$ and $\Var[w(t)]$, for each fixed $t\in[0,1]$.}

The first approach considered in this manuscript was introduced by \citet{bharath2020distribution}. The authors analyze a simulation procedure for warping functions originally proposed by \citet{cheng2016bayesian} and show that property (C1) holds only for the identity function, while the variability, hence (C2), cannot be prescribed. Within a distributional framework, they propose an extension of this algorithm, referred to as the BK Algorithm, which overcomes these first limitations. The idea is to start with $n$ uniformly spaced data points on $[0,1]$ for the $x$-axis, and to define random jumps using an $n$-dimensional Dirichlet distribution with parameters adapted to a targeted warping function. 
\rev{We view $n$ as an auxiliary parameter, not as a parameter to control the pointwise variance.}
The proposed warping process is then simply obtained by linear interpolation. \citet{bharath2020distribution} prove the weak convergence of the process indexed by warping functions as $n\to \infty$. We complete their work by deriving explicit expressions for the first two moments of the simulated process, thus satisfying (C1)-(C2). Moreover, we obtain bounds on the convergence rates of these two moments and thus on the $L^2([0,1])$-norm, as $n$ increases.

The second approach was proposed by \citet{ma2024stochastic}. The authors present a completely different algorithm, based on modeling a second-order centered stochastic process in $L^2([0,1])$ using a Karhunen-Loève decomposition, see e.g. \citep{wang2016functional}, and on an isometry between $L^2([0,1])$ and $\Gamma_1$, a subspace of $\Gamma$. \rev{Their framework addresses a different modelling objective and was not designed to satisfy the pointwise moment control sought in (C1) and (C2). In its original formulation, the algorithm is centered at the identity in the geometry of $\Gamma_1$.} We reparameterize the Karhunen-Loève decomposition in the most natural way and control its moments so that the resulting warping process is an element of $\Gamma_1 \subset \Gamma$ and a natural candidate to achieve property (C1). This new version, referred to as the MZW Algorithm, is also very simple to implement\rev{, and the control of the Fréchet mean and variance in the geometry of $\Gamma_1$ is presented. Because the transformation back to warping functions is nonlinear, this control does not automatically translate into control of the pointwise moments of the warping process. Our numerical study illustrates this distinction: the pointwise mean can remain different from the target even for large truncation orders in the Karhunen-Loève decomposition.} %A result similar to (C1) is however shown in the Gaussian case when the variances of the random variables in the Karhunen-Loève decomposition converge to 0. 

As mentioned above, one of the main contributions of this manuscript is the introduction of a third algorithm. Observing that any continuous cumulative distribution function is a particular case of a warping function, we propose to generate an empirical cumulative  distribution function with a prescribed expectation and then smooth it linearly, following the approach of \citet{blanke2018polygonal}, to obtain a continuous function on $[0,1]$ that serves as a natural and simple candidate for a warping process. To control variability, we further introduce into the procedure random weights drawn from an $n$-dimensional Dirichlet distribution with constant concentration parameter. The resulting method, referred to as the CDF Algorithm, is straightforward to implement and enjoys \rev{similar} properties as the warping process generated by the BK Algorithm. \rev{The interest of the BK and CDF Algorithms lies in their parsimony. They provide a simple semiparametric model $(\varphi,\theta)$, where the function $\varphi$ represents the target and is related to condition (C1), while the concentration parameter $\theta>0$ is associated with condition (C2).}

The rest of the paper is organized as follows. Section~\ref{section_description_algos} provides a detailed description of the three algorithms. Section~\ref{section::properties} presents their main theoretical properties. Simulations performed with each algorithm and illustrations of the properties announced in Section~\ref{section::properties} are presented in Section~\ref{section::Simulations}. We end this simulation study with some general recommendations. Finally, Section~\ref{section::Simulations} also contains an application to \rev{hourly temperature data spatially averaged over Canada}, illustrating the practical relevance of the proposed simulation procedures. 
Additional figures and proofs of our different results are postponed to Appendices~\ref{Appendix_autres_fig}-\ref{app:MZW}. \rev{Code and data availability are described at the end of the paper.}

\section{Warping process simulation}\label{section_description_algos}

We now describe the three simulation algorithms of warping processes introduced in this paper. The BK and CDF Algorithms, sharing some similarities, are presented in Sections \ref{sec:BK}–\ref{sec:CDF}.
Both rely on a linear interpolation of $n+2$ points, $n$ of which are random. In Section \ref{sec:MZW}, we present the MZW Algorithm proposed by \citet{ma2024stochastic} which is, as already mentioned, of a different kind since it relies on mapping a second-order stochastic process in $L^2([0,1])$ to an inner-product space contained within the set of warping functions.
\rev{
We remind the definition of the set of warping functions $\Gamma,$ introduced in \eqref{expression_Gamma_set}, which refers to 
\begin{equation*} 
\Gamma =  \{ \gamma:[0,1] \to [0,1] \mid \gamma(0) = 0, \gamma(1) =1, \gamma \text{ is a strictly increasing continuous function}\}. 
\end{equation*}}
Moreover, for $f,g \in L^2([0,1])$, we denote by $\langle f,g \rangle $ the $L^2$ inner-product between $f$ and $g$ on $[0,1]$, defined as $\langle f,g \rangle = \int_0^1 f(t)g(t) \d t$, and by $\| f \|$ the associated norm. Finally, for $i\le j \in \mathbb N$, we let $\llbracket i,j\rrbracket=\{k \in \mathbb N \mid  i\le k\le j\}$.

\subsection{Bharath and Kurtek's Algorithm (BK Algorithm)}\label{sec:BK} 

Building on the procedure originally designed by \citet[Section 4.2]{cheng2016bayesian} for Bayesian registration of functions and curves, \citet{bharath2020distribution} extend this method in several important directions. We begin by presenting the original approach.

\paragraph{Cheng, Dryden and Huang's Algorithm.} The algorithm takes three inputs: a natural number $n \in \mathbb{N}^\ast$ which specifies the number of points of the x-axis partition; a deterministic partition of the line segment $[0,1]$ denoted by $t_0^\ast <t_1^\ast < \dots < t_n^\ast < t_{n+1}^\ast  $ with $t_0^\ast=0$ and $t_{n+1}^\ast=1;$ and a positive real parameter $\theta \in \mathbb{R}_+^\ast$ which is a concentration parameter for the Dirichlet distribution. The algorithm proceeds as follows:  
\begin{enumerate} 
\item \textit{Construction of the x-axis partition.} \newline The x-axis partition is given by the deterministic vector $(t_0^\ast, t_1^\ast, \dots, t_n^\ast, t_{n+1}^\ast).$
\item \textit{Construction of the y-axis partition.} \newline
One generates a Dirichlet random vector of size $n+1,$ denoted by $(p_1, \dots, p_{n+1}),$ with vector parameter $( \theta , \dots, \theta).$ One then forms the random vector of size $n+2$, denoted by $(\Tilde{p}_0,\Tilde{p}_1, \dots \Tilde{p}_{n+1})$, where $(\Tilde{p}_1,\cdots, \Tilde{p}_{n+1})$ is the vector of cumulative sums of $(p_1,\dots, p_{n+1})$ and $\Tilde{p}_0=0.$
\item \textit{Linear Interpolation.} \newline One performs a linear interpolation of the points $(t_i^\ast, \Tilde{p}_i),$ for $i \in \llbracket 0, n+1 \rrbracket.$ 
\end{enumerate} 

\citet{bharath2020distribution} examine this algorithm and provide functional convergence results showing that, as $n \to \infty$, the limit process is given by the identity warp map. In particular, this algorithm does not allow obtaining a limit process that can be centered at a chosen warp map and whose variance term can be controlled. \citet[Corollary 1]{bharath2020distribution} propose a modification of the construction of the y-axis partition, which extends the degeneracy behaviour of the limit process as $n \to \infty$ to a larger class of distributions than the Dirichlet distribution. This result also highlights that the deterministic nature of the x-axis partition is not suitable for obtaining a limit process whose variance term can be controlled.

Algorithm 2 in \citet{bharath2020distribution} aims at overcoming this issue and is presented as able to simulate a warping process having a chosen expectation (at least as $n \to \infty$) and for which the fluctuations can be controlled. The presentation of their algorithm contains a few typos. In particular, the x-axis presented in Steps~1 and~3 is not specified, $H$ should be $H^{-1}$ in Step~1 and a parameter $\theta$ is missing as a factor of the Dirichlet parameters in Step~2. Discussions with the authors enabled us to correct the algorithm. We now present the corrected version, referred to as the BK Algorithm hereafter.\\

\noindent\fbox{\begin{minipage}{.98\textwidth}
\paragraph{BK Algorithm.} The algorithm takes three parameters as inputs: a natural number $n \in \mathbb{N}^\ast$ corresponding to the number of random points in the partition of the x-axis; a positive real number $ \theta \in \mathbb{R}_+^\ast,$  which is a concentration parameter for the Dirichlet distribution; and a warping function $\varphi$ from $[0,1]$ to $[0,1].$ The algorithm proceeds as follows. 

\begin{enumerate} 
\item \textit{Construction of the x-axis partition.}

One generates an $n$-sample from the Uniform distribution on $[0,1],$ denoted by \linebreak$(U_1,\dots, U_n).$ Let $(U_1^\ast, \dots, U_n^\ast)$ be the sorted sample and set $U_0^\ast=0$ and $U_{n+1}^\ast= 1$. 
One forms the vector $U^\ast= (U_0^\ast, U_1^\ast, \dots, U_n^\ast, U_{n+1}^\ast)$. %, where $U_0^\ast=0$ and $U_{n+1}^\ast= 1.$
\item \textit{Construction of the y-axis partition.}

One generates a Dirichlet random vector of size $n+1,$ denoted by $(\alpha_1, \dots, \alpha_{n+1}),$ with vector parameter given by $( \theta ( \varphi(U_j^\ast) - \varphi(U_{j-1}^\ast)))_{j \in \llbracket 1, n+1 \rrbracket}.$ Let $\alpha_0=0$ and $\alpha= (\alpha_0,\alpha_1,\dots, \alpha_{n+1}).$ One forms the vector denoted by $\Tilde{\alpha},$ corresponding to the cumulative sums of $\alpha.$  
\item \textit{Linear Interpolation.} 

One performs linear interpolation of the points $(U_i^\ast, \Tilde{\alpha}_i),$ for $i \in \llbracket0, n+1 \rrbracket.$ 
\end{enumerate} 
\end{minipage}}

\medskip

The resulting path, denoted by $w_n^{\BK}(\cdot;\theta)$, is, for $t \in [0,1)$,  \begin{equation}\label{expression_BK}
    w_n^{\BK}(t;\theta) =
        \displaystyle \sum_{j=0}^n \left( \Tilde{\alpha}_j + \alpha_{j+1} \frac{t - U_j^\ast}{U_{j+1}^\ast - U_j^\ast} \right) \mathds{1}_{[ U_j^\ast, U_{j+1}^\ast)}(t),
\end{equation} and $w_n^{\BK}(1;\theta) = 1.$

\subsection{Empirical cumulative distribution function with random weights (CDF Algorithm)}\label{sec:CDF}

The intuition behind this approach is quite simple. An empirical cumulative distribution function (CDF for short), being a CDF, could be considered as the realization of a warping process. We refine this idea in two ways: first, by using a linearized version of the standard empirical CDF; and second, by introducing random weights (again based on the Dirichlet distribution) to allow control over the variance of the resulting warping process.

For the first step, we build on the work of \citet{blanke2018polygonal}, who proposed a nonparametric continuous estimator of the CDF of a random variable. Their construction relies on a smoothing technique applied to the associated empirical cumulative distribution, named polygonal smoothing. It is performed as described hereafter. Let $\left( U_1,\dots, U_n\right)$ be an $n$-sample from the Uniform distribution on $[0,1],$ $G$ be an absolutely continuous cumulative distribution function from $[0,1]$ to $[0,1],$ and $p \in [0,1]$ be a real parameter. The vector $\left(U_0^\ast, U_1^\ast, \dots, U_{n+1}^\ast\right)$ denotes the sorted vector of $\{U_i\;;\; i \in \llbracket1,n\rrbracket \}$, with $U_0^\ast= 0$ and $U_{n+1}^\ast=1.$ We also set $X_i^\ast= G^{-1}(U_i^\ast),$ for $i \in \llbracket1,n \rrbracket.$ Interpolation knots are given by $(X_i^\ast, (i-p)/n),$ for $i \in \llbracket1,n \rrbracket$ and two extremes points $(0,0)$ and $(1,1).$ In particular, the corresponding curve intersects each constant step of the empirical cumulative distribution function at abscissas $ X_i^\ast + p  (X_{i+1}^\ast - X_i^\ast),$ for $i$ in $\llbracket1,n-1 \rrbracket.$ Thus, parameter $p$ acts as a localization parameter of the abscissa of each interpolation knot on the line segment $[X_i^\ast, X_{i+1}^\ast],$ for $i \text{ in } \llbracket1, n-1\rrbracket.$ 

We adapt this idea to our problem; this requires adapting the previous construction to the case where the jumps have different values and are random. We thus propose to perform linear interpolation of the constant steps of the path of the process defined as $t \in [0,1] \mapsto \sum_{i=1}^n \beta_i \mathds{1}_{[0, \varphi(t))}(U_i^\ast),$ where $(\beta_1,\dots, \beta_n)$ is a Dirichlet distributed random vector with parameter $(\theta/n, \dots, \theta/n)$, for a given concentration parameter $\theta \in \mathbb{R}_+^\ast,$ and where $\varphi$ is an absolutely continuous warping function. Note that for $p \in \{ 0,1\},$ the resulting path has a constant step and hence is not a warping function. Consequently, we assume that $p \in (0,1).$ We finally obtain the following algorithm, based on the linearized empirical cumulative distribution function with random weights, which we refer to as the CDF Algorithm.\\

\noindent\fbox{\begin{minipage}{.98\textwidth}
\paragraph{CDF Algorithm.} The algorithm takes four input parameters: $n,\theta, \varphi$ and $p.$ Parameters $n,\theta, \varphi$ have the same roles as in the BK Algorithm, and $p \in (0,1).$ The linear interpolation is performed as follows: \begin{enumerate}\vspace{5pt}
    \item \textit{Construction of the x-axis partition.} \newline  One generates an $n$-sample from Uniform distribution on $[0,1].$ Let $\left( U_1^\ast, \dots, U_n^\ast\right)$ be the sorted sample and set $U_0^\ast=0$ and $U_{n+1}^\ast=1.$ \newline One forms the vector $\varphi^{-1}(U^\ast)= \left( \varphi^{-1}(U_0^\ast), \dots, \varphi^{-1}(U_{n+1}^\ast)\right).$
    \item \textit{Construction of the y-axis partition.} \newline 
    One generates a random vector $(\beta_1, \dots, \beta_n)$ from a Dirichlet distribution with parameter $( \theta/n, \dots, \theta/n )$, set $\beta_{-1} = \beta_0 = \beta_{n+1}=0$, and forms the random vector $\gamma_p= (\gamma_{j,p})_{j \in \llbracket0,n+1\rrbracket}$ with $\gamma_{j,p} = (1-p)\beta_j + p \beta_{j-1}$ for $j \in \llbracket 0, n+1 \rrbracket.$ \newline One forms the vector $\Tilde{\gamma}_p=(\Tilde{\gamma}_{i,p})_{i \in \llbracket 0, n+1 \rrbracket}$ which corresponds to the cumulative sums of the vector $\gamma_p$ meaning that $\Tilde{\gamma}_{i,p} = \sum_{j=0}^i \gamma_{j,p}.$ 
    \item \textit{Linear interpolation.} \newline One performs linear interpolation of the points $\left( \varphi^{-1}(U_i^\ast), \Tilde{\gamma}_{i,p} \right), \text{ for } i \in \llbracket0, n+1 \rrbracket.$
\end{enumerate} 
\end{minipage}}

\medskip

The resulting path, denoted by $w_n^{\CDF}(\cdot; \theta,p)$, intersects the constant steps of the jump process $\sum_{i=1}^n \beta_i \mathds{1}_{[0, \varphi(\cdot))}(U_i^\ast)$ at the points with abscissa $\varphi^{-1}(U_j^\ast) + (p \beta_j / \gamma_{j+1,p})(\varphi^{-1}(U_{j+1}^\ast) - \varphi^{-1}(U_j^\ast))$, for $j \in \llbracket1,n-1\rrbracket$, and admits the explicit expression, for $t \in [0,1)$,  \[
    w_n^{\CDF}(t;\theta,p) =  \displaystyle\sum_{j=0}^n \left( \Tilde{\gamma}_{j,p} + \gamma_{j+1,p}  \frac{t - \varphi^{-1}(U_j^\ast)}{\varphi^{-1}(U_{j+1}^\ast) - \varphi^{-1}(U_j^\ast)}  \right) \mathds{1}_{[ \varphi^{-1}(U_j^\ast), \varphi^{-1}(U_{j+1}^\ast))}(t), \] and $w_n^{\CDF}(1;\theta,p) = 1.$

\subsection{Ma, Zhou and Wu's Algorithm (MZW Algorithm)} \label{sec:MZW}

In this subsection, we describe the algorithm introduced by \citet{ma2024stochastic}. The notation $\mathcal{D}([0,1])$ stands for the set of functions that admit a derivative at each point of the line segment $[0,1].$  Let $\Gamma_1$ be the subset of $\Gamma$ defined as \[\Gamma_1~=\{\gamma \in \Gamma \cap \mathcal{D}([0,1]) \;|\; \exists\;(m_\gamma, M_\gamma)\in (\mathbb{R}_+^\ast)^2, 0< m_\gamma< \gamma' <M_\gamma <\infty\}.\] The theoretical justification of this algorithm relies on two key points. The first one is the introduction of an inner product space structure on $\Gamma_1.$ Hence, with the binary operation $\oplus_{\Gamma_1},$ the scalar multiplication $\odot_{\Gamma_1}$ and the inner product $\langle \cdot , \cdot \rangle_{\Gamma_1}$ (whose definitions are recalled in Appendix~\ref{rappel_MZW}), $\Gamma_1$ becomes an inner product space. Let $H(0,1)$ be the subset of $L^2([0,1])$ defined as 
\[ H(0,1)= \{ h \in L^2([0,1]) \;|\; \exists\;(m_h,M_h) \in \mathbb{R}^2, -\infty<m_h < h<M_h< \infty, \int_0^1 h(t)\mathrm{d}t = 0 \}.\] 
The second key point is the existence of an isometric isomorphism between $\Gamma_1$ and $H(0,1),$ given by \[\psi^{-1}\;:\;h \in H(0,1) \mapsto \psi^{-1}(h) = \frac{\int_0^\cdot \exp(h(s))\mathrm{d}s}{ \int_0^1 \exp(h(\tau))\mathrm{d}\tau } \in \Gamma_1.\] 
These arguments lead to a two-step algorithm: first, simulate a stochastic process with paths in $H(0,1)$; second, apply the isometric isomorphism $\psi^{-1}$ to obtain a path of a warping process in $\Gamma_1$. Thus, a reformulation of the problem is to find a simulation method of a path in $H(0,1).$ The method in~\citet{ma2024stochastic} relies on the representation of the paths of a zero-mean second-order stochastic process in the smallest Hilbert space included in $L^2([0,1])$ and containing $H(0,1).$ More precisely, let $E(0,1)=\{ h \in L^2([0,1])| \int_0^1 h(t) \mathrm{d}t = 0 \},$ and $B=(\phi_i)_{i \geq 1}$ be the Fourier basis of $L^2([0,1])$ to which we remove the constant function equals to $1.$ Then, the space $E(0,1)$ is the smallest extension of $H(0,1)$ included in $L^2([0,1])$ with a Hilbert space structure, and $B$ is an orthonormal basis of $E(0,1).$ The following lemma provides a way to construct a covariance kernel in $E(0,1).$ 

\begin{lem}{\citep[Proposition 2.2]{ma2024stochastic}}\label{lemme_MZW}
For any non-negative sequence $(v_i)_{i\geq 1}$ such that $\sum_{i=1}^{\infty}v_i<\infty,$ let \begin{equation*}
    K(s,t) = \sum_{i=1}^{\infty} v_i \phi_i(s)\phi_i(t), \text{ for all } s, t \in [0,1].
\end{equation*} Then $K$ converges absolutely and uniformly and $K$ is a continuous, symmetric, non-negative definite function.
\end{lem}

Combining Lemma~\ref{lemme_MZW} and the Karhunen-Loève representation, we obtain a path of an element in $E(0,1).$ Indeed, let $c$ be a covariance kernel constructed as in Lemma~\ref{lemme_MZW} and let $X$ be a second-order stochastic process with mean function $\mu_X$ and covariance kernel $c.$ The Karhunen-Loève theorem (see for example \citep{wang2016functional}) states that $X$ admits the representation
\[X = \mu_X +  \sum_{i=1}^\infty \langle X-\mu_X, \phi_i\rangle \phi_i \] 
in the sense that 
\[ \underset{J \to \infty}{\lim}\; \underset{t \in [0,1]}{\sup} \mathbb{E}\left[ \left(X(t) - \mu_X(t) - \sum_{i=1}^J \langle X-\mu_X, \phi_i\rangle \phi_i(t) \right)^2\right] = 0 ,\] where the random variables $\langle X - \mu_X, \phi_i \rangle, $ for $i \geq 1,$ called scores, are non-correlated with zero mean and variance $v_i.$ In particular, for a second-order stochastic process $Y$ with mean zero and covariance kernel $c,$ we get an element $Y_m$ in $H(0,1)$ with a truncation at order $m \in \mathbb{N}^\ast $ as $Y_m = \sum_{i=1}^{m} \langle Y , \phi_i \rangle \phi_i$.
Since the expectation and the variance of the scores are known, this expression entails an easy simulation method of a path in $H(0,1)$, and we finally get the simulation algorithm, called the original version of the MZW Algorithm.

\paragraph{MZW Algorithm (original version).} The algorithm has two input parameters: a natural number $m \in \mathbb{N}^\ast,$ corresponding to the order of truncation of the process to be simulated and an $m$-tuple of positive coefficients denoted by $(v_i)_{i \in \llbracket1,m\rrbracket}.$ The algorithm proceeds as follows:
\begin{enumerate}
    \item \textit{Simulation of a path in $H(0,1).$ } \newline 
    One simulates an $m$-tuple of coefficients denoted by $(G_1,\dots, G_m),$ with $G_i$ a stochastic process with mean $0$ and variance $v_i$, and forms $t \in [0,1] \mapsto Y_m(t)= \sum_{i=1}^m G_i \phi_i(t).$ 
    \item \textit{Projection in $\Gamma_1.$} \newline One defines $t \in [0,1] \mapsto w_{m,0}^{\MZW}(t)= \psi^{-1}(Y_m)(t)$.
\end{enumerate} 
Simulations carried out in \citet{ma2024stochastic} seem to show that paths fluctuate around the identity map. Based on this observation, a naive approach to obtain simulations fluctuating around a targeted $\varphi$ function, consists in defining $t \in [0,1] \mapsto w_{m,0}^{\MZW}(t) - t + \varphi(t)$. This modification of the initial MZW Algorithm is not suitable since this mapping is not necessarily increasing, hence not a warping function. Thus, the only possible modification of the algorithm is to act on the mean function of the process whose paths are in $H(0,1).$ \rev{The} natural choice for this mean function is given by $\psi(\varphi).$ 
This leads to the following algorithm, that we still denote by MZW Algorithm for short in the following.\\

\noindent\fbox{\begin{minipage}{.98\textwidth}
\paragraph{MZW Algorithm (modified version).} The algorithm has four input parameters: a natural number $m \in \mathbb{N}^\ast,$ corresponding to the order of truncation of the process to be simulated, an $m$-tuple of positive real numbers  $(v_i)_{i \in \llbracket1,m\rrbracket},$ a concentration parameter $\theta \in \mathbb{R}_+^\ast$ and a warping function $\varphi \in \Gamma_1.$ The algorithm proceeds as follows: 
\begin{enumerate}
\item \textit{Simulation of one path in $H(0,1).$ } \newline 
One simulates an $m$-tuple of coefficients denoted by $(G_1,\dots, G_m),$ with $G_i$ a stochastic process with mean 0 and variance  $v_i / (1+\theta),$ and forms $t \in [0,1] \mapsto X_m(t;\theta) = \psi(\varphi)(t) + \sum_{i=1}^m G_i \phi_i(t).$ 
\item \textit{Projection on $\Gamma_1.$} \newline One defines $t \in [0,1] \mapsto w_m^{\MZW}(t;\theta) = \psi^{-1}(X_m(t;\theta)).$ 
\end{enumerate} 
\end{minipage}}

\medskip

The control of the variability in the BK and CDF Algorithms is achieved via a concentration parameter $\theta$. The proposed parameterization by $\theta$ of the variances of $G_i$ has the same flavour.

The function $\psi:\Gamma_1\rightarrow H(0,1)$ is defined by $\psi(\varphi)  = \log(\varphi') - \int_0^1 \log(\varphi')$, which leads to the warp path given, for $t\in[0,1]$, by 
\[w_m^{\MZW}(t;\theta) = \frac{\int_0^t \varphi'(s) \exp(\sum_{i=1}^m G_i \phi_i(s))\d s}{\int_0^1 \varphi'(\tau) \exp(\sum_{i=1}^m G_i \phi_i(\tau))\d\tau}.\] 
The simulated path in $H(0,1)$ separates the deterministic term $\psi(\varphi)$ from the random component $\sum_{i=1}^m G_i \phi_i$, whereas in the expression of $w_m^{\mathrm{MZW}}(\cdot;\theta)$, the terms $\varphi'$ and $\exp(\sum_{i=1}^m G_i \phi_i)$ are coupled, making the calculation of the associated expectation difficult. The nonlinearity of the isometric isomorphism involved in the algorithm is not surprising, as it reflects the geometry of the underlying spaces, particularly that of $\Gamma_1$, which is inherently complex. 
\rev{Although imposing $\psi(\varphi)$ as mean function in Step~1 of $\MZW$ Algorithm is a natural choice, Figure~\ref{fig_main_mvarie} presented in Section~\ref{section::Simulations} illustrates that it is not enough to get the pointwise convergence of $\mathbb{E}[w_m^\MZW(t;\theta)]$ to $\varphi(t),$ for $t \in [0,1].$}

\section{Properties of simulated warping functions}\label{section::properties} 

In this section, we present theoretical results to illustrate the significance of the algorithms introduced earlier. To this end, we introduce the subset $\Tilde{\Gamma}$ of $\Gamma$ defined as $\Tilde{\Gamma} = \{ \gamma \in \Gamma \;|\; \gamma \text{ is absolutely continuous}\}$, which is relevant to study the theoretical properties of BK and CDF Algorithms. 

We start with a study of the first two moments of $w_n^\BK(t; \theta)$ and $w_n^\CDF(t; \theta,p),$ for a given $t \in (0,1)$. We provide their explicit expressions, their limits along with the corresponding convergence rates as $n\to \infty$. For the BK Algorithm, Proposition~\ref{prop:warping} completes \citet[Proposition~1]{bharath2020distribution}.

\begin{proposition}\label{prop:warping}
For any $t \in (0,1),$ $p \in (0,1)$, $\theta \in \mathbb{R}_+^\ast$ and for any warping function $\varphi \in \Tilde{\Gamma},$ the following two statements hold:\\
(i) 
\begin{equation*}
\mathbb E[w_n^{\BK}(t ; \theta)] = \mathcal E_n^{\BK}(t)
\qquad \text{ and } \qquad
\mathbb E[w_n^{\CDF}(t; \theta, p)] = \mathcal E_n^{\CDF}(t; p)
\end{equation*}
where $\mathcal E_n^{\BK}(t)$ and $\mathcal E_n^{\CDF}(t; p)$ are explicitly given by \eqref{esp_explicite_BK} and \eqref{esp_explicite_CDF}. In particular, as $n\to \infty$, we have
\begin{equation*}
\mathbb E[w_n^{\BK}(t ; \theta)] -\varphi(t)=O\left( \frac{1}{n}\right)
\qquad \text{ and } \qquad 
\mathbb E[w_n^{\CDF}(t; \theta, p)] -\varphi(t)= O\left(\frac{1}{n}\right).
\end{equation*}
(ii) 
\begin{equation*}
\Var[w_n^{\BK}(t; \theta)] = \mathcal V_n^{\BK}(t; \theta) 
\qquad \text{ and } \qquad 
\Var[w_n^{\CDF}(t; \theta, p)] = \mathcal V_n^{\CDF}(t;\theta,p)
\end{equation*}
where $\mathcal V_n^{\BK}(t;\theta)$ and $\mathcal V_n^{\CDF}(t;\theta,p)$ are explicitly given by \eqref{var_explicite_BK} and \eqref{var_explicite_CDF}. In particular, as $n\to \infty$, we have
\begin{equation*}
\Var[w_n^{\BK}(t;\theta)] - \frac{1}{1 + \theta} \varphi(t) (1 - \varphi(t)) = O\left( \frac{1}{n} \right)
\end{equation*} and \begin{equation*}
\Var[w_n^{\CDF}(t; \theta, p)] - \frac{1}{1 + \theta} \varphi(t) (1 - \varphi(t)) = O\left( \frac{1}{n} \right).
\end{equation*}
\end{proposition}

The following comments ensue from Proposition~\ref{prop:warping}. First, by construction, the paths obtained from CDF and BK Algorithms are warping functions in $\Tilde{\Gamma}.$ Second, Statement~(i) highlights that the expectation term for each procedure is independent of the concentration parameter $\theta$ and tends, as expected, to $\varphi$ on $[0,1]$. Note that centered versions around $\varphi$ of the two warping paths can be derived from~\eqref{esp_explicite_BK} and~\eqref{esp_explicite_CDF}. Third, Statement~(ii) ensures the convergence of the variance term for each procedure to $\varphi(1-\varphi) / (1 + \theta)$ on $[0,1]$. Hence, the variances can be easily controlled through the parameter $\theta$. Fourth, we can notice that the parameter $p$, involved in the CDF Algorithm, has no influence on the limiting variance. This is not surprising as the length of each element of the form $[\varphi^{-1}(U_j^\ast), \varphi^{-1}(U_{j+1}^\ast)]$ gets smaller as $n \to \infty.$

The proof of Proposition~\ref{prop:warping} is postponed to Appendices~\ref{annexe_preuve}-\ref{appendix::proof_CDF}. The calculation of the associated expectation and variance terms is performed on each element of the partition $\{[U_j^\ast, U_{j+1}^\ast]\;;\; j\in \llbracket0, n\rrbracket \}$ of $[0,1]$ for the BK case and on each element of the partition $\{[\varphi^{-1}(U_j^\ast), \varphi^{-1}(U_{j+1}^\ast)]\;;\; j \in \llbracket0,n\rrbracket \}$ for the CDF case. Getting the explicit expressions relies on lengthy computations mainly based on properties of Dirichlet distribution of order statistics of uniform random vectors. The convergence results are essentially obtained by controlling terms of the form $\mathbb{E}[\sum_{j=0}^n  \Tilde{\alpha}_j^k \mathds{1}_{[ U_j^\ast, U_{j+1}^\ast)}(t)]$ for $k=1,2$ for the BK Algorithm and in terms of the form $\mathbb{E}[\sum_{j=0}^n  \Tilde{\gamma}_{j,p}^k \mathds{1}_{[ \varphi^{-1}(U_j^\ast), \varphi^{-1}(U_{j+1}^\ast))}(t)]$ for $k=1,2$ for the CDF Algorithm. In other words, all the information to prove the asymptotic results is contained in the jump process  $\sum_{j=0}^n \tilde{\alpha}_j \mathds{1}_{[U_j^\ast, U_{j+1}^\ast)}$ for the regularized process $w_n^\BK(\cdot;\theta)$ and in the jump process $\sum_{j=0}^n  \Tilde{\gamma}_{j,p} \mathds{1}_{[ \varphi^{-1}(U_j^\ast), \varphi^{-1}(U_{j+1}^\ast))}$ for the regularized process $w_n^\CDF(\cdot;\theta,p).$

As a consequence of Proposition~\ref{prop:warping}, we get the asymptotic behaviour of the $L^2$ risk function stated in Corollary~\ref{cv_risk_BK_CDF}.

\begin{corollary}\label{cv_risk_BK_CDF} For any $p \in (0,1),\;\theta \in \mathbb{R}_+^\ast$ and any warping function $\varphi \in \Tilde{\Gamma},$ we have as $n\to \infty$   
\[ \mathbb{E}\left[ \|w_n^{\BK}(\cdot;\theta) - \varphi \|^2 \right] 
\to
\frac{1}{1 + \theta} \int_0^1 \varphi(t)(1 - \varphi(t)) \mathrm{d}t \] 
and 
\[ \mathbb{E}\left[ \|w_n^{\CDF}(\cdot;\theta,p) - \varphi \|^2 \right] 
\to
\frac{1}{1 + \theta} \int_0^1 \varphi(t)(1 - \varphi(t)) \mathrm{d}t. \] 
\end{corollary} 

\rev{For the $\CDF$ case, Proposition~\ref{prop:CDF_comp_asymp_var_precise} enriches the result on the asymptotic behaviour of the variance term stated in Proposition~\ref{prop:warping}~(ii).}
\rev{
\begin{proposition} \label{prop:CDF_comp_asymp_var_precise}
For any $p \in (0,1),$ any concentration parameter $\theta \in \mathbb{R}_+^\ast,\; \varphi \in \tilde{\Gamma}$ and any $a \in (0, 1/2),$ we have, as $ n \xrightarrow{}\infty$ and $n\theta \xrightarrow{} \infty:$
\[
\underset{t \in (a, 1-a)}{\sup} \left| 
\Var\left[w_n^\CDF(t;\theta, p)\right] 
- \left( \frac{1}{n} + \frac{1}{1 + \theta} \right) \varphi(t)(1 - \varphi(t)) \right| = o\left( \frac{1}{n} \right) + O\left( \frac{1}{n\theta} \right). 
\]
\end{proposition}
}

\rev{
The proof of Proposition~\ref{prop:CDF_comp_asymp_var_precise} is postponed to Appendix~\ref{appendix::proof_CDF}. The term $\varphi(t)(1-\varphi(t))/n$ highlights the improvement in accuracy when studying the asymptotic behaviour of the variance term of $w_n^\CDF(\cdot;\theta,p)$ over $(a,1-a).$ Technical arguments based on the behaviour of $\varphi(t)^n$ or $(1 - \varphi(t))^n$ for large values of $n$ and $t \in (a, 1-a)$ are used in the proof and prevent from considering the case $a=0.$ Further comments, related to the problem of inference of the concentration parameter $\theta$, are presented in Subsection~\ref{ssec:inference} and emphasise the importance of this result.
}

Now, we focus on properties of paths obtained with MZW Algorithm. As mentioned in the previous section, the resulting paths are intrinsically linked to the geometry of the spaces $\Gamma_1$ and $L^2([0,1]).$ In particular, since $w_m^{\MZW}(\cdot;\theta)$ and $\varphi$ belongs to $\Gamma_1,$ it seems relevant to seek a control of the term $\mathbb{E}[\| w_m^{\MZW}(\cdot;\theta) \ominus_{\Gamma_1} \varphi \|_{\Gamma_1}^2].$ A simple calculation shows that, for all warping function $\Check{\varphi} \in \Gamma_1,$ the term $\mathbb{E}[\| w_m^{\MZW}(\cdot;\theta) \ominus_{\Gamma_1} \Check{\varphi} \|_{\Gamma_1}^2 ]$ admits the following decomposition \begin{equation}\label{decompo_biais_variance}
    \mathbb{E}[\| w_m^{\MZW}(\cdot;\theta) \ominus_{\Gamma_1} \Check{\varphi} \|_{\Gamma_1}^2 ] = \sum_{i=1}^m \Var[G_i] + \| \varphi \ominus_{\Gamma_1} \Check{\varphi} \|_{\Gamma_1}^2
\end{equation} 
whereby we deduce the next result.

\begin{proposition}\label{prop_MZW}
Let $m \in \mathbb{N}^\ast, \theta \in \mathbb{R}_+^\ast$ and $\varphi \in \Gamma_1.$ Let $(v_i)_{i \geq 1}$ be a sequence of positive real numbers satisfying the hypothesis of Lemma~\ref{lemme_MZW}. 
Then, the resulting process $w_m^{\MZW}(\cdot;\theta)$ has a unique Fréchet mean in $\Gamma_1$ for the distance associated to the inner product $\langle\cdot,\cdot\rangle_{\Gamma_1}$ given by $\varphi$ and the associated Fréchet variance is given by 
\[\mathbb{E}\left[ \| w_m^{\MZW}(\cdot;\;\theta) \ominus_{\Gamma_1} \varphi \|_{\Gamma_1}^2 \right] = \mathbb{E}\left[ \|X_m - \psi(\varphi)\|^2 \right] = \frac{1}{1 + \theta} \sum_{i=1}^m v_i. \]
\end{proposition}

Since $\sum_{i=1}^m v_i < \infty,$ we can notice that the Fréchet mean and Fréchet variance are analogous to the asymptotic expressions of expectation and variance of Proposition~\ref{prop:warping}. 
Thus, Proposition~\ref{prop_MZW} justifies the relevance of the MZW procedure for simulating $\varphi$ in the space $\Gamma_1$ and for controlling the associated Fréchet variance term, thanks to the parameter $\theta$. It also highlights that the $L^2$-norm $\mathbb{E}\left[ \|X_m - \psi(\varphi)\|^2 \right]$ can be controlled. However, it does not lead to a control of the $L^2$-norm $\mathbb{E}[\|w_m^{\MZW}(\cdot; \theta) - \varphi\|^2]$ as $m \to \infty$, as will be illustrated in Section~\ref{ssec:num_study}.

\section{Simulation and application to temperature data}\label{section::Simulations}

\subsection{Numerical study}\label{ssec:num_study}

This subsection provides a numerical comparison of the three algorithms presented in Section~\ref{section_description_algos}. We proceed in two steps. First, we illustrate the warping functions obtained with each algorithm for different target function $\varphi$. Second, we examine in more detail how the simulated warping functions vary in terms of empirical mean and variance under different algorithm parameters. All simulations follow a common framework. As mentioned in the description of MZW Algorithm, the way of simulating the coefficients $G_i$ depends on user's choice and classical examples are given in \citet{ma2024stochastic}. \rev{In the sequel, we assume that $(G_i)_{i \geq 1}$ is a sequence of independent Normal random variables with mean 0 and variance parameter $v_i$ given by $v_i = i^{-2}. $ }

\paragraph{Simulation setup and sample paths.} As a first illustration, we generate 30 paths from each algorithm for \rev{four} different warping functions \rev{$\varphi_i, i\in \llbracket 1,4 \rrbracket,$} where $ \varphi_1$ is the cumulative distribution function of a Uniform distribution on $[0,1],$ $\varphi_2$ is the cumulative distribution function of a $\text{Beta}(5,2)$ distribution, $\varphi_3$ is the warping function defined by $t \in [0,1] \mapsto (e^{-5t} -1)/(e^{-5}-1)$ \rev{and $\varphi_4$ is the warping function defined by $t \in  [0,1] \mapsto (1/4) t \mathds{1}_{[0, 3/4]}(t)+ ((13/4) t - 9/4) \mathds{1}_{(3/4, 1]}(t).$} In particular, $\varphi_i$ for $i = 1,3$ belong to $\Gamma_1$ while \rev{ $\varphi_i$ for $i = 2,4$ only belong to $\Tilde{\Gamma}.$} \rev{We recall that parameters $n$ and $m$ do not have the same role: while $n$ corresponds to the number of random interpolation knots in $\CDF$ and $\BK$ procedures, parameter $m$ corresponds to the truncation order in the decomposition based on Karhunen-Loève representation in $\MZW$ Algorithm. For this simulation, the parameters $n$, $m$, and the concentration parameter $\theta$ are kept fixed, and the parameter $p$ for the CDF algorithm is set to $p=0.5$.}  
The corresponding sample paths are presented in Figure~\ref{fig_main_plusieursphi}.

\rev{Note that for the CDF Algorithm, Proposition~\ref{prop:warping} shows that the parameter $p$ does not affect the asymptotic expectation and variance terms. We further investigate the influence of $p$ through a numerical study based on the empirical mean of the squared $L^2$-norm of the difference between $\varphi=\varphi_i$, $i\in\llbracket1,4\rrbracket$, and a path obtained using the $\CDF$ procedure, for $p= 0.25,0.5,0.75$. The results are reported in Table~\ref{tab:comp_L2_error_BK_CDF_p_phi_3} for $\varphi=\varphi_3$ and Tables~\ref{tab:comp_L2_error_BK_CDF_p_phi_1}-\ref{tab:comp_L2_error_BK_CDF_p_phi_4}, which are provided in Appendix~\ref{Appendix_autres_fig}, for $\varphi=\varphi_1,\varphi_2$ and $\varphi_4$ respectively. As a benchmark, the limiting value $(1+\theta)^{-1}\int_0^1 \varphi(t)(1-\varphi(t))\d t$ from Corollary \ref{cv_risk_BK_CDF} is also reported in the caption of each table. Since this study does not reveal a clear optimal choice of $p$, we follow the theoretical result of \citet{blanke2018polygonal}, which establishes that, when the weights $\omega_j$ are deterministic and equal to $1/n$, the $L^2$ error is minimized for $p=0.5$. We therefore set $p=0.5$ in all simulations involving the $\CDF$ Algorithm.}

\begin{table}[H]
\centering
\begin{tabular}{ccccccc}
\hline
&\multicolumn{3}{c}{$n=25$}&\multicolumn{3}{c}{$n = 200$}\\
 & 
$ \theta = 0.1$ & $\theta = 1$ & $\theta = 50$ & $ \theta = 0.1$ & $\theta = 1$ & $\theta = 50$ \\
\hline
  $\BK$ & 0.072  & 0.043 & 0.002 & 0.085 & 0.048 & 0.002\\
  $\CDF\; (p = 0.25)$ &  0.080 & 0.044 & 0.005 & 0.074 & 0.045 & 0.002  \\
  $\CDF\; (p = 0.50)$ & 0.079 & 0.041 & 0.005 & 0.090 & 0.047 & 0.002  \\
  $\CDF\; (p = 0.75)$ & 0.092 & 0.046 & 0.005 & 0.083 & 0.045 & 0.002  \\
 \hline
\end{tabular}
    \caption{\rev{Empirical $L^2$-risk based on $250$ paths associated to $ \varphi =  \varphi_3$ for different values of $n,\theta.$ For the $\CDF$ case, three values of $p$ are examined. In comparison, for $\theta = 0.1, 1, 50,$ the value of $\int_0^1 \varphi_3(1-\varphi_3)/(1+\theta)$ are respectively given by $0.086, 0.047$ and $ 0.002.$}}
    \label{tab:comp_L2_error_BK_CDF_p_phi_3}
\end{table}

\paragraph{Behaviour of the expectation.} Then, we study the behaviour of the empirical mean of the simulated processes for each algorithm. For the CDF and BK Algorithms, we illustrate Proposition~\ref{prop:warping}(i). For the MZW procedure, we have no equivalent result at our disposal since the metric $\| \cdot \|_{\Gamma_1}$ involved in Proposition~\ref{prop_MZW} is not appropriate to this study. This leads us to numerically study the behaviour of the expectation term of the MZW simulated process as $m \to \infty$. In the sequel, we let $m$ equal to $n.$ For each algorithm, we simulate $250$ paths with fixed parameter $\theta$ for  $n,m=5,25,200.$ For each algorithm, we represent the empirical mean of the simulated curves in black and the theoretical warping function $\varphi$ in red. 
Figure~\ref{fig_main_mvarie} reports the results for the theoretical warping function $\varphi=\varphi_3$. Illustrations for \rev{$\varphi = \varphi_i, i \in \{1,2,4\}$} can be found in Appendix~\ref{Appendix_autres_fig}. For the CDF and BK procedures, we can observe the convergence of the empirical mean to $\varphi$ as $n \to \infty,$ which is in agreement with Proposition~\ref{prop:warping}(i). In the MZW case, we can remark that the empirical mean curve can be quite far from the target function $\varphi$ even if $m$ is large.

\paragraph{Behaviour of the variance term.} We now illustrate the importance of the parameter $\theta$ to control the variability of the simulated warping functions. To this end, we simulate 250 paths with each procedure with fixed parameters \rev{$n,m=25$} for $\theta=0.1, 1,50$, again showing the empirical mean in black and the theoretical warping function $\varphi$ in red. Figure~\ref{fig_main_thetavarie} provides the results for $\varphi=\varphi_3$. Illustrations for \rev{$\varphi=\varphi_i,\;i \in \{1,2,4\}$} can be found in Appendix~\ref{Appendix_autres_fig}. For BK and CDF procedures, we remark that the variability decreases when the value of the parameter $\theta$ increases, which is in line with Proposition~\ref{prop:warping}(ii). We also notice the same behaviour with MZW Algorithm. This observation is not surprising. Indeed, the term $X_m$ is the sum of a deterministic component $\psi(\varphi)$ and of a random component $ \sum_{i=1}^m G_i \phi_i.$ When the value of $\theta$ increases, the variance parameter of each $G_i$ decreases and thus, $X_m$ behaves closer to the deterministic process $\psi(\varphi)$. A proof of this result in the case where $G_i$ for $i \in \llbracket 1,m\rrbracket,$ are independent Normal distributed random variables is given in Appendix~\ref{appendix::cv_moments_MZW}. 

\rev{\paragraph{A simulation related to the dataset.} 
We next focus on the curves obtained with the $\CDF$ Algorithm using fixed parameters chosen to be consistent with those involved in the real-data application presented in Section~\ref{subsection:application_dataset}. More precisely, we simulate $100$ paths using the $\CDF$ procedure, with $n=8760$, $\theta=350$, and $p=0.5$, for each of the four warping functions $\varphi=\varphi_i$, $i\in\llbracket1,4\rrbracket$. We also plot the empirical mean of the simulated warping functions in black and the theoretical warping function in red. The resulting figure is presented in Figure~\ref{fig:cdf_avec_param_boot}. For this value of $n$, the observed convergence of the empirical mean towards the theoretical warping function is consistent with the numerical study presented above. Furthermore, the strong concentration of the simulated curves for this value of $\theta$ is consistent with the behavior predicted by the analysis of the variance term.}

\paragraph{\rev{General comments} and recommendations.} The MZW Algorithm is easy to implement and fast. By construction, using the Karhunen-Loève expansion, this method is also highly flexible: (i) the user can adjust the variances of the random variables $G_i$ as well as the choice of basis functions; (ii) the regularity of the warping process can be easily controlled. When $m<\infty$, the regularity of sample paths are governed by the regularity by the one of $\psi(\varphi)$, whereas the two other algorithms produce only continuous \rev{piecewise linear} paths. \rev{As already mentioned, the MZW Algorithm has not been designed to satisfy the pointwise criteria adopted in our comparison. Indeed, our simulation illustrates that, for moderate values of $\theta$,  $\mathbb E[ w_m^\MZW(t;\theta)]-\varphi(t)$ is non negligible even when $m$ is large.}

Let us now compare the BK and CDF Algorithms. From a practical perspective, both are  straightforward to implement and very fast. The two methods are theoretically nearly identical: the rates of convergence for the expectation and variance terms are the same. \rev{Table~\ref{tab:comp_L2_error_BK_CDF_p_phi_3} and Tables~\ref{tab:comp_L2_error_BK_CDF_p_phi_1}-\ref{tab:comp_L2_error_BK_CDF_p_phi_4} evaluating Monte Carlo approximations (over 250 replications) of the $L^2$-errors of the simulated warping processes for these two algorithms, notably across all situations considered in Figures~\ref{fig_main_mvarie} and  \ref{Prop1(i)_varphi1}-\ref{Prop1(i)_varphi4}, do not reveal any clear difference between the BK and CDF Algorithms.} 

From a numerical point of view, the CDF Algorithm tends to be slightly simpler and more stable than the BK Algorithm. Indeed, the parameter for the Dirichlet distribution involved in the CDF Algorithm is constant and equal to $\theta/n$, whereas in the BK Algorithm it is equal to $\delta_j:=\theta(\varphi(U_j^*) - \varphi(U_{j-1}^*))$. This observation has two consequences. First, the same Dirichlet realizations could be reused for different target functions $\varphi$ in the CDF Algorithm, which is not possible with the BK Algorithm. \rev{In other words, computing one realization of $K$ target functions $\varphi$ takes approximately $K$ times longer with the BK Algorithm than with the CDF Algorithm.} Second, depending on the programming language and implementation procedure to generate Dirichlet realizations, when $\varphi$ contains flat regions, some $\delta_j$ may become extremely close to zero, potentially leading to numerical issues. \rev{ To illustrate this point, using Julia programming language, we perform $100$ runs with $\CDF$ and $\BK$ procedures for three warping functions corresponding to the cumulative distribution functions of  $\text{Beta}(2, 2),$  $\text{Beta}(2, 5)$ and  $\text{Beta}(2, 10)$ distributions. The simulations of the Dirichlet distribution are obtained with the Dirichlet function of the package Distributions.jl. Table~\ref{tab:success_rate} gives the associated success rate expressed as percentage of having an output for a run of the algorithm. While the $\CDF$ procedure leads to a success rate of $100\%$ for the three warping functions, success rate associated to the $\BK$ procedure are extremely low for the $\text{Beta}(2, 5), \text{Beta}(2, 10)$ distributions.}
We also believe that the CDF Algorithm is conceptually more natural. It is directly built on the idea of generating a continuous cumulative distribution function, which is, by definition, a warping function. This construction conveniently separates the control of the expected warping function, through the generation of variables of the form $\varphi^{-1}(U_j^*)$, and the control of the variance, which is achieved through random weights.

\begin{table}[h]
\centering
\begin{tabular}{ccc}
\hline
&\multicolumn{2}{c}{Success rate} \\
Distribution & $\CDF$ & $\BK$ \\
\hline
   $\text{Beta}(2, 2)$  & 100 & 100 \\
   $\text{Beta}(2,5)$  & 100 & 3\\
   $\text{Beta}(2, 10)$  & 100 & 0 \\
 \hline
\end{tabular}
    \caption{\rev{Success rates, expressed as percentages, for each algorithm and three warping functions corresponding to the empirical cumulative distribution functions of $\text{Beta}(2,2), \text{Beta}(2, 5)$ and $\text{Beta}(2, 10)$ distributions. Each success rate is based on $100$ runs of the algorithm under study, with $\theta = 350,$ and $n = 8760.$ The parameter $p$ is fixed to $0.5$ for the $\CDF$ procedure.}}
    \label{tab:success_rate}
\end{table}

\newpage

\begin{figure}[H]
\includegraphics[width=\textwidth]%scale=0.70]
{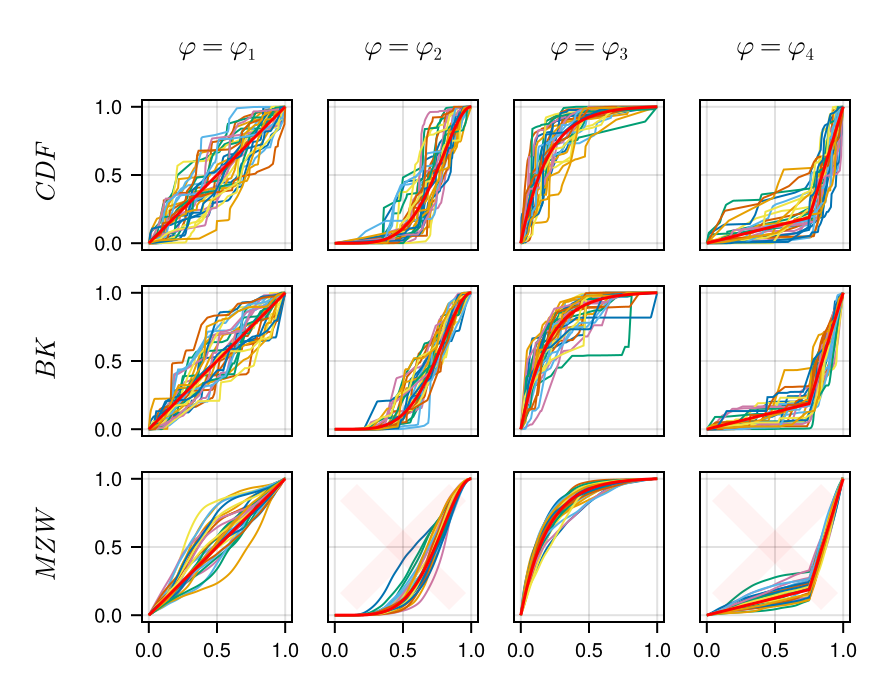}
\caption{\rev{Simulation of 30 paths with each algorithm for four different functions $\varphi :$ $\varphi_1, \varphi_2, \varphi_3$ and $\varphi_4.$ Here, $n,m~= 25,\;\theta = 1.$ In each case, the red curve corresponds to the associated theoretical function $\varphi.$ Crosses indicate cases where the target function $\varphi$ does not belong to $\Gamma_1.$}}
\label{fig_main_plusieursphi}
\end{figure}

\begin{figure}[H]
\includegraphics[width=\textwidth]%scale =0.80]
{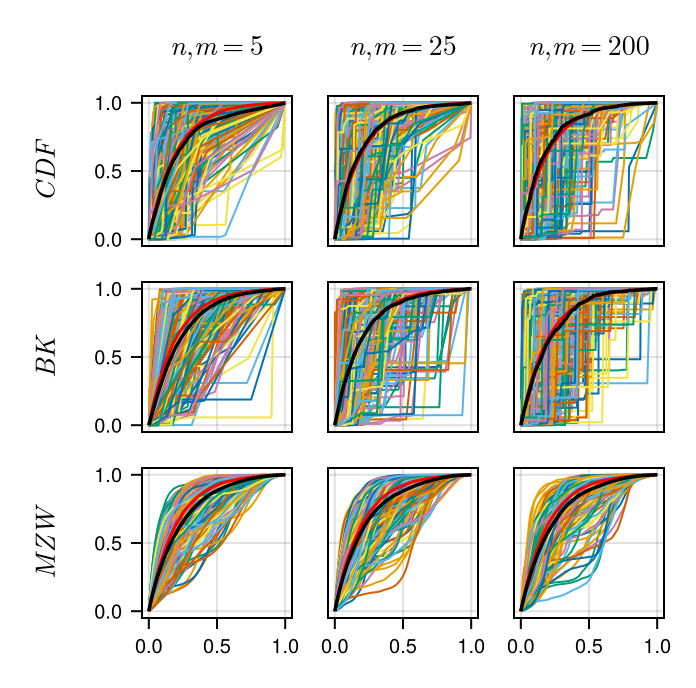}
\caption{Simulation of 250 paths with CDF and BK (respectively MZW) Algorithms for different values of $n$ (respectively $m$). Here we take $\varphi$ equals to $\varphi_3$ and parameter $\theta$ is fixed to $1.$} 
\label{fig_main_mvarie}
\end{figure}

\begin{figure}[H]
\includegraphics[width=\textwidth]%scale =0.80]
{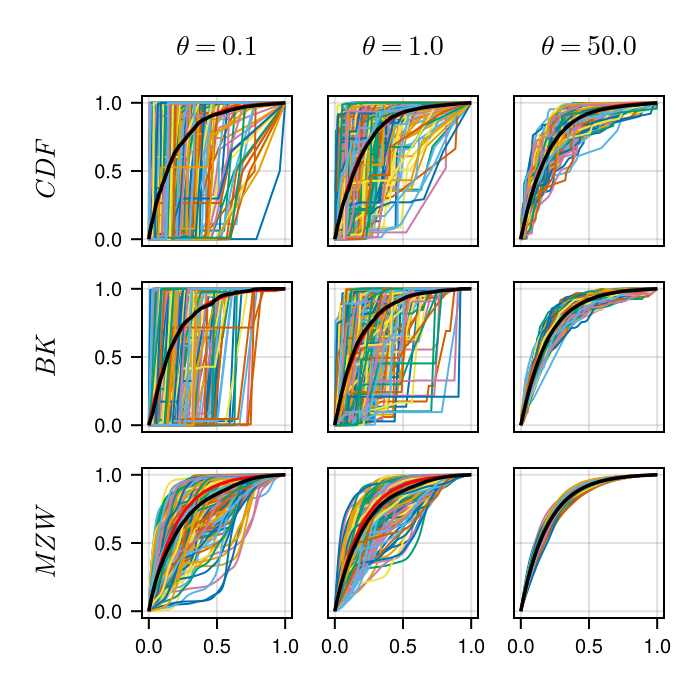}
\caption{Simulation of 250 paths with each algorithm for different values of $\theta.$ Here we take $\varphi$ equals to $\varphi_3.$ Parameters $n$ and $m$ are fixed to 25.}
\label{fig_main_thetavarie}
\end{figure}

\begin{figure}[H]
\centering
\includegraphics[width=\textwidth]{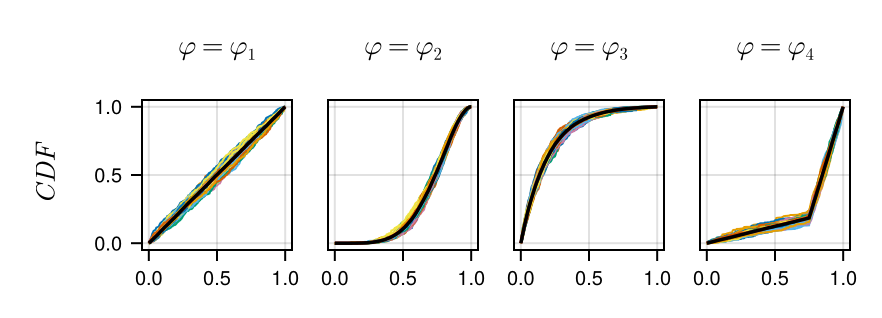}
\caption{\rev{Simulation of $100$ paths with CDF Algorithm for $\varphi = \varphi_i,\; i\in \llbracket1,4 \rrbracket.$ Parameter $n$ is fixed to $8760,\; \theta$ is fixed to $350 $ and $p$ is fixed to $0.5.$}} 
\label{fig:cdf_avec_param_boot}
\end{figure}

\subsection{Inference of the mean function and of the concentration parameter}\label{ssec:inference}

\rev{Assume that we dispose of an $m$-sample $(w_1,\dots,w_m)$ of warping functions from $\BK$ (respectively $\CDF)$ procedure with fixed value of $n$ (respectively $n$ and $p$). In other words, we consider $\BK$ and $\CDF$ procedures as a semi-parametric model governed by a warping function $\varphi \in \Tilde{\Gamma}$ and a concentration parameter $\theta \in \mathbb{R}_+^\ast$ and we focus on estimation procedures of $(\varphi, \theta)$. We apply these two estimators in Section~\ref{subsection:application_dataset} on temperature data in Canada. In particular, the estimate of the parameter $\theta$ is required to construct  confidence bands for $\varphi$ (using Bootstrap replications).}  

\rev{First, from Proposition~\ref{prop:warping}(i), for $n$ sufficiently large, a natural estimator of $\varphi$ is given by the pointwise empirical mean of the $m$-functions of the sample: 
\begin{equation*}
     \forall t \in [0,1],\;\hat{\varphi}(t) = \frac{1}{m} \sum_{j=1}^m w_j(t) .
\end{equation*}
In particular, $\hat{\varphi}$ is an asymptotically unbiased estimator of $\varphi,$ as $n \xrightarrow{} \infty.$ Second, let $a \in (0,1/2)$ and let $\hat v^c$ be the corrected version of the pointwise empirical variance of the $m$ curves of the sample. Proposition~\ref{prop:warping}(ii) naturally leads to the following estimator  for the concentration parameter $\theta:$
\begin{equation} \label{estimator_theta_1}
    \hat{\theta}^{(1)} = \left( \frac{ \int_{a}^{1-a} \hat{v}^c(t) \mathrm{d}t }{\int_{a}^{1-a} \hat{\varphi}(t) (1 - \hat{\varphi}(t)) \mathrm{d}t} \right)^{-1} - 1.    
\end{equation}}

\rev{Lemma~\ref{lem:term_in_precise_var} provides a refinement of the variance of $w^\CDF(t)$ which leads us, for the $\CDF$ procedure, to propose the following estimator:
\begin{equation} \label{estimator_theta_2}
\hat{\theta}^{(2)} = \left|\left( \frac{ \int_{a}^{1-a} \hat{v}^c(t) \mathrm{d}t }{\int_{a}^{1-a} \hat{\varphi}(t) (1 - \hat{\varphi}(t)) \mathrm{d}t} - \frac{1}{n}\right)^{-1} - 1\right|
\end{equation}
The absolute value is necessary because, depending on the parameters and the data, the expression inside may be negative.}

\rev{We investigate the properties of the estimators $\Hat \theta^{(1)}$ and $\hat \theta^{(2)}$ through a short simulation study. For each of several fixed values of $\theta,$ we generate $m = 30$ warping curves using the CDF Algorithm with $n=8760, p =0.5$ and $\varphi$ as the identity function, and compute $\Hat{\theta}^{(1)}$ and $\hat \theta^{(2)}$ with $a = 1/5.$ This procedure is repeated $B= 100$ times for each value of $\theta$. The values of $n,m$ and $\varphi$ are motivated by the application presented in Section~\ref{subsection:application_dataset}. Tables~\ref{tab:sim_study_theta_1}-\ref{tab:sim_study_theta_3} report the empirical means and standard deviations, denoted by $\tilde\theta^{\mathrm{MC}}$ and $\Tilde \sigma^{\mathrm{MC}}$, of the scaled parameter estimates $\hat \theta_b/\theta$, $b=1,\dots,B$ for the two different estimation procedures.}

\rev{\paragraph{Case $\mathbf{n \gg \theta.}$} Both tables indicate that for values of $(n,\theta)$ such that $n \gg \theta$, i.e. $n^{-1} \ll (1 + \theta)^{-1}$, the values of $\tilde{\theta}^\MC$ are close to $1$. This is consistent with Proposition~\ref{prop:warping}(i) which highlights that we do not expect any improvement considering $\hat{\theta}^{(2)}$ more than $\hat{\theta}^{(1)}$ for this setting. }

\rev{\paragraph{Case $\mathbf{n \simeq \theta.}$} For such values, the estimator $\hat{\theta}^{(2)}$ appears more relevant than estimator $\hat{\theta}^{(1)}.$ Indeed, Table~\ref{tab:sim_study_theta_1} shows values of $\tilde{\theta}^{\MC}$ close to $1/2$ for the estimator $\hat{\theta}^{(1)},$ while Table~\ref{tab:sim_study_theta_3} for estimator $\hat{\theta}^{(2)}$ shows values of $\tilde{\theta}^{\MC}$ close to $1$. This observation is in line with theoretical results. \rev{Indeed, for $n \simeq \theta$ we observe that the term $1/(1+\theta) + 1/n$ in Proposition~\ref{prop:CDF_comp_asymp_var_precise} is close to $2/(1+\theta)$ and explains the value $1/2$ obtained for the estimator $\hat{\theta}^{(1)}$ whose expression only involved a term $1/(1+\theta).$}}

\rev{\paragraph{Case $\mathbf{n \ll\theta.}$} For such values, Tables~\ref{tab:sim_study_theta_1}-\ref{tab:sim_study_theta_3} show values of $\tilde{\theta}^{\MC}$ far from $1.$ This is not surprising provided that the term $ 1 / (1+\theta)$, at the heart of the expressions of both estimators, is negligible compared to $1/n$. The two estimators of $\theta$ are definitely not appropriate in this setting.}

\begin{table}[H]
\centering
\begin{tabular}{ccccccccccc}
\hline
&\multicolumn{2}{c}{$\theta=5$}&\multicolumn{2}{c}{$\theta=350$}&\multicolumn{2}{c}{$\theta=500$}&\multicolumn{2}{c}{$\theta=2500$}&\multicolumn{2}{c}{$\theta=5000$} \\
$n$& 
$\Tilde{\theta}^{\mathrm{MC}}$& $\Tilde{\sigma}^{\mathrm{MC}}$ &
$\Tilde{\theta}^{\mathrm{MC}}$& $\Tilde{\sigma}^{\mathrm{MC}}$ &
$\Tilde{\theta}^{\mathrm{MC}}$& $\Tilde{\sigma}^{\mathrm{MC}}$ &
$\Tilde{\theta}^{\mathrm{MC}}$& $\Tilde{\sigma}^{\mathrm{MC}}$ &
$\Tilde{\theta}^{\mathrm{MC}}$& $\Tilde{\sigma}^{\mathrm{MC}}$ \\
\hline
  25 & 0.910 & 0.195 & 0.067 & 0.016 & 0.052 & 0.010 & 0.010 &  0.002 & 0.005 & 0.001 \\
  500 & 1.027 & 0.201 & 0.616 & 0.124 & 0.538 & 0.106 & 0.171 & 0.035 & 0.096 & 0.021 \\
  2500 & 1.003 & 0.214 & 0.919 & 0.179 & 0.876 & 0.157 & 0.503 & 0.093 & 0.339 & 0.058 \\
  5000 & 1.035 & 0.200 & 1.001 & 0.185 & 0.924 & 0.187 & 0.707 & 0.127 & 0.515 & 0.092 \\
  8760 & 1.010 & 0.195 & 1.004 & 0.183 & 0.961 & 0.178 & 0.807 & 0.153 & 0.673 & 0.135\\
  10000 & 1.034 & 0.177 & 1.023 & 0.204 & 0.971 & 0.192 & 0.848 & 0.143 & 0.687 & 0.140\\
 \hline
\end{tabular}
    \caption{\rev{Empirical means and standard deviations, denoted by $\tilde\theta^{\mathrm{MC}}$ and $\tilde \sigma^{\mathrm{MC}}$, of the scaled parameter estimates $\hat \theta_b/\theta$, $b=1,\dots,B$ for different values of $\theta$ with estimator $\hat{\theta}^{(1).}$}}
    \label{tab:sim_study_theta_1}
\end{table}

\begin{table}[H]
\centering
\begin{tabular}{ccccccccccc}
\hline
&\multicolumn{2}{c}{$\theta=5$}&\multicolumn{2}{c}{$\theta=350$}&\multicolumn{2}{c}{$\theta=500$}&\multicolumn{2}{c}{$\theta=2500$}&\multicolumn{2}{c}{$\theta=5000$} \\
$n$& 
$\Tilde{\theta}^{\mathrm{MC}}$& $\Tilde{\sigma}^{\mathrm{MC}}$ &
$\Tilde{\theta}^{\mathrm{MC}}$& $\Tilde{\sigma}^{\mathrm{MC}}$ &
$\Tilde{\theta}^{\mathrm{MC}}$& $\Tilde{\sigma}^{\mathrm{MC}}$ &
$\Tilde{\theta}^{\mathrm{MC}}$& $\Tilde{\sigma}^{\mathrm{MC}}$ &
$\Tilde{\theta}^{\mathrm{MC}}$& $\Tilde{\sigma}^{\mathrm{MC}}$ \\
\hline
  25 & 1.167 &  0.285 & 1.555 & 3.653 & 1.709 & 7.491 &  0.454 &  1.917 &  0.793 & 5.742 \\
  500 & 1.020 &  0.187 & 1.204 & 0.61 &  1.168 & 0.486 & 25.934 & 234.693 &  1.381 & 2.274 \\
  2500 & 1.061 & 0.207 & 1.043 & 0.243 & 1.081 & 0.237 & 1.136 & 0.447 & 35.551 & 339.563\\
  5000 & 1.020 & 0.187 & 1.030 & 0.198 & 1.050 & 0.212 & 1.078 & 0.330 & 1.251 & 0.611  \\
  8760 & 1.010 & 0.197 & 1.077 & 0.228 & 1.048 & 0.224 &  1.070 & 0.255 &  1.116 & 0.371 \\
  10000 & 0.992 & 0.188 & 1.038 & 0.229 & 1.032 & 0.194 & 1.080 & 0.280 & 1.118 & 0.429\\
 \hline
\end{tabular}
    \caption{\rev{Empirical means and standard deviations, denoted by $\tilde\theta^{\mathrm{MC}}$ and $\tilde \sigma^{\mathrm{MC}}$, of the scaled parameter estimates $\hat \theta_b/\theta$, $b=1,\dots,B$ for different values of $\theta$ with estimator $\hat{\theta}^{(2).}$} }
    \label{tab:sim_study_theta_3}
\end{table}

\begin{comment}
\begin{table}[H]
\centering
\begin{tabular}{ccccccccccc}
\hline
&\multicolumn{2}{c}{$\theta=5$}&\multicolumn{2}{c}{$\theta=350$}&\multicolumn{2}{c}{$\theta=500$}&\multicolumn{2}{c}{$\theta=2500$}&\multicolumn{2}{c}{$\theta=5000$} \\
$n$& 
$\Tilde{\theta}^{\mathrm{MC}}$& $\Tilde{\sigma}^{\mathrm{MC}}$ &
$\Tilde{\theta}^{\mathrm{MC}}$& $\Tilde{\sigma}^{\mathrm{MC}}$ &
$\Tilde{\theta}^{\mathrm{MC}}$& $\Tilde{\sigma}^{\mathrm{MC}}$ &
$\Tilde{\theta}^{\mathrm{MC}}$& $\Tilde{\sigma}^{\mathrm{MC}}$ &
$\Tilde{\theta}^{\mathrm{MC}}$& $\Tilde{\sigma}^{\mathrm{MC}}$ \\
\hline
  25 & 1.183 & 0.297 & 1.095 & 1.986 & 1.192 & 3.068 & 0.207 & 0.482 & 0.110 & 0.278 \\
  500 & 1.398 & 0.193 & 1.093 & 0.369 & 1.166 & 0.590 & 2.449 & 4.740 & 5.967 & 38.373 \\
  2500 & 1.017 & 0.188 & 1.064 &  0.239 & 1.134 & 0.325 & 1.160 & 0.543 & 1.840 & 2.989 \\
  5000 & 1.046 & 0.195 & 1.024 & 0.198 & 1.043 & 0.223 & 1.107 & 0.303 & 1.186 & 0.607  \\
  8760 & 1.054 & 0.212 & 1.009 & 0.208 & 1.035 & 0.186 & 1.056 & 0.268 & 1.133 & 0.371 \\
  10000 & 1.037 & 0.206 & 1.060 & 0.198 & 1.023 & 0.228 & 1.055 & 0.265 & 1.072 & 0.381\\
 \hline
\end{tabular}
    \caption{\rev{Empirical means and standard deviations, denoted by $\tilde\theta^{\mathrm{MC}}$ and $\tilde \sigma^{\mathrm{MC}}$, of the scaled parameter estimates $\hat \theta_b/\theta$, $b=1,\dots,B$ for different values of $\theta$ with estimator $\hat{\theta}^{(2).}$} }
    \label{tab:sim_study_theta_3}
\end{table}
\end{comment}

\subsection{Application to temperature data} \label{subsection:application_dataset}
\rev{
We illustrate this paper using hourly temperature data spatially averaged over Canada from 1951 to 2025, for a total of 657,456 observations. The data considered are extracted from the Copernicus Climate Change Service (C3S) Climate Data Store (CDS)~\citep{hersbach2023era5}. 
Following standard climatological practice, climate normals are typically defined over 30-year periods~\citep{arguez2011definition}. 
}\rev{Accordingly, our goal is to characterize the evolution of these spatially averaged temperatures over 30-year periods, shifted forward by ten years.  
The periods considered, along with the corresponding number of available observations $N_i$, are presented in Table~\ref{tab:donnees_par_periode}. Note that the last period contains only 25 years.}

\begin{table}[ht]
\centering
\begin{tabular}{ccc}\hline
 Number i of the period & Corresponding period & Number $N_i$ of data \\ \hline
0 & 1951-1980 & 262,992  \\ %\hline
1 & 1961-1990 & 262,968  \\ %\hline
2 & 1971-2000 & 262,992\\%\hline
3 & 1981-2010 & 262,968\\%\hline
4 & 1991-2020 & 262,992 \\%\hline
5 & 2001-2025 & 219,144 \\\hline
\end{tabular}
\caption{\rev{Periods considered and corresponding number of available data per period.}}
    \label{tab:donnees_par_periode}
\end{table}
\rev{
To study the evolution of temperatures, a standard approach consists in comparing summary statistics of annual data with reference values computed from temperature data collected over a reference period. These comparisons yield temperature anomaly series (e.g. \citep{jones1986global,li2024record}). Here, rather than limiting the analysis to summary statistics, we propose to characterize and compare the entire temperature distributions across periods, using 1951-1980 as the reference period.}
\rev{
Our approach is based on comparing the quantile function $F_i^{-1}$ of each period $i \in \llbracket 1, I \rrbracket,$ where $I=5$,  with the quantile function $F_0^{-1}$ of the reference period. More precisely, we consider a modification of the model introduced in \citet{gallon2013statistical}, under which, 
$$F_i^{-1} = F_0^{-1} \circ w_i
\quad  \text{ for } i \in \llbracket 1, I \rrbracket,
$$
where $w_i$ belongs to the set $\Gamma$ defined in Equation~\eqref{expression_Gamma_set}. If the temperature distribution remains unchanged over time, the corresponding warping function is simply the identity function. Our objective is therefore to assess the departure of each warping function $w_i$ from the identity on $[0,1]$.
}

\rev{
To do so, for each period $i \in \llbracket 1, I \rrbracket$, we split the data into $m$ samples, each consisting of the complete set of average temperature data for one year within the period. Thus, the first four periods are divided into $m = 30$ samples, while the last period is divided into $m = 25$ samples. For each period, this yields $m$ empirical quantile functions $F_{i1}^{-1},\dots,F_{im}^{-1}$ and, consequently, $m$ warping functions $w_{ij}=F_0 \circ F_{ij}^{-1},\; j \in \llbracket1,m\rrbracket$. From a computational perspective, the construction of the $m$ warping functions $w_{ij}$ relies on linear interpolation, ensuring that each $w_{ij}, j \in \llbracket 1, m \rrbracket$, belongs to $\tilde{\Gamma}.$ Figure~\ref{fig:quantile_data} displays the $m = 25$ empirical quantile functions for the period 2001-2025.}

\begin{figure}[H]
\centering
\includegraphics[width=\textwidth]{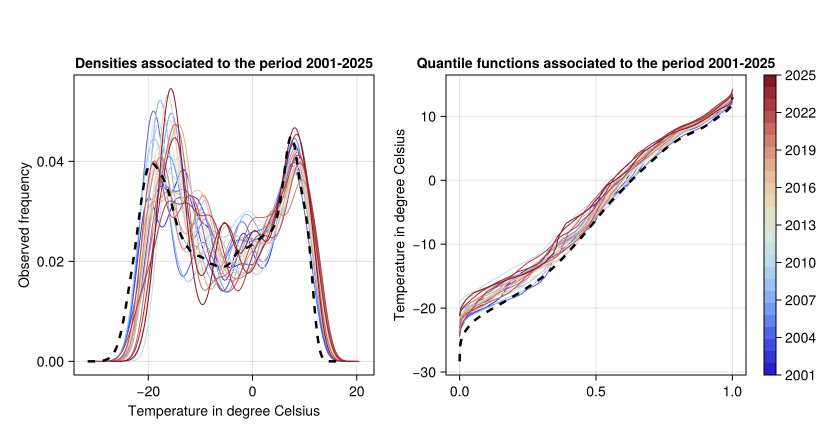}
\caption{\rev{Left: Empirical densities of the temperatures data for the $25$ samples associated to the 25 years of the period 2001-2025, obtained via kernel density estimation method. In addition, the empirical density of the reference period is represented by the black dashed curve. Right: Corresponding quantile functions curves. The black dashed curve corresponds to the quantile function of the reference period.}}
\label{fig:quantile_data}
\end{figure}

\rev{Next, we assume that the $m$ warping  functions are realizations of an $m$-sample from a warping process $w_i$ obtained with the CDF Algorithm (with  $p = 0.5$). We estimate the function $\varphi_i =\mathbb E [w_i]$ and the concentration parameter $\theta_i$ as follows
\begin{equation}\label{appli_estim}
    \Hat{\varphi_i}(t) = \frac{1}{m}\sum_{j=1}^m w_{ij}(t), \; \forall t \in [0,1] \quad \textrm{and} \quad \Hat{\theta}_i = \left| \left( \frac{ \int_{0.2}^{0.8} \hat{v}_i^c(t) \mathrm{d}t }{\int_{0.2}^{0.8} \hat{\varphi_i}(t) (1 - \hat{\varphi_i}(t)) \mathrm{d}t}  - \frac{1}{n_i} \right)^{-1} - 1\right|,
\end{equation}
where $\Hat{v}_i^c(t)$ is the pointwise unbiased sample variance of $w_{i1}(t),\ldots,w_{im}(t)$ and where we set $n_i=\lfloor N_i/m\rfloor$. We obtain values of $\hat\theta_i$ close to $350\ll 8760 \approx n_i$, which, from the numerical study presented in Section~\ref{ssec:inference}, ensures the efficiency of estimates.}

\rev{Then, we propose constructing confidence bands with coverage probability $1-\alpha = 95\%$ for the theoretical mean $\varphi_i$ of period $i$ using a parametric bootstrap resampling approach. For each period $i$, we proceed as follows. First, we generate $B = 100$ realizations $\tilde \varphi_{ib}$, each defined as the empirical mean of $m$ warping processes with parameters $p, n_i, \Hat{\varphi}_i, \Hat{\theta}_i$.  Although leap years have 24 more observations (i.e., approximately $0.3\%$ more) than other years, we use the same number $n_i = \lfloor N_i/m \rfloor$ within each period for the sake of simplicity. Second, we compute the half-width $h_i$ of the confidence bands as the $1-\alpha$ empirical quantile of the vector $(\| \Hat{\varphi}_i - \Tilde{\varphi}_{ib} \|_{\infty})_{b \in \llbracket 1, B \rrbracket}$. This choice is motivated by the straightforward interpretation of the supremum norm. A comparative study of the $L^2$-norm and the supremum norm for constructing confidence bands is presented in \citet{cuevas2006use}.}

\rev{Figure~\ref{fig:BC_bootstrap} presents the functions $p \mapsto \Hat{\varphi}_i(p) - p$ together with the corresponding confidence bands $\{(\Hat{\varphi}_i(p) - p)\pm h_i\;;\; p \in [0,1]\}$ for each period $i \in \llbracket1, I \rrbracket$. The confidence bands for the periods 1981-2010, 1991-2020, and 2001-2025 remain entirely above $0$, indicating that the corresponding quantile functions differ significantly from that of the reference period. Moreover, their position above $0$ indicates an upward shift of the quantile functions relative to the reference quantile function. This suggests that the hourly temperature distributions during these periods are shifted towards higher values compared with the reference distribution, with the magnitude of this shift increasing across the three periods.} 
\begin{figure}[H]
\centering
\includegraphics[width=\textwidth]%scale =0.85]
{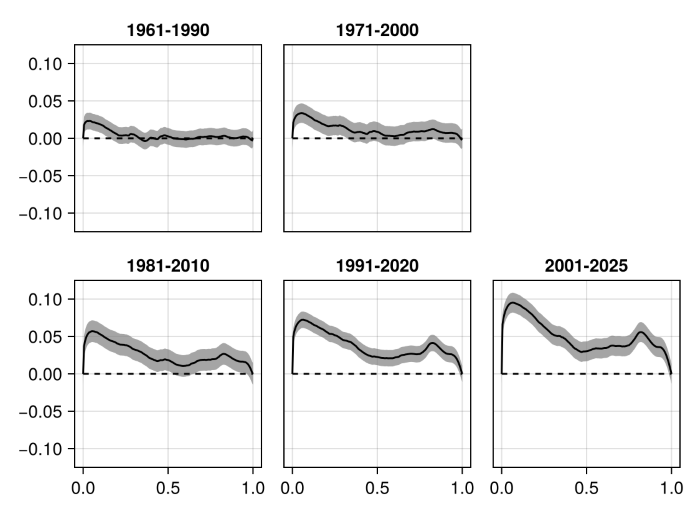}
\caption{\rev{Confidence bands obtained via bootstrap resampling approach for the $I$ periods. The black  curve represents the difference between the sample mean $\Hat{\varphi}_i$ and the identity function on $[0,1],$ for $i \in \llbracket1, I \rrbracket.$}}
\label{fig:BC_bootstrap}
\end{figure}

\section{Conclusion and perspectives}

\rev{In this paper, we compare simulation procedures for warping processes from the specific perspective of targeting a pointwise expectation while controlling pointwise variability. The BK and CDF algorithms achieve these objectives asymptotically as the number of interpolation knots increases. The MZW framework instead provides a geometric approach to stochastic modelling, and allows control of the Fréchet mean and variance in $\Gamma_1$. Its different behaviour under our pointwise criteria reflects these distinct objectives, rather than a general limitation of the method. BK and CDF procedures share similar asymptotic properties, with a more refined characterization of the asymptotic variance for the $\CDF$ Algorithm, leading to a more robust inference procedure for the concentration parameter $\theta$. Together with the slightly greater stability of the $\CDF$ procedure compared with $\BK$, these results support the use of the $\CDF$ Algorithm.}

\rev{While this work specifically focuses on one-dimensional simulation procedures of a given warping function, several directions for future work arise, particularly from the application to temperature data.
First, an extension to a smoothed version of the $\CDF$ procedure, using, for example, kernel methods or polynomial smoothing, could be considered. Second, a multivariate or continuous spatial extension of the $\CDF$ procedure could be of practical interest, for example to jointly model warping processes across different climatic regions of Canada. Third, we leave open the study of the asymptotic properties of the estimators of $\varphi$ and $\theta$, particularly in settings where the warping curves are not independent.}

\newpage

\begin{appendix} 
\section{Additional figures and tables}\label{Appendix_autres_fig}
\rev{
Figures~\ref{Prop1(i)_varphi1}-\ref{Prop1(ii)_varphi4} are similar to Figures~\ref{fig_main_mvarie}-\ref{fig_main_thetavarie} for functions $\varphi=\varphi_1,\varphi_2,\varphi_4.$
Tables~\ref{tab:comp_L2_error_BK_CDF_p_phi_1}-\ref{tab:comp_L2_error_BK_CDF_p_phi_4} are similar to Table~\ref{tab:comp_L2_error_BK_CDF_p_phi_3} for $\varphi=\varphi_1,\varphi_2,\varphi_4$.}

\begin{figure}[h]
\centering
\includegraphics[width=\textwidth]%scale=0.80]
{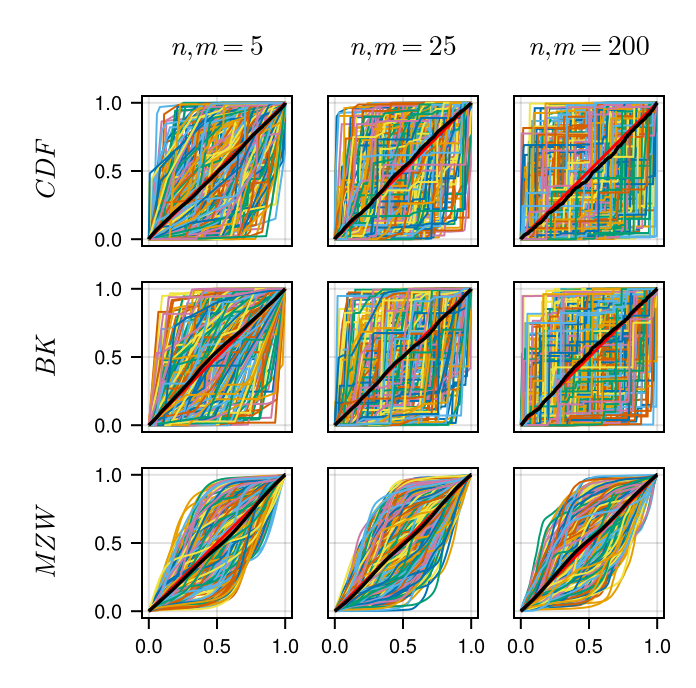}
\caption{Simulation of 250 paths with CDF and BK (respectively MZW) Algorithms for different values of $n$ (respectively $m$). Here $\varphi$ is equal to $\varphi_1$ and parameter $\theta$ is fixed to $1.$ For CDF Algorithm, parameter $p$ is fixed to $0.5.$}
\label{Prop1(i)_varphi1}
\end{figure}

\begin{figure}[H]
\centering
\includegraphics[width=\textwidth]%[scale = 0.80]
{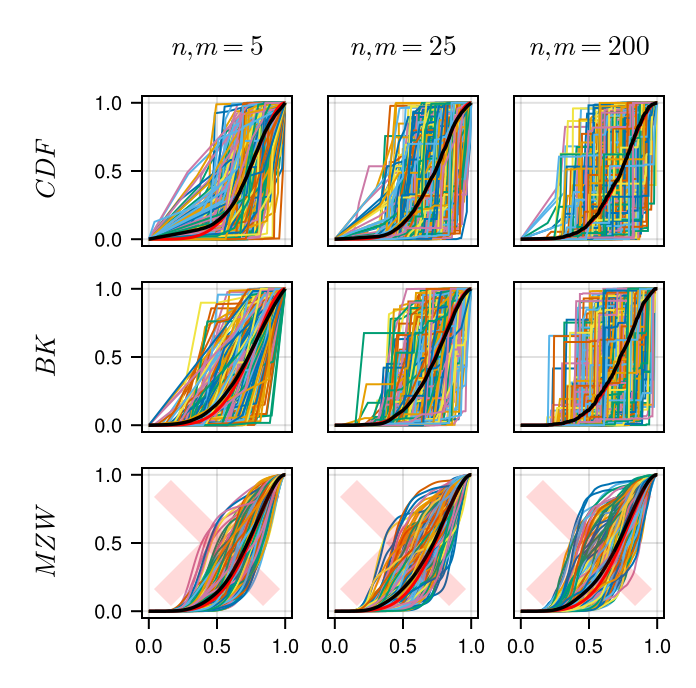}
\caption{Simulation of 250 paths with CDF and BK (respectively MZW) Algorithms for different values of $n$ (respectively $m$). Here $\varphi$ is equal to $\varphi_2$ and parameter $\theta$ is fixed to $1.$ For CDF Algorithm, parameter $p$ is fixed to $0.5.$  Crosses highlight that the target function $\varphi_2$ does not belong to $\Gamma_1.$} 
\label{Prop1(i)_varphi2}
\end{figure}

\begin{figure}[H]
\centering
\includegraphics[width=\textwidth]%[scale = 0.80]
{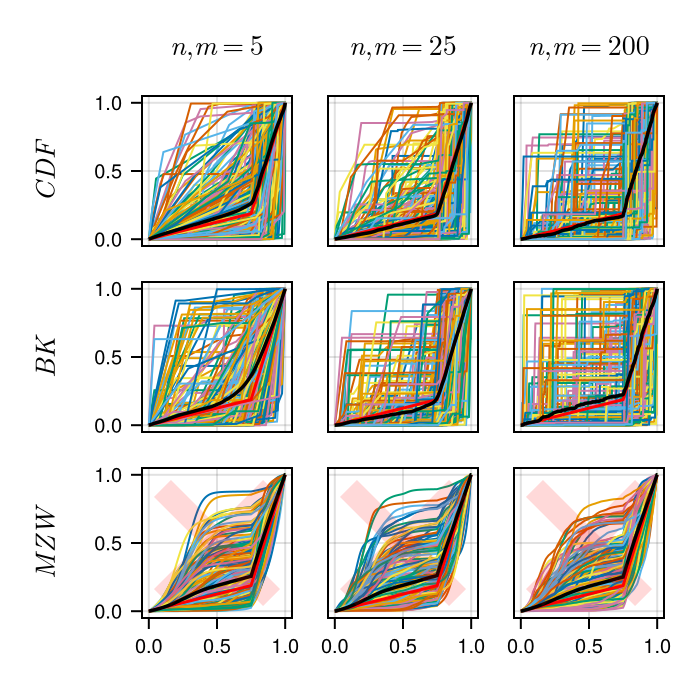}
\caption{Simulation of 250 paths with CDF and BK (respectively MZW) Algorithms for different values of $n$ (respectively $m$). Here $\varphi$ is equal to $\varphi_4$ and parameter $\theta$ is fixed to $1.$ For CDF Algorithm, parameter $p$ is fixed to $0.5.$ Crosses highlight that the target function $\varphi_4$ does not belong to $\Gamma_1.$} 
    \label{Prop1(i)_varphi4}
\end{figure}

\begin{figure}[H]
\centering
\includegraphics[width=\textwidth]%[scale = 0.80]
{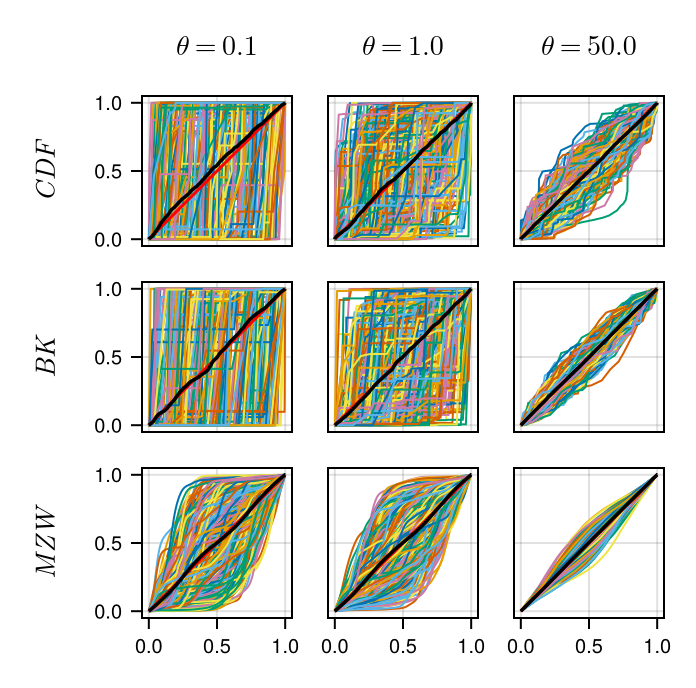}
\caption{Simulation of 250 paths with each algorithm for different values of $\theta.$ Here $\varphi$ is equal to $\varphi_1.$ Parameters $n,m$ are fixed to 25. For CDF Algorithm, parameter $p$ is fixed to $0.5.$}
\label{Prop1(ii)_varphi1}
\end{figure}

\begin{figure}[H]
\centering
\includegraphics[width=\textwidth]%[scale = 0.80]
{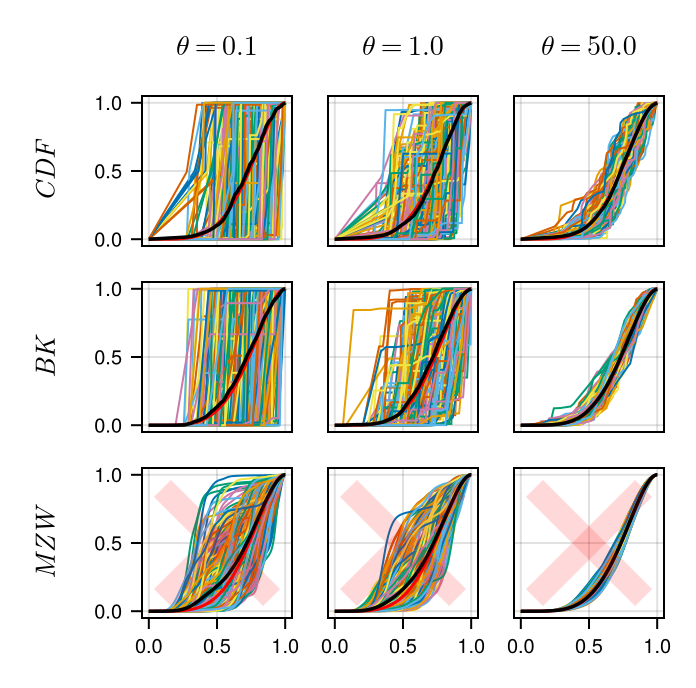}
\caption{Simulation of 250 paths with each algorithm for different values of $\theta.$ Here $\varphi$ is equal to $\varphi_2.$ Parameters $n,m$ are fixed to 25. For CDF Algorithm, parameter $p$ is fixed to $0.5.$ Crosses highlight that the target function $\varphi_2$ does not belong to $\Gamma_1.$}
\label{Prop1(ii)_varphi2}
\end{figure}

\begin{figure}[H]
\centering
\includegraphics[width=\textwidth]%[scale = 0.80]
{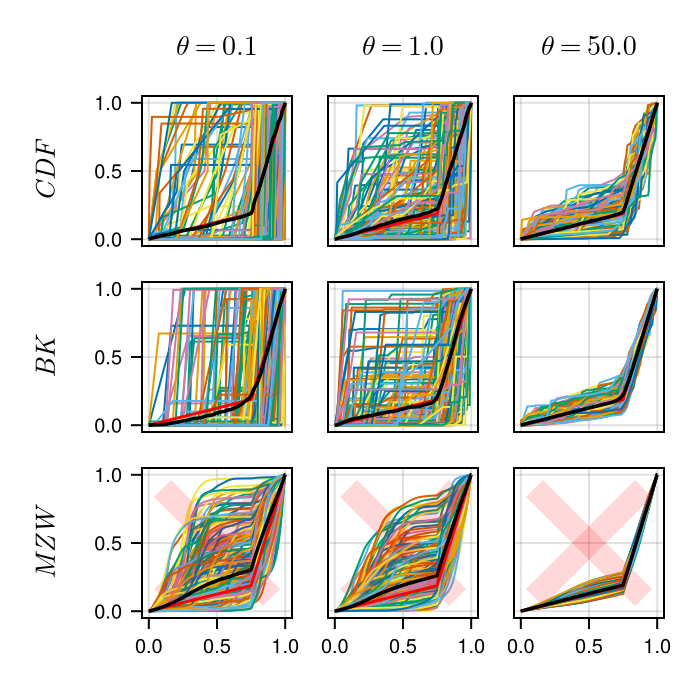}
\caption{Simulation of 250 paths with each algorithm for different values of $\theta.$ Here $\varphi$ is equal to $\varphi_4.$ Parameters $n,m$ are fixed to 15. For CDF Algorithm, parameter $p$ is fixed to $0.5.$ Crosses highlight that the target function $\varphi_4$ does not belong to $\Gamma_1.$}
\label{Prop1(ii)_varphi4}
\end{figure}

\begin{table}[H]
\centering
\begin{tabular}{ccccccc}
\hline
&\multicolumn{3}{c}{$n=25$}&\multicolumn{3}{c}{$n = 200$}\\
 & 
$ \theta = 0.1$ & $\theta = 1$ & $\theta = 50$ & $ \theta = 0.1$ & $\theta = 1$ & $\theta = 50$ \\
\hline
  $\BK$ & 0.137 &  0.083 & 0.003 & 0.142 & 0.079 & 0.003 \\
  $\CDF\; (p = 0.25)$ & 0.144 & 0.087 & 0.008 & 0.156 & 0.078 & 0.004 \\
  $\CDF\; (p = 0.50)$ & 0.144 & 0.086 & 0.009 & 0.150  & 0.084 & 0.004 \\
  $\CDF\; (p = 0.75)$ & 0.140  & 0.084 & 0.009 & 0.155 & 0.077 & 0.004 \\
 \hline
\end{tabular}
\caption{\rev{Empirical $L^2$-risk based on $250$ paths associated to $ \varphi =  \varphi_1$ for different values of $n,\theta.$ For the $\CDF$ case, three values of $p$ are examined. In comparison, for $\theta = 0.1, 1, 50,$ the value of $\int_0^1 \varphi_1(1-\varphi_1)/(1+\theta)$ are respectively given by $ 0.152, 0.083$ and $ 0.003.$}} 
\label{tab:comp_L2_error_BK_CDF_p_phi_1}
\end{table}

\begin{table}[H]
\centering
\begin{tabular}{ccccccc}
\hline
&\multicolumn{3}{c}{$n=25$}&\multicolumn{3}{c}{$n = 200$}\\
 & 
$ \theta = 0.1$ & $\theta = 1$ & $\theta = 50$ & $ \theta = 0.1$ & $\theta = 1$ & $\theta = 50$ \\
\hline
  $\BK$ &  0.076 & 0.038 & 0.002 & 0.079 & 0.047 & 0.002   \\
  $\CDF\; (p = 0.25)$ & 0.080 &  0.049 & 0.005 & 0.082 & 0.049 & 0.002 \\
  $\CDF\; (p = 0.50)$ & 0.077 & 0.040 &  0.005 & 0.080 & 0.047 & 0.002 \\
  $\CDF\; (p = 0.75)$ & 0.070  & 0.046 & 0.005 & 0.082 &  0.046 & 0.002 \\
 \hline
\end{tabular}
\caption{\rev{Empirical $L^2$-risk based on $250$ paths associated to $ \varphi =  \varphi_2$ for different values of $n,\theta.$ For the $\CDF$ case, three values of $p$ are examined. In comparison, for $\theta = 0.1, 1, 50,$ the value of $\int_0^1 \varphi_2(1-\varphi_2)/(1+\theta)$ are respectively given by $ 0.082, 0.045$ and $ 0.002.$}}
\label{tab:comp_L2_error_BK_CDF_p_phi_2}
\end{table}

\begin{table}[ht]
\centering
\begin{tabular}{ccccccc}
\hline
&\multicolumn{3}{c}{$n=25$}&\multicolumn{3}{c}{$n = 200$}\\
 & 
$ \theta = 0.1$ & $\theta = 1$ & $\theta = 50$ & $ \theta = 0.1$ & $\theta = 1$ & $\theta = 50$ \\
\hline
  $\BK$ & 0.072 & 0.047 & 0.002 & 0.104 & 0.064 & 0.002  \\
  $\CDF\; (p = 0.25)$ & 0.087 & 0.050 &  0.006 & 0.087 & 0.055 & 0.003 \\
  $\CDF\; (p = 0.50)$ & 0.085 & 0.054 & 0.006 & 0.093 & 0.050 &  0.002 \\
  $\CDF\; (p = 0.75)$ & 0.079 & 0.049 & 0.005 & 0.102 & 0.060 &  0.003 \\
 \hline
\end{tabular}
\caption{\rev{Empirical $L^2$-risk based on $250$ paths associated to $ \varphi =  \varphi_4$ for different values of $n,\theta.$ For the $\CDF$ case, three values of $p$ are examined. In comparison, for $\theta = 0.1, 1, 50,$ the value of $\int_0^1 \varphi_4(1-\varphi_4)/(1+\theta)$ are respectively given by $ 0.098, 0.054$ and $ 0.002.$}}
    \label{tab:comp_L2_error_BK_CDF_p_phi_4}
\end{table}

\section{Auxiliary lemmas for the proof of Proposition~\ref{prop:warping}} \label{annexe_preuve}
Lemmas~\ref{lem_exp_coeff_BK}-\ref{lem_exp_mom_phi_BK} (respectively Lemmas~\ref{lem_exp_beta_CDF}-\ref{lem:esp_ind_CDF}) provide useful results for the proof of Proposition~\ref{prop:warping} for the BK method (respectively for the CDF method). We remind the main notations. The vector $U^\ast = (U_0^\ast, \dots, U_{n+1}^\ast)$ is the sorted sample of $(U_0, \dots, U_{n+1})$ with $U_0~=0,\; U_{n+1} = 1,$ and $U_i,$ for $i \in \llbracket 1,n \rrbracket,$ are independent and identically distributed random variables from Uniform distribution on $[0,1].$ The function $\varphi \in \Tilde{\Gamma}$ is the cumulative distribution function that we want to simulate and we denote by $\varphi^{-1}$ its inverse function. Moreover, since $\varphi$ is an absolutely continuous function from $[0,1]$ to $[0,1],$ we deduce that $\varphi^2$ is also an absolutely continuous function and we denote by $\varphi'$ the derivative of $\varphi$ defined almost everywhere on $[0,1].$ The notation $\theta$ refers to a positive real number. The vector $(\alpha_1, \dots, \alpha_{n+1})$ follows a Dirichlet distribution with vector parameter given by $(\theta (\varphi(U_j^\ast) - \varphi(U_{j-1}^\ast)))_{j \in \llbracket1,n+1\rrbracket}$ and $\alpha_0$ is equals to $0.$ For all $j \in \llbracket0, n+1 \rrbracket,$ $\Tilde{\alpha}_j$ denotes the element $\sum_{i = 0}^j \alpha_i.$ The vector $(\beta_1,\dots, \beta_n)$ is a Dirichlet distributed random vector whose vector-parameter is given by $(\theta/n, \dots, \theta/n)$ and $\beta_{-1} = \beta_0 = \beta_{n+1}=0.$ The element $\Tilde{ \beta_j}$ refers to the $j$th component of the vector of cumulative sums of $(\beta_0, \dots, \beta_{n+1}),$ that is $\Tilde{\beta}_j = \sum_{i=0}^j \beta_i.$ The random vector $\gamma_p=(\gamma_{j,p})$ is defined as $\gamma_{j,p} = (1-p)\beta_j + p \beta_{j-1},$ for $j \in \llbracket 0, n+1 \rrbracket$ and $\Tilde{ \gamma}_p$ is the vector of cumulative sums of the vector $\gamma_p.$ The proof of Lemmas~\ref{lem_exp_coeff_BK}-\ref{lem:esp_ind_CDF} relies on standard properties of Dirichlet distribution that can be found for example in \citet[page 6]{Frigyik2010} and due to the fact that the expression of the densities of $U_1^\ast, U_{n}^\ast$ and of the pairs $(U_j^\ast, U_{j+1}^\ast),$ for $j \in \llbracket1,n-1\rrbracket,$ are known, as mentioned for example in \citet[page 63]{gentle2009computational}. In particular, denoting by $f_{(U_j^\ast, U_{j+1}^\ast)}$ the joint density of the pair $(U_j^\ast, U_{j+1}^\ast),$ for a fixed $j \in \llbracket1, n-1 \rrbracket, $ we have: 
\begin{equation}\label{densite_jointe_Ui}
    f_{(U_j^\ast, U_{j+1}^\ast)}(u,v) = n! \frac{u^{j-1}}{(j-1)!} \frac{(1-v)^{n-j-1}}{(n-j-1)!} \mathds{1}_{0 \leq u < v \leq 1}.
\end{equation}
Lastly, since $\varphi$ is differentiable almost everywhere and by denoting by $\dom(\varphi')$ for the domain of the map $\varphi',$ we introduce the following notation 
\begin{equation}
\| \varphi' \|_\infty = \underset{\dom(\varphi')}{\sup}\varphi'. \label{sup}
\end{equation}

\begin{lem} \label{lem_exp_coeff_BK}
Let $j \in \llbracket1, n+1 \rrbracket.$ Then, we have the following statements: 
\begin{align}
    & \mathbb{E}\left[ \alpha_j \;|\; U^\ast \right] = \varphi(U_j^\ast) - \varphi(U_{j-1}^\ast). \label{prop_dir1}\\
    & \Var[\alpha_j \; | \; U^\ast] = \frac{1}{1 + \theta} \left( \varphi(U_j^\ast) - \varphi(U_{j-1}^\ast) \right) - \frac{1}{1 + \theta} \left( \varphi(U_j^\ast) - \varphi(U_{j-1}^\ast) \right)^2. \label{prop_dir2} \\
    & \forall i \neq j, \; \Cov[\alpha_i, \alpha_j \; | \; U^\ast] = - \frac{1}{1 + \theta} \left( \varphi(U_i^\ast) - \varphi(U_{i-1}^\ast) \right) \left( \varphi(U_{j}^\ast) - \varphi(U_{j-1}^\ast) \right). \label{prop_dir3}\\
    & \mathbb{E}\left[ \Tilde{\alpha}_j \; | \; U^\ast \right] = \varphi(U_j^\ast). \label{esp_tilde_alpha} \\
    & \mathbb{E}\left[ \alpha_j^2 \; | \; U^\ast \right] = \frac{1}{1 + \theta} \left( \varphi(U_j^\ast) - \varphi(U_{j-1}^\ast) \right) + \left(1 - \frac{1}{1 + \theta} \right) \left( \varphi(U_{j}^\ast) - \varphi(U_{j-1}^\ast) \right)^2 .\label{mom_ord_2_alpha}\\ 
    & \mathbb{E}\left[ \Tilde{\alpha}_j^2 \; | \; U^\ast \right] = \frac{1}{1 + \theta} \varphi(U_j^\ast) + \left( 1 - \frac{1}{1 + \theta} \right) \varphi(U_j^\ast)^2.\label{mom_ord_2_tilde_alpha} \\
    & \forall j \neq n+1, \mathbb{E}\left[ \Tilde{\alpha}_j \alpha_{j+1} \;|\; U^\ast \right] = \left(1 - \frac{1}{1 + \theta} \right) \left( \varphi(U_{j+1}^\ast) - \varphi(U_j^\ast) \right) \varphi(U_j^\ast) . \label{cross_esp}
\end{align}
\end{lem}

\begin{proof}
Equations~\eqref{prop_dir1}-\eqref{prop_dir3} follow from standard properties of the Dirichlet distribution. Equation~\eqref{esp_tilde_alpha} follows from Equation~\eqref{prop_dir1}, from a telescopic property and from the equation $\varphi(0) = 0.$ Equation~\eqref{mom_ord_2_alpha} results from Equations~\eqref{prop_dir1}-\eqref{prop_dir2} and the conditional variance definition: 
\begin{equation*}
\mathbb{E}[ \alpha_j^2 \; | \; U^\ast] = \Var[\alpha_j\;|\; U^\ast] + \mathbb{E}[ \alpha_j \;|\; U^\ast]^2. \end{equation*}
Now, we prove Equation~\eqref{mom_ord_2_tilde_alpha}. By definition of $\Tilde{\alpha}_j,$ we have \begin{equation}
    \mathbb{E}\left[ \Tilde{\alpha}_j ^2 \; | \; U^\ast \right]  = \sum_{i=0}^j \sum_{k=0}^j \mathbb{E}\left[ \alpha_i \alpha_k \;|\; U^\ast\right].
\end{equation}
Then, using the expression of conditional variance and covariance, we get 
\begin{align*}
    \mathbb{E}\left[ \alpha_j^2 \; | \; U^\ast \right] &= \sum_{i=0}^j \mathbb{E}\left[ \alpha_i^2 \;|\; U^\ast \right] + \sum_{i=0}^j \sum_{\substack{k=0 \\ k \neq i}}^j \mathbb{E}\left[ \alpha_i \alpha_k \; | \; U^\ast \right] \\
    &= \sum_{i=0}^j \left( \Var[\alpha_i \; | \; U^\ast] + \mathbb{E}\left[ \alpha_i \;|\; U^\ast \right]^2 \right) \\
    &\quad + \sum_{i=0}^j \sum_{\substack{k=0 \\ k \neq i}}^j \left( \Cov[\alpha_i, \alpha_k \; | \; U^\ast ] + \mathbb{E}\left[ \alpha_i \; | \; U^\ast \right] \mathbb{E}\left[ \alpha_k \;|\; U^\ast \right] \right).
\end{align*} 
Using Equations~\eqref{prop_dir1}-\eqref{prop_dir3}, we derive the following equations 
\begin{align*} 
    \mathbb{E}\left[ \Tilde{\alpha}_j ^2 \; | \; U^\ast \right] &= \frac{1}{1+ \theta} \sum_{i=1}^j \left( \varphi(U_i^\ast) - \varphi(U_{i-1} ^\ast) \right) + \left(1 - \frac{1}{1 + \theta} \right) \sum_{i=1}^j \left( \varphi(U_i^\ast) - \varphi(U_{i-1}^\ast) \right)^2 \\
    & \quad \quad + \left(1 - \frac{1}{1 + \theta} \right) \sum_{i=1}^j \sum_{\substack{k=1 \\ k \neq i}}^j \left( \varphi(U_k^\ast) - \varphi(U_{k-1}^\ast) \right) \left( \varphi(U_i^\ast) - \varphi(U_{i-1}^\ast) \right)  \\
    &= \frac{1}{1 + \theta}  \sum_{i=1}^j \left( \varphi(U_i^\ast) - \varphi(U_{i-1}^\ast) \right)\\
    &\quad + \left( 1-\frac{1}{1 + \theta} \right) \sum_{i=1}^j \sum_{k=1}^j \left( \varphi(U_k^\ast) - \varphi(U_{k-1}^\ast) \right) \left( \varphi(U_i^\ast) - \varphi(U_{i-1}^\ast) \right) \\
    &= \frac{1}{1 + \theta} \varphi(U_j^\ast) + \left( 1 -  \frac{1}{1 + \theta} \right) \varphi(U_j^\ast)^2.
\end{align*}
Equation~\eqref{cross_esp} results from the following remark and from the previous results: 
\begin{align*}
    \mathbb{E}\left[ \Tilde{\alpha}_j \alpha_{j+1} \;|\; U^\ast   \right] &= \sum_{i=0}^j \mathbb{E}\left[ \alpha_i \alpha_{j+1} \; | \; U^\ast \right] \\
    &= \sum_{i=0}^j \left( \Cov[\alpha_i, \alpha_{j+1} \; | \; U^\ast]  + \mathbb{E}\left[ \alpha_i \; |\; U^\ast \right] \mathbb{E}\left[ \alpha_{j+1} \;|\; U^\ast \right] \right).
\end{align*}
\end{proof}

\begin{lem}\label{lem_exp_mom_phi_BK}
Let $t \in (0,1).$ We obtain the following equations. 
\begin{align}
    & \sum_{j=1}^{n-1} \mathbb{E}\left[ \mathds{1}_{[U_j^\ast, U_{j+1}^\ast)}(t) \right] = 1 - t^n - (1-t)^n \\
    &\sum_{j=1}^{n-1} \mathbb{E}\left[ \varphi(U_j^\ast) \mathds{1}_{[U_j^\ast, U_{j+1}^\ast)}(t) \right] =  \varphi(t) - \varphi(t) t^n - \int_0^t \varphi'(u) \left( (1-t+u)^n -u^n \right) \mathrm{d}u. \label{esp_phi}\\ 
    \begin{split}&\sum_{j=1}^{n-1} \mathbb{E}\left[ \varphi(U_{j+1}^\ast) \mathds{1}_{[U_j^\ast, U_{j+1}^\ast)}(t) \right] \\
    &\quad= -t^n + \varphi(t) - (1-t)^n\varphi(t) + \int_t^1 \varphi'(v) \left( (1-v+t)^n - (1-v)^n \right) \mathrm{d}v \end{split}\label{Newton2}
\end{align}
\begin{align}
    &\sum_{j=1}^{n-1} \mathbb{E}\left[ \varphi(U_j^\ast)^2 \mathds{1}_{[U_j^\ast, U_{j+1}^\ast)}(t) \right] = \varphi(t)^2 - \varphi(t)^2t^n  - 2 \int_0^t \varphi'(u) \varphi(u) \left( (1-t+u)^n - u^n \right) \mathrm{d}u \label{Newton3}\\
    \begin{split}&\sum_{j=1}^{n-1} \mathbb{E}\left[ \varphi(U_{j+1}^\ast)^2 \mathds{1}_{[U_j^\ast, U_{j+1}^\ast)}(t) \right]\\
    &\quad = -t^n + \varphi(t)^2 - (1-t)^n \varphi(t)^2 + 2 \int_t^1 \varphi'(v) \varphi(v) \left( (1+t-v)^n - (1-v)^n \right) \mathrm{d}v .\end{split} \label{Newton4}
\end{align}
\end{lem}

\begin{proof}
From Equation~\eqref{densite_jointe_Ui} that reminds the expression of the joint density of the pairs $(U_j^\ast, U_{j+1}^\ast),$ for $j \in \llbracket1, n-1 \rrbracket,$ we derive an explicit expression of $\sum_{j=1}^{n-1 }\mathbb{E}[ \mathds{1}_{[U_j^\ast, U_{j+1}^\ast)}(t)] $ as follows: 
\begin{align*}
    \sum_{j=1}^{n-1 }\mathbb{E}\left[ \mathds{1}_{[U_j^\ast, U_{j+1}^\ast)}(t) \right] &= \sum_{j=1}^{n-1} \int_0^t \int_t^1 n! \frac{u^{j-1}}{(j-1)!} \frac{(1-v)^{n-j-1}}{(n-j-1)!} \mathrm{d}v \mathrm{d}u \\
    &= \sum_{j=1}^{n-1} n! \frac{t^j}{j!} \frac{(1-t)^{n-j}}{(n-j)!} \\
    &= 1 - t^n - (1-t)^n.
\end{align*}
Now, we prove Equation~\eqref{esp_phi}.
We have 
\begin{align*}
    \sum_{j=1}^{n-1} \mathbb{E}\left[ \varphi(U_j^\ast) \mathds{1}_{[U_j^\ast, U_{j+1}^\ast)}(t) \right] &= \sum_{j=1}^{n-1} \int_t^1 \int_0^t \varphi(u) \frac{n!}{(j-1)!(n-j-1)!} u^{j-1} (1-v)^{n-j-1} \mathrm{d}u \mathrm{d}v \\
    &= \sum_{j=1}^{n-1} \int_0^t \varphi(u) n! \frac{u^{j-1}}{(j-1)!} \frac{(1-t)^{n-j}}{(n-j)!} \mathrm{d}u  \\
    &= \int_0^t \varphi(u) n \sum_{j=1}^{n-1} \binom{n-1}{j-1} u^{j-1} (1-t)^{n-1-(j-1)} \mathrm{d}u 
\end{align*} 
Applying the binomial theorem leads to the hereunder equation. 
\begin{align*}
  \sum_{j=1}^{n-1} \mathbb{E}\left[ \varphi(U_j^\ast) \mathds{1}_{[U_j^\ast, U_{j+1}^\ast)}(t) \right]    &= \int_0^t \varphi(u) n \left( (1-t + u)^{n-1} - u^{n-1} \right) \mathrm{d}u. 
\end{align*} 
Then, by integrating by parts, and by using the property $\varphi(0) = 0,$ we finally get  
\begin{align*} 
    &\sum_{j=1}^{n-1} \mathbb{E}\left[ \varphi(U_j^\ast) \mathds{1}_{[U_j^\ast, U_{j+1}^\ast)}(t) \right]\\ &= \left[ \varphi(u) \left( (1-t+u)^n - u^n \right) \right]_0^t - \int_0^t \varphi'(u) \left( (1-t+u)^n -u^n \right) \mathrm{d}u\\ 
    &= \varphi(t) - \varphi(t) t^n - \int_0^t \varphi'(u) \left( (1-t+u)^n -u^n \right) \mathrm{d}u,
\end{align*} 
Proofs of Equations~\eqref{Newton2}-\eqref{Newton4} rely on similar arguments.  We write the details below. 

\noindent We have 
\begin{align*}
&\sum_{j=1}^{n-1} \mathbb{E}\left[ \varphi(U_{j+1}^\ast) \mathds{1}_{[U_j^\ast, U_{j+1}^\ast)}(t) \right]\\ &= \sum_{j=1}^{n-1} \int_t^1 \int_0^t \varphi(v) \frac{n!}{(j-1)!(n-j-1)!} u^{j-1} (1-v)^{n-j-1} \mathrm{d}u \mathrm{d}v \\
&= \sum_{j=1}^{n-1} \int_t^1 \varphi(v) \frac{(1-v)^{n-j-1}}{(n-j-1)!}\frac{1}{(j-1)!} n! \frac{t^j}{j} \mathrm{d}v \\
&= \int_t^1 \varphi(v) n \sum_{j=1}^{n-1} \binom{n-1}{j} t^j (1-v)^{n-1-j} \mathrm{d}v \\
&= \int_t^1 \varphi(v) n \left( (1-v+t)^{n-1} - (1-v)^{n-1} \right) \mathrm{d}v \\
&= \left[ - \varphi(v) \left( (1+t-v)^n - (1-v)^n \right) \right]_t^1 + \int_t^1 \varphi'(v) \left( (1-v+t)^n - (1-v)^n \right) \mathrm{d}v 
\end{align*} 
The equation $\varphi(1) = 1$ leads to    
\begin{align*}
&\sum_{j=1}^{n-1} \mathbb{E}\left[ \varphi(U_{j+1}^\ast) \mathds{1}_{[U_j^\ast, U_{j+1}^\ast)}(t) \right]\\
&= -t^n + \varphi(t) (1 - (1-t)^n) + \int_t^1 \varphi'(v) \left( (1-v+t)^n - (1-v)^n \right) \mathrm{d}v \\
&= -t^n + \varphi(t) - (1-t)^n\varphi(t) + \int_t^1 \varphi'(v) \left( (1-v+t)^n - (1-v)^n \right) \mathrm{d}v.
\end{align*} 
We now focus on Equation~\eqref{Newton3}. 
\begin{align*}
&\sum_{j=1}^{n-1} \mathbb{E}\left[ \varphi(U_j^\ast)^2 \mathds{1}_{[U_j^\ast, U_{j+1}^\ast)}(t) \right] \\
&= \sum_{j=1}^{n-1} \int_0^1 \int_0^1 \varphi(u)^2 \mathds{1}_{[u,v)}(t) n! \frac{u^{j-1}}{(j-1)!} \frac{(1-v)^{n-j-1}}{(n-j-1)!} \mathrm{d}v \mathrm{d}u \\
&= \sum_{j=1}^{n-1} \int_0^t \int_t^1 \varphi(u)^2 n! \frac{u^{j-1}}{(j-1)!} \frac{(1-v)^{n-j-1}}{(n-j-1)!} \mathrm{d}v \mathrm{d}u \\
&= \sum_{j=1}^{n-1} \int_0^t \varphi(u)^2 n! \frac{u^{j-1}}{(j-1)! (n-j-1)!} \frac{(1-t)^{n-j}}{(n-j)} \mathrm{d}u \\
&= \int_0^t \varphi(u)^2 n \sum_{j=1}^{n-1} \binom{n-1}{j-1} u^{j-1} (1-t)^{n-1-(j-1)} \mathrm{d}u \\
&= \int_0^t \varphi(u)^2 n \sum_{k=0}^{n-2} \binom{n-1}{k} u^{k} (1-t)^{n-1-k} \mathrm{d}u \\
&= \int_0^t \varphi(u)^2 n \left( (1-t+u)^{n-1} - u^{n-1} \right) \mathrm{d}u  \\
&= \left[ \varphi(u)^2 \left( (1-t+u)^n - u^n \right) \right]_0^t - 2 \int_0^t \varphi'(u) \varphi(u) \left( (1-t+u)^n - u^n \right) \mathrm{d}u  \\
&= \varphi(t)^2 (1-t^n) - 2 \int_0^t \varphi'(u) \varphi(u) \left( (1-t+u)^n - u^n \right) \mathrm{d}u.
\end{align*} 
Regarding Equation~\eqref{Newton4}. 
\begin{align*}
&\sum_{j=1}^{n-1} \mathbb{E}\left[ \varphi(U_{j+1}^\ast)^2 \mathds{1}_{[U_j^\ast, U_{j+1}^\ast)}(t) \right] \\
&= \sum_{j=1}^{n-1} \int_0^1 \int_0^1 \varphi(v)^2 n! \frac{u^{j-1}}{(j-1)!} \frac{(1-v)^{n-j-1}}{(n-j-1)! } \mathds{1}_{[u,v)}(t) \mathrm{d}u \mathrm{d}v \\
&= \sum_{j=1}^{n-1} \int_t^1 \varphi(v)^2 n! \frac{(1-v)^{n-j-1}}{(j-1)!(n-j-1)! } \int_0^t u^{j-1}\mathrm{d}u \mathrm{d}v \\
&= \sum_{j=1}^{n-1} \int_t^1 \varphi(v)^2 n! \frac{(1-v)^{n-j-1}}{(j-1)! (n-j-1)!}\frac{t^j}{j} \mathrm{d}v \\
&= \int_t^1 \varphi(v)^2 n \sum_{j=1}^{n-1} \binom{n-1}{j} t^j (1-v)^{n-1-j} \mathrm{d}v \\
&= \int_t^1 \varphi(v)^2 n \left( (1+t-v)^{n-1} - (1-v)^{n-1} \right) \mathrm{d}v \\
&= \left[ \varphi(v)^2 \left( (1-v)^n - (1+t-v)^n \right) \right]_t^1 - 2 \int_t^1 \varphi'(v) \varphi(v) \left( (1-v)^{n} - (1+t-v)^n \right) \mathrm{d}v \\
&= -t^n + \varphi(t)^2 (1 - (1-t)^n) + 2 \int_t^1 \varphi'(v) \varphi(v) \left( (1+t-v)^n - (1-v)^n \right) \mathrm{d}v.
\end{align*}
\end{proof} 
Now we state auxiliary lemmas for the proof of Proposition~\ref{prop:warping} in the CDF case.

\begin{lem}\label{lem_exp_beta_CDF}
Let $j \in \llbracket 1,n \rrbracket.$ Then, we have the following statements: \begin{align}
    &\mathbb{E}\left[ \beta_j \right] = \frac{1}{n} \;; \label{esp_beta}\\
    & \Var[\beta_j] = \frac{1}{n} \frac{1}{1 + \theta } - \frac{1}{n^2} \frac{1}{1 + \theta} \label{var_beta}\;;\\
    & \forall j \neq k, \; \Cov[ \beta_j, \beta_k] = - \frac{1}{n^2} \frac{1}{1 + \theta}\;; \label{cov_beta}\\
    & \mathbb{E}\left[ \beta_j^2\right]  = \frac{1}{n} \frac{1}{1 + \theta} + \frac{1}{n^2}\left( 1 - \frac{1}{1 + \theta} \right) \;;\label{mom_ord_2_beta}\\
    & \forall k \neq j, \; \mathbb{E}\left[ \beta_j \beta_k \right] = \frac{1}{n^2} \left(1 - \frac{1}{1 + \theta} \right)\;; \label{esp_cross_beta}\\
    & \mathbb{E}\left[ \beta_{j} \Tilde{\beta}_{j} \right] = \frac{1}{n} \frac{1}{1 + \theta} + \frac{j}{n^2} \left( 1 - \frac{1}{1 + \theta} \right) \;;\label{esp_beta_tildebeta}\\
    & \mathbb{E}\left[ \beta_{n+1} \Tilde{\beta}_{n+1} \right] = 0\;; \label{esp_beta_tildebeta_n+1} \\
    & \mathbb{E}\left[ \Tilde{\beta}_j ^2 \right] =  \frac{j}{n} \frac{1}{1 + \theta}  + \frac{j^2}{n^2 } \left( 1 - \frac{1}{1 + \theta} \right) \;; \label{mom_ord_2_tilde_beta}\\
    & \mathbb{E}\left[ \Tilde{\beta}_{n+1}^2 \right] = 1. \label{mom_ord_2_tilde_beta_n+1}
\end{align}
\end{lem}

\begin{proof}
Equations~\eqref{esp_beta}-\eqref{cov_beta} follow from standard properties of the Dirichlet distribution. Equations~\eqref{mom_ord_2_beta}-\eqref{esp_cross_beta} are direct consequences of the variance and covariance formulas. Now, we show Equation~\eqref{esp_beta_tildebeta}. The definition of $\Tilde{\beta}_j$ leads to \begin{align*}
    \mathbb{E}\left[ \beta_j \Tilde{\beta}_j \right] &= \sum_{i=1}^j \mathbb{E}\left[ \beta_j \beta_i \right] = \mathbb{E}\left[ \beta_j^2 \right] + \sum_{i=1}^{j-1} \mathbb{E}\left[ \beta_j \beta_i \right]. \end{align*}  Then, using Equations~\eqref{mom_ord_2_beta}-\eqref{esp_cross_beta}, we get the following expression: \begin{align*}\mathbb{E}\left[ \beta_j \Tilde{\beta}_j \right]
    &= \frac{1}{n} \frac{1}{1 + \theta} + \frac{1}{n^2} \left( 1 - \frac{1}{1 + \theta} \right) + \frac{j-1}{n^2} \left( 1 - \frac{1}{1 + \theta} \right) \\
    &= \frac{1}{n} \frac{1}{1 + \theta} + \frac{j}{n^2} \left( 1 - \frac{1}{1 + \theta} \right).
\end{align*} Equation~\eqref{esp_beta_tildebeta_n+1} ensues from the fact that $\beta_{n+1}=0.$ Consider Equation~\eqref{mom_ord_2_tilde_beta}. Term $\mathbb{E}[ \tilde{\beta}_j ^2 ]$ can be expressed in the following manner: 
\begin{align*}
    \mathbb{E}\left[ \Tilde{\beta}_j^2 \right] &= \sum_{i=1}^j \sum_{k=1}^j \mathbb{E}\left[ \beta_i \beta_k \right] = \sum_{i=1}^j \mathbb{E}\left[ \beta_i^2 \right] + \sum_{i=1}^j \sum_{\substack{k = 1 \\ k \neq i}}^j \mathbb{E}\left[ \beta_i \beta_k \right]. \end{align*} Using Equations~\eqref{mom_ord_2_beta}-\eqref{esp_cross_beta}, we derive the following formula \begin{align*}
    \mathbb{E}\left[ \Tilde{\beta}_j^2 \right]&= j \left(\frac{1}{n} \frac{1}{1 + \theta} + \frac{1}{n^2}\left( 1 - \frac{1}{1 + \theta} \right) \right) + j(j-1) \frac{1}{n^2} \left(1 - \frac{1}{1 + \theta} \right) \\
    &= \frac{j}{n} \frac{1}{1 + \theta} + \frac{j^2}{n^2} \left(1 - \frac{1}{1 + \theta} \right).
\end{align*} 
Equation~\eqref{mom_ord_2_tilde_beta_n+1} results from $ \sum_{i=1}^n\beta_i = 1$ and $\beta_{n+1}=0.$
\end{proof}

\begin{lem}\label{lem_exp_mom_gamma_CDF}
Let $j \in \llbracket 1,n \rrbracket.$ Then, we have the following statements: 
\begin{align}
    & \mathbb{E}\left[ \gamma_{1,p} \right] = \frac{1-p}{n} \label{esp_gamma_1} \;;\\
    &\forall j \neq 1, \mathbb{E}\left[ \gamma_{j,p} \right] = \frac{1}{n} \label{esp_gamma_j}\; ;\\
    & \mathbb{E}\left[ \gamma_{n+1,p} \right] = \frac{p}{n} \label{esp_gamma_n+1}\;;\\
    & \mathbb{E}\left[ \gamma_{1,p}^2 \right] = (1-p)^2 \left[ \frac{1}{n} \frac{1}{1 + \theta} + \frac{1}{n^2} \left(1 - \frac{1}{1+\theta} \right) \right] \label{mom_ord_2_gamma_1}\;; \\
    & \mathbb{E}\left[ \gamma_{n+1,p}^2 \right] = p^2 \left[ \frac{1}{n} \frac{1}{1 + \theta} + \frac{1}{n^2} \left(1 - \frac{1}{1 + \theta} \right) \right] \label{mom_ord_2_gamma_n+1}\;;\\
    & \forall j \neq 1, \mathbb{E}\left[ \gamma_{j,p}^2  \right] = \frac{1}{n} \frac{1}{1 + \theta} \left( 1 -2p + 2 p^2 \right) + \frac{1}{n^2} \left(1 - \frac{1}{1 + \theta} \right) \label{mom_ord_2_gamma_j}\;; \\
    & \mathbb{E}\left[ \Tilde{\gamma}_{1,p} \right] = \frac{1-p}{n} \label{esp_tilde_gamma_1}\;;\\
    & \mathbb{E}\left[ \Tilde{\gamma}_{j,p} \right] = \frac{j-p}{n} \label{esp_tilde_gamma_j}\;;\\
    & \mathbb{E}\left[ \Tilde{\gamma}_{n+1,p} \right] = 1 \label{esp_tilde_gamma_n+1}\;;\\
    & \mathbb{E}\left[ \Tilde{\gamma}_{1,p}^2 \right] = (1-p)^2 \left[ \frac{1}{n} \frac{1}{1 + \theta} + \frac{1}{n^2} \left(1 - \frac{1}{1+\theta} \right) \right] \label{mom_ord_2_tilde_gamma_1}\;;\\
    & \mathbb{E}\left[ \Tilde{\gamma}_{j,p}^2 \right] = (j-2p + p^2) \frac{1}{n} \frac{1}{1 + \theta} + (j-p)^2 \frac{1}{n^2} \left(1 - \frac{1}{1 + \theta} \right)\label{mom_ord_2_tilde_gamma_j}\;;\\ 
    & \mathbb{E}\left[ \Tilde{\gamma}_{n+1,p}^2 \right] = 1 \label{mom_ord_2_tilde_gamma_n+1}\;;\\
    & \mathbb{E}\left[ \gamma_{j+1,p} \Tilde{\gamma}_{j,p} \right] = \frac{1}{n} \frac{1}{1 + \theta} (p-p^2) + \frac{1}{n^2} \left( 1 - \frac{1}{1 + \theta} \right) (j-p) \label{esp_cross_gamma_tilde_gamma}\;;\\
    & \mathbb{E}\left[ \gamma_{n+1, p} \Tilde{\gamma}_{n,p} \right] = p \frac{1}{n} - p^2 \frac{1}{n} \frac{1}{1 + \theta} - p^2 \frac{1}{n^2} \left(1 - \frac{1}{1 + \theta} \right). \label{esp_cross_gamma_tilde_gamma_n}
\end{align}
\end{lem}

\begin{proof}
We have $ \Tilde{\gamma}_{1,p} = (1-p) \beta_1$ which leads to 
\[
    \mathbb{E}\left[ \Tilde{\gamma}_{1,p} \right] = (1-p) \mathbb{E}[\beta_1] = \frac{1-p}{n}.
\]
For $j \in \llbracket 2, n \rrbracket,\; \gamma_{j,p} = (1-p)\beta_j + p \beta_{j-1}.$ Thus, we have  
\[
\mathbb{E}\left[ \gamma_{j,p} \right] = (1-p)\mathbb{E}[\beta_j] + p \mathbb{E}\left[\beta_{j-1}\right] = \frac{1-p}{n} + \frac{p}{n} = \frac{1}{n}. \] 
The definition of $\gamma_{n+1,p}$ leads to  
\[\mathbb{E}\left[ \gamma_{n+1,p} \right] = \mathbb{E}\left[ p \beta_n \right] = \frac{p}{n}.\] 
Using the expression of $\gamma_{n+1,p}$ and Equation~\eqref{mom_ord_2_beta}, we get 
\[ \mathbb{E}\left[ \gamma_{1,p}^2 \right] = \mathbb{E}\left[ (1-p)^2 \beta_1^2 \right] = (1-p)^2 \left[ \frac{1}{n} \frac{1}{1 + \theta} + \frac{1}{n^2} \left( 1 - \frac{1}{1 + \theta} \right) \right]. \] 
In the same way, we have 
\[\mathbb{E}\left[ \gamma_{n+1,p}^2 \right] = \mathbb{E}\left[ p^2 \beta_n^2\right] = p^2 \left[ \frac{1}{n} \frac{1}{1 + \theta} + \frac{1}{n^2} \left( 1 - \frac{1}{1 + \theta} \right) \right].\] 
For all $j \in \llbracket 1, n \rrbracket,$ we have 
\begin{align*}    \mathbb{E}\left[ \gamma_{j,p}^2 \right] &= \mathbb{E}\left[ \left( (1-p) \beta_j + p \beta_{j-1} \right)^2  \right] \\
&= (1-p)^2 \mathbb{E}\left[ \beta_j^2 \right] + 2p(1-p) \mathbb{E}\left[ \beta_j \beta_{j-1} \right] + p^2 \mathbb{E}\left[ \beta_{j-1}^2 \right] \\
&= \left( (1-p)^2 + p^2 \right) \left[ \frac{1}{n} \frac{1}{1 + \theta} + \frac{1}{n^2} \left(1 - \frac{1}{ 1 + \theta} \right)  \right] + 2p(1-p) \frac{1}{n^2} \left( 1 - \frac{1}{1 + \theta} \right) \\
&=  (1 - 2p +2p^2) \frac{1}{n} \frac{1}{1 + \theta} + \frac{1}{n^2} \left( 1 - \frac{1}{1 + \theta} \right).
\end{align*}
Equation~\eqref{esp_tilde_gamma_1} is obtained by noticing that $ \Tilde{\gamma}_{1,p} = \gamma_{1,p}.$ Equation~\eqref{esp_tilde_gamma_n+1} is derived by writing 
\[\Tilde{\gamma}_{n+1,p} = \gamma_{n+1,p} + \Tilde{\gamma}_{n,p} = p \beta_n + 1 - p\beta_n = 1. \]
Regarding Equation~\eqref{esp_tilde_gamma_j}. Since $\Tilde{\gamma}_{j,p}$ can be expressed as $ \Tilde{\beta}_j - p \beta_j,$ we deduce 
\begin{align*}
\mathbb{E}\left[\Tilde{\gamma}_{j,p} \right] = \frac{j}{n} - \frac{p}{n} = \frac{j-p}{n}.
\end{align*} 
Equation~\eqref{mom_ord_2_tilde_gamma_1} is deduced from the fact that $ \Tilde{\gamma}_{1,p} = \gamma_{1,p}$ and Equation~\eqref{mom_ord_2_tilde_gamma_n+1} results from $\Tilde{\gamma}_{n+1,p} =1.$  
Consider Equation~\eqref{mom_ord_2_tilde_gamma_j}. For all $j \in \llbracket 2,n \rrbracket,$ we have 
\[
    \mathbb{E}\left[ \Tilde{\gamma}_{j,p}^2 \right] = \mathbb{E}\left[ \left( \Tilde{\beta}_j - p \beta_j \right)p^2 \right] = \mathbb{E}\left[ \Tilde{\beta}_j^2 \right] - 2p \mathbb{E}\left[ \Tilde{\beta}_j \beta_j \right] + p^2 \mathbb{E}\left[ \beta_j^2 \right]. 
\] 
Using Equations~\eqref{mom_ord_2_beta}, \eqref{esp_beta_tildebeta} and \eqref{mom_ord_2_tilde_beta}, we get \begin{align*} 
    \mathbb{E}\left[ \Tilde{\gamma}_{j,p}^2 \right]
    &= \frac{j}{n} \frac{1}{1 + \theta} + \frac{j^2}{n^2} \left( 1 - \frac{1}{1 + \theta} \right) -2p \frac{1}{n} \frac{1}{1 + \theta} \\
    & \quad -2p \frac{j}{n^2} \left(1 - \frac{1}{1 + \theta} \right) + p^2 \frac{1}{n} \frac{1}{1 + \theta} + p^2 \frac{1}{n^2} \left(1 - \frac{1}{1 + \theta} \right)  \\
    &= (j-2p + p^2 ) \frac{1}{n} \frac{1}{1 + \theta} + (j^2 - 2pj + p^2) \frac{1}{n^2} \left( 1 - \frac{1}{1 + \theta} \right) \\
    &= (j-2p + p^2 ) \frac{1}{n} \frac{1}{ 1 + \theta}  + (j-p)^2 \frac{1}{n^2} \left( 1 - \frac{1}{1 + \theta} \right).
\end{align*}
Now, we prove Equation~\eqref{esp_cross_gamma_tilde_gamma}. For $j=1,$ we get 
\begin{align*}
    \mathbb{E}\left[ \gamma_{j+1,p} \Tilde{\gamma}_{j,p} \right] &= \mathbb{E}\left[ \gamma_{2,p} \gamma_{1,p} \right] \\
    &= \mathbb{E}\left[ \left( (1-p)\beta_2 + p \beta_1 \right)  (1-p)\beta_1 \right] \\
    &= \mathbb{E}\left[ (1-p)^2 \beta_2 \beta_1 + p(1-p) \beta_1^2 \right] \\
    &= \frac{(1-p)^2}{n^2} \left(1 - \frac{1}{1 + \theta}\right) + (p - p^2) \left[ \frac{1}{n} \frac{1}{1+\theta} + \frac{1}{n^2} \left( 1 - \frac{1}{1 + \theta} \right) \right]\\
    &= \frac{p - p^2}{n} \frac{1}{1+\theta} + \frac{1 - p}{n^2} \left( 1 - \frac{1}{1 + \theta} \right). 
\end{align*} 
For $j \neq1,$ we have 
\begin{align*}
    \mathbb{E}\left[ \gamma_{j+1,p} \Tilde{\gamma}_{j,p} \right] &= \mathbb{E}\left[ \left((1-p) \beta_{j+1} + p \beta_j \right) ( \Tilde{\beta}_j - p \beta_j ) \right] \\
    &= (1-p) \mathbb{E}\left[ \beta_{j+1} \Tilde{\beta}_j \right] + p \mathbb{E}\left[ \beta_j \Tilde{\beta}_j \right]  - p(1-p) \mathbb{E}\left[ \beta_{j+1} \beta_j \right] -p^2 \mathbb{E}\left[ \beta_j^2 \right].
\end{align*} 
Since $ \Tilde{\beta}_{j+1} = \Tilde{\beta}_j + \beta_{j+1},$ we get 
\begin{align*}
    &\mathbb{E}\left[ \gamma_{j+1,p} \Tilde{\gamma}_{j,p} \right]  \\
    &= (1-p) \mathbb{E}\left[ \beta_{j+1} ( \Tilde{\beta}_{j+1} - \beta_{j+1} ) \right] + p \mathbb{E}\left[ \beta_j \Tilde{\beta}_j \right] - p \mathbb{E}\left[ \beta_{j+1}\beta_j \right] + p^2 \mathbb{E}\left[ \beta_{j+1} \beta_j \right] - p^2 \mathbb{E}\left[ \beta_j^2 \right] \\
    &= (1-p) \mathbb{E}\left[ \beta_{j+1} \Tilde{\beta}_{j+1} \right] - (1-p) \mathbb{E}\left[ \beta_{j+1}^2 \right] + p \mathbb{E}\left[ \beta_j \Tilde{\beta}_j \right] - p \mathbb{E}\left[ \beta_{j+1} \beta_j \right] + p^2 \mathbb{E}\left[ \beta_{j+1} \beta_j \right] \\
    &\quad - p^2 \mathbb{E}\left[ \beta_j^2 \right] \\
    &= \mathbb{E}\left[ \beta_{j+1} \Tilde{\beta}_{j+1} \right] - p \left( \mathbb{E}\left[ \beta_{j+1} \Tilde{\beta}_{j+1} \right] - \mathbb{E}\left[ \beta_j \Tilde{\beta}_j \right] \right) \\
    & \quad - \left[ (1-p) \mathbb{E}\left[ \beta_{j+1}^2 \right] + p^2 \mathbb{E}\left[ \beta_j^2 \right] \right]  - p \mathbb{E}\left[ \beta_{j+1} \beta_j \right] + p^2 \mathbb{E}\left[ \beta_{j+1} \beta_j \right] \\
    &= \frac{1}{n} \frac{1}{1 + \theta} + \frac{j+1}{n^2} \left( 1 - \frac{1}{1+ \theta} \right) - p\frac{1}{n^2} \left(1 - \frac{1}{1 + \theta} \right)  \\
    &\quad - \left[ (1-p)+p^2 \right] \left[ \frac{1}{n} \frac{1}{1 + \theta} + \frac{1}{n^2} \left( 1 - \frac{1}{1 + \theta} \right)\right] -p \frac{1}{n^2} \left(1 - \frac{1}{1 + \theta} \right) + p^2 \frac{1}{n^2} \left(1 - \frac{1}{1 + \theta} \right) \\
    &= \frac{1}{n} \frac{1}{1 + \theta} \left( 1 + p - 1 -p^2 \right) + \frac{1}{n^2} \left( 1 - \frac{1}{ 1 + \theta} \right) \left( j+1 -p-1 + p -p^2 -p+p^2  \right) \\
    &= \frac{1}{n} \frac{1}{1 + \theta} (p - p^2 ) + \frac{1}{n^2} \left( 1 - \frac{1}{1 + \theta} \right) (j-p).
\end{align*}
For Equation~\eqref{esp_cross_gamma_tilde_gamma_n}, we have 
\begin{align*}
    \mathbb{E}\left[ \Tilde{\gamma}_{n+1,p} \Tilde{\gamma}_{n,p} \right] &= (1-p)^2 \left[ \frac{1}{n} \frac{1}{1 + \theta} + \frac{1}{n^2} \left(1 - \frac{1}{1+\theta} \right) \right] \\
    &= \mathbb{E}\left[ p \beta_n ( \Tilde{\beta}_n - p \beta_n ) \right] \\
    &= p \mathbb{E}\left[ \beta_n (1 - p \beta_n) \right] \\
    &= p \mathbb{E}\left[\beta_n\right] - p^2 \mathbb{E}\left[ \beta_n^2 \right] \\
    &= p \frac{1}{n} - p^2 \left[\frac{1}{n} \frac{1}{1 + \theta} + \frac{1}{n^2} \left( 1 - \frac{1}{1 + \theta} \right) \right].
\end{align*}
\end{proof}

\begin{lem}\label{lem:esp_ind_CDF}
We have the following statements:  
\begin{align}
    &\sum_{j=1}^{n-1} \mathbb{E}\left[ \mathds{1}_{[\varphi^{-1}(U_j^\ast), \varphi^{-1}(U_{j+1}^\ast))}(t) \right] = 1 - \varphi(t)^n - (1 - \varphi(t))^n \ ; \label{esp_indicatrice} \\
    & \sum_{j=1}^{n-1} j \mathbb{E}\left[ \mathds{1}_{[\varphi^{-1}(U_j^\ast), \varphi^{-1}(U_{j+1}^\ast))}(t) \right] = n \varphi(t) - n \varphi(t)^n \ ;\label{jesp}\\
    & \sum_{j=1}^{n-1} j^2 \mathbb{E}\left[ \mathds{1}_{[ \varphi^{-1}(U_j^\ast), \varphi^{-1}(U_{j+1}^\ast)) }(t) \right] = n^2\varphi(t)^2 - n \varphi(t)^2  -n^2 \varphi(t)^n + n \varphi(t). \label{j2_esp}
\end{align}
In particular, for any $k\geq 1$ we deduce that 
\begin{align}
    &\sum_{j=1}^{n-1} (j-k) \mathbb{E}\left[ \mathds{1}_{[\varphi^{-1}(U_j^\ast), \varphi^{-1}(U_{j+1}^\ast))}(t) \right] =  n \varphi(t) - n \varphi(t)^n - k \left( 1 - \varphi(t)^n - (1 - \varphi(t))^n \right) \ ; \label{j-k_esp}\\
    \begin{split}& \sum_{j=1}^{n-1} (j-k)^2 \mathbb{E}\left[ \mathds{1}_{[\varphi^{-1}(U_j^\ast), \varphi^{-1}(U_{j+1}^\ast))}(t) \right] \\
    &= n(n-1) \varphi(t)^2 - (n-k)^2 \varphi(t)^n + n (1-2k) \varphi(t) + k^2 - k^2 (1 - \varphi(t))^n.\end{split} \label{(j-k)^2_esp}
\end{align}
\end{lem}

\begin{proof}
Using the known expression of the joint density of the pairs $(U_j^\ast, U_{j+1}^\ast)$, recalled in Equation~\eqref{densite_jointe_Ui}, we can write 
\begin{align*}
&\sum_{j=1}^{n-1} \mathbb{E}\left[ \mathds{1}_{[\varphi^{-1}(U_j^\ast), \varphi^{-1}(U_{j+1}^\ast))}(t) \right]\\
&= \sum_{j=1}^{n-1} \int_0^1 \int_0^1 \mathds{1}_{[ \varphi^{-1}(u), \varphi^{-1}(v))}(t) n! \frac{u^{j-1}}{(j-1)!} \frac{(1-v)^{n-j-1}}{(n-j-1)!} \mathrm{d}v \mathrm{d}u \\
&= \sum_{j=1}^{n-1} \frac{n!}{(j-1)!(n-j-1)!} \int_0^{\varphi(t)} \int_{\varphi(t)}^1 u^{j-1} (1-v)^{n-j-1} \mathrm{d}v\mathrm{d}u
\end{align*} 
Then, by integration and by using the binomial theorem, we get 
\begin{align*}
\sum_{j=1}^{n-1} \mathbb{E}\left[ \mathds{1}_{[\varphi^{-1}(U_j^\ast), \varphi^{-1}(U_{j+1}^\ast))}(t) \right]\;
&= \sum_{j=1}^{n-1} \frac{n!}{j! (n-j)!} \varphi(t)^j (1-\varphi(t))^{n-j} \\
&= 1 - \varphi(t)^n - (1 - \varphi(t))^n. \end{align*} Same arguments lead to the remaining equations. We now focus on Equation~\eqref{jesp}: \begin{align*}
\sum_{j=1}^{n-1} j \mathbb{E}\left[\mathds{1}_{[\varphi^{-1}(U_j^\ast), \varphi^{-1}(U_{j+1}^\ast))}(t) \right] &= \sum_{j=1}^{n-1} \frac{n!}{(j-1)! (n-j)!} \varphi(t)^j (1 - \varphi(t))^{n-j}\\
&= n \varphi(t) \sum_{k=0}^{n-2} \binom{n-1}{k} \varphi(t)^k (1-\varphi(t))^{n-1-k} \\
&= n \varphi(t) (1 - \varphi(t)^{n-1} ) \\
&= n \varphi(t) - n \varphi(t)^n.
\end{align*} 
Regarding Equation~\eqref{j2_esp}, we have 
\begin{align*}
&\sum_{j=1}^{n-1} j^2 \mathbb{E}\left[ \mathds{1}_{[ \varphi^{-1}(U_j^\ast), \varphi^{-1}(U_{j+1}^\ast)) }(t) \right]\\
&= \sum_{j=1}^{n-1} n! j^2 \frac{\varphi(t)^j}{j!} \frac{(1 - \varphi(t))^{n-j}}{(n-j)!} \\
&= \sum_{j=1}^{n-1} n! j \frac{\varphi(t)^j}{(j-1)!} \frac{(1 - \varphi(t))^{n-j}}{(n-j)!}\\
&=\sum_{j=1}^{n-1} n! (j-1 +1) \frac{\varphi(t)^j}{(j-1)!} \frac{(1 - \varphi(t))^{n-j}}{(n-j)!} \\
&= \sum_{j=1}^{n-1} n! (j-1) \frac{\varphi(t)^j}{(j-1)!} \frac{(1 - \varphi(t))^{n-j}}{(n-j)!} + \sum_{j=1}^{n-1} n! \frac{\varphi(t)^j}{(j-1)!} \frac{(1 - \varphi(t))^{n-j}}{(n-j)!} \\
&= \sum_{j=2}^{n-1} n! (j-1) \frac{\varphi(t)^j}{(j-1)!} \frac{(1 - \varphi(t))^{n-j}}{(n-j)!} + n \varphi(t) \sum_{j=1}^{n-1} (n-1)! \frac{\varphi(t)^{j-1}}{(j-1)!} \frac{(1 - \varphi(t))^{n-j}}{(n-j)!}. 
\end{align*} 
Then,  
\begin{align*} 
&\sum_{j=1}^{n-1} j^2 \mathbb{E}\left[ \mathds{1}_{[ \varphi^{-1}(U_j^\ast), \varphi^{-1}(U_{j+1}^\ast)) }(t) \right]\\
    &= n(n-1) \varphi(t)^2 \sum_{j=2}^{n-1} \binom{n-2}{j-2} \varphi(t)^{j-2} (1 - \varphi(t))^{n-2 -(j-2)}\\
    &\quad \quad + n \varphi(t) \sum_{k=0}^{n-2} \binom{n-1}{k} \varphi(t)^k (1 - \varphi(t))^{n-1 -k} \\
    &= n(n-1) \varphi(t)^2 \sum_{k=0}^{n-3} \binom{n-2}{k} \varphi(t)^{k} (1 - \varphi(t))^{n-2 -k} \\
    &\quad \quad + n \varphi(t) \sum_{k=0}^{n-2} \binom{n-1}{k} \varphi(t)^k (1 - \varphi(t))^{n-1 -k}\\
    &= n(n-1) \varphi(t)^2 \left( 1 - \varphi(t)^{n-2} \right) + n \varphi(t) \left( 1 - \varphi(t)^{n-1} \right) \\
    &= n(n-1) \varphi(t)^2 - n(n-1) \varphi(t)^n + n \varphi(t) - n \varphi(t)^n \\
    &= n^2\varphi(t)^2 - n \varphi(t)^2  -n^2 \varphi(t)^n + n \varphi(t).
\end{align*} 
Equation~\eqref{j-k_esp} ensues from Equations~\eqref{jesp}-\eqref{j2_esp}. For Equation~\eqref{(j-k)^2_esp}, we have
\begin{align*}
    &\sum_{j=1}^{n-1} (j-k)^2 \mathbb{E}\left[ \mathds{1}_{[ \varphi^{-1}(U_j^\ast), \varphi^{-1}(U_{j+1}^\ast)) }(t) \right]\\
    &= \sum_{j=1}^{n-1} j^2  \mathbb{E}\left[ \mathds{1}_{[ \varphi^{-1}(U_j^\ast), \varphi^{-1}(U_{j+1}^\ast)) }(t) \right] -2k \sum_{j=1}^{n-1} j \mathbb{E}\left[ \mathds{1}_{[ \varphi^{-1}(U_j^\ast), \varphi^{-1}(U_{j+1}^\ast)) }(t) \right] \\
    &\quad \quad \quad + k^2 \sum_{j=1}^{n-1} \mathbb{E}\left[ \mathds{1}_{[ \varphi^{-1}(U_j^\ast), \varphi^{-1}(U_{j+1}^\ast)) }(t) \right] \\
    &= n^2 \varphi(t)^2 - n \varphi(t)^2 - n^2 \varphi(t)^n + n \varphi(t) -2k \left(n \varphi(t) - n \varphi(t)^n \right)\\
    &\quad \quad \quad + k^2 \left( 1 - \varphi(t)^n - (1 - \varphi(t))^n \right) \\
    &= (n^2 - n) \varphi(t)^2 - (n^2 - 2kn + k^2) \varphi(t)^n + n(1-2k)\varphi(t) + k^2 - k^2(1 - \varphi(t))^n \\
    &= (n^2 - n) \varphi(t)^2 - (n-k)^2 \varphi(t)^n + n(1-2k)\varphi(t) + k^2 -k^2(1-\varphi(t))^n.
\end{align*} 
\end{proof}

\section{Proof of Proposition~\ref{prop:warping} in BK case}\label{appendix::proof_BK}
In this appendix, we aim at proving Proposition~\ref{prop:warping} for the BK procedure using Lemmas~\ref{lem_exp_coeff_BK}-\ref{lem_exp_mom_phi_BK} proved in Appendix~\ref{annexe_preuve}. It is divided into two parts corresponding to the two statements of Proposition~\ref{prop:warping}. We recall that the notation $t \in [0,1] \mapsto w_n^{\BK}(t;\theta)$ refers to the path obtained with the BK procedure, with $\theta \in \mathbb{R}_+^\ast$ and $n \in \mathbb{N}^\ast.$  

\subsection{Proof of Proposition~\ref{prop:warping}(i) for BK Algorithm}

In this subsection, we give an explicit expression of the expectation of the process simulated with the BK procedure. Then, we prove asymptotic results leading to the convergence of the term $\mathbb{E}[w_n^{\BK}(t;\theta)]$ to $\varphi(t),$ for $t\in (0,1)$ and obtain a bound on the associated rate of convergence.

\subsubsection{Explicit expression of the expectation term}

For all $j \in \llbracket 0,n \rrbracket,$ for all $t \in (0,1),$ we set  
\begin{align*}
    A_j(t) = \left( \Tilde{\alpha}_j + \alpha_{j+1} \frac{t - U_j^\ast}{ U_{j+1}^\ast - U_j^\ast} \right) \mathds{1}_{[U_j^\ast, U_{j+1}^\ast)}(t).
\end{align*}  
Let $t \in (0,1).$ The expectation of $w_n^{\BK}(t;\theta)$ can be written as 
\begin{align*}
    \mathbb{E}\left[ w_n^{\BK}(t;\theta) \right] = \sum_{j=0}^n \mathbb{E}\left[ A_j(t) \right] = \sum_{j=0}^n \mathbb{E}\left[ \mathbb{E}\left[A_j(t) \;| \; U^\ast \right] \right].
\end{align*} 
First, $A_0(t) = ( \alpha_1 t / U_1^\ast) \mathds{1}_{[0, U_1^\ast)}(t)$ which, using Lemma~\ref{lem_exp_coeff_BK}, yields \begin{align*}
\mathbb{E}\left[ \mathbb{E}\left[ A_0(t) \;| \; U^\ast \right] \right] &= \mathbb{E}\left[ \frac{t}{U_1^\ast} \mathbb{E}\left[ \alpha_1 \; | \; U^\ast \right] \mathds{1}_{[0, U_1^\ast)}(t) \right] = \mathbb{E}\left[ \frac{t}{U_1^\ast} \varphi(U_1^\ast) \mathds{1}_{[0, U_1^\ast)}(t) \right].
\end{align*}
Second, we can rewrite the term $A_n(t)$ as follows: 
\begin{align*}
    A_n(t) &= \left( \Tilde{\alpha}_n + \frac{\alpha_{n+1}}{U_{n+1}^\ast - U_n^\ast} (t - U_n^\ast) \right) \mathds{1}_{[U_n^\ast, 1)}(t) \\
    &= \left( (1 - \alpha_{n+1} ) + \frac{\alpha_{n+1}}{1 - U_n^\ast} (t-U_n^\ast) \right) \mathds{1}_{[U_n^\ast, 1)}(t) \\
    &= \left( 1 - \alpha_{n+1} \frac{1 - t}{1 - U_n^\ast} \right) \mathds{1}_{[U_n^\ast, 1)}(t),
\end{align*}
whereby we deduce that $\mathbb{E}\left[ A_n(t) \right]$ expresses as follows: 
\begin{align*}
    \mathbb{E}\left[\mathbb{E}\left[A_n(t) | U^\ast \right] \right] &= \mathbb{E}\left[ \left( 1 - \mathbb{E}\left[\alpha_{n+1} \;|\; U^\ast \right] \frac{1-t}{1 -U_n^\ast} \right) \mathds{1}_{[U_n^\ast, 1)}(t) \right]\\
    &=\mathbb{E}\left[ \left(1 - (1 - \varphi(U_n^\ast)) \frac{1 - t}{1 - U_n^\ast} \right) \mathds{1}_{[U_n^\ast, 1)}(t) \right].
\end{align*}
Third, we express the remaining term $ \sum_{j=1}^{n-1} \mathbb{E}\left[A_j(t) \right].$ For all $j \in \llbracket1, n-1 \rrbracket,$ we have 
\begin{align*}
    \mathbb{E}\left[ \mathbb{E}\left[ A_j(t) \;| \; U^\ast \right] \right] &= \mathbb{E}\left[ \left( \mathbb{E}\left[ \Tilde{\alpha}_j \; | \; U^\ast \right] + \frac{t - U_j^\ast}{U_{j+1}^\ast - U_j^\ast} \mathbb{E}\left[ \alpha_{j+1} \; | \; U^\ast \right] \right) \mathds{1}_{[U_j^\ast, U_{j+1}^\ast)}(t) \right] \\
    &= \mathbb{E}\left[ \left( \varphi(U_j^\ast) + \frac{t - U_j^\ast}{U_{j+1}^\ast - U_j^\ast} ( \varphi(U_{j+1}^\ast) - \varphi(U_j^\ast)) \right) \mathds{1}_{[U_j^\ast, U_{j+1}^\ast)}(t)\right].
\end{align*}
Given that the joint distributions of the pairs $(U_j^\ast, U_{j+1}^\ast),$ for all $j \in \llbracket1, n \rrbracket,$  are known, as well as those of $U_1^\ast$ and $U_n^\ast,$ this expression is explicit and can be rewritten as
\begin{align} \label{esp_explicite_BK}
    \begin{split}&\mathcal E_n^{\BK}(t)\\ &= \int_t^1 \frac{t}{u} \varphi(u) n (1-u)^{n-1} \mathrm{d}u \\
    & \quad + \sum_{j=1}^{n-1} \int_t^1 \int_0^t \left(\varphi(u) + \frac{t - u}{v-u} (\varphi(v) - \varphi(u)) \right) \frac{n!}{(j-1)!(n-j-1)!} u^{j-1} (1-v)^{n-j-1} \mathrm{d}u \mathrm{d}v \\
    & \quad  + \int_0^t \left( 1 - (1 - \varphi(v)) \frac{1-t}{1-v} \right) n v^{n-1} \mathrm{d}v. \end{split} 
\end{align}

\subsubsection{Asymptotic behaviour of the expectation term}

Let $t \in (0,1).$ The definition of $w_n^{\BK}(.\;;\theta)$ leads to the following decomposition: \[
    \mathbb{E}\left[ w_n^{\BK}(t;\theta) \right] = \mathbb{E}\left[ A_0(t) \right] + \sum_{j=1}^{n-1} \mathbb{E}\left[ A_j(t) \right] + \mathbb{E}\left[ A_n(t) \right].\] 
\textit{Step 1: Highlighting $\varphi(t).$} 

\noindent We can separate $\sum_{j=1}^{n-1} \mathbb{E}\left[ A_j(t)\right] $ into two terms as follows: 
\begin{align*}
    &\sum_{j=1}^{n-1} \mathbb{E}\left[ A_j(t)\right]\\
    &= \sum_{j=1}^{n-1} \mathbb{E}\left[ \varphi(U_j^\ast) \mathds{1}_{[U_j^\ast, U_{j+1}^\ast)}(t) \right] + \sum_{j=1}^{n-1} \mathbb{E}\left[ \frac{t - U_j^\ast}{U_{j+1}^\ast - U_j^\ast} \left( \varphi(U_{j+1}^\ast) - \varphi(U_j^\ast) \right) \mathds{1}_{[U_j^\ast, U_{j+1}^\ast)}(t) \right].\end{align*} Using Equation~\eqref{esp_phi}, we get  \begin{align*}
    \sum_{j=1}^{n-1} \mathbb{E}\left[ A_j(t)\right]&= \varphi(t) - \varphi(t)t^n - \int_0^t \varphi'(u) \left( (1 - t+ u)^n - u^n \right) \mathrm{d}u \\
    & \quad+ \sum_{j=1}^{n-1} \mathbb{E}\left[ \frac{t - U_j^\ast}{U_{j+1}^\ast - U_j^\ast} \left( \varphi(U_{j+1}^\ast) - \varphi(U_j^\ast) \right) \mathds{1}_{[U_j^\ast, U_{j+1}^\ast)}(t) \right].
\end{align*} 

\noindent Consequently, we can get an expression of the expectation of $w_n^{\BK}(t;\theta)$ of the form \[
    \mathbb{E}\left[ w_n^{\BK}(t;\theta) \right] = \varphi(t) + R_n^{\BK}(t;\theta),\] where we set  
\begin{align*}
    R_n^{\BK}(t;\theta) &= \mathbb{E}\left[A_0(t) \right] +  \mathbb{E}\left[ A_n(t) \right] - \varphi(t)t^n - \int_0^t \varphi'(u) \left( (1 - t+ u)^n - u^n \right) \mathrm{d}u \\ 
     &\quad + \sum_{j=1}^{n-1} \mathbb{E}\left[ \frac{t - U_j^\ast}{U_{j+1}^\ast - U_j^\ast} \left( \varphi(U_{j+1}^\ast) - \varphi(U_j^\ast) \right) \mathds{1}_{[U_j^\ast, U_{j+1}^\ast)}(t) \right]. 
\end{align*} 

\noindent \textit{Step 2: asymptotic behaviour of} $R_n^{\BK}(t;\theta).$ 

\noindent $\bullet$ \textit{Term $\mathbb{E}\left[ A_0(t) \right]$.} We remark that  
\begin{align*}
    \mathbb{E}\left[ A_0(t) \right] &=  \mathbb{E}\left[ \frac{t}{U_1^\ast} \varphi(U_1^\ast) \mathds{1}_{[0, U_1^\ast)}(t) \right] \leq \mathbb{E}\left[ \varphi\left(U_1^\ast \right) \mathds{1}_{[0, U_1^\ast)}(t) \right]. \end{align*} Then, using the bound $\varphi \leq 1,$ we get  \begin{align*} 
    \mathbb{E}\left[ A_0(t) \right] &\leq \mathbb{E}\left[ \mathds{1}_{[0, U_1^\ast)}(t) \right]. 
\end{align*}
Thus, we have  
\begin{equation*}
    0\leq \mathbb{E}\left[A_0(t) \right] \leq (1-t)^n.
\end{equation*} 
Since $t \in (0,1),$ we have $ \displaystyle \lim_{n \to \infty} n (1-t)^n = 0.$ We deduce, as  $n \to  \infty,$  \[\mathbb{E}\left[ A_0(t) \right] =  O\left(\frac{1}{n}\right).\] 
$\bullet$ \textit{Term $\mathbb{E}\left[ A_n(t) \right].$} We can notice that $ \mathbb{E}\left[A_n(t) \right] $ is a positive term. Moreover, we can derive the inequality 
\[
    \mathbb{E}\left[A_n(t) \right]= \mathbb{E}\left[ \left(1 - (1 - \varphi(U_n^\ast)) \frac{1 - t}{1 - U_n^\ast} \right) \mathds{1}_{[U_n^\ast, 1)}(t) \right] \leq \mathbb{E}\left[ \mathds{1}_{[U_n^\ast, 1)}(t) \right],
\] 
which leads to  
\[
    0 \leq \mathbb{E}\left[ A_n(t) \right] \leq t^n \qquad \text{and} \qquad 
    \mathbb{E}\left[A_n(t) \right] = O\left( \frac{1}{n} \right).
\] 
$\bullet$ \textit{Term $\sum_{j=1}^{n-1} \mathbb{E}\left[ A_j(t)\right] - \varphi(t).$} The expression of $\sum_{j=1}^{n-1} \mathbb{E}\left[ A_j(t)\right] - \varphi(t)$ is given by the following formula: 
\begin{align*}
\sum_{j=1}^{n-1} \mathbb{E}\left[ A_j(t)\right] - \varphi(t)&= - \varphi(t)t^n - \int_0^t \varphi'(u) ( (1 - t+ u)^n - u^n) \mathrm{d}u \\
& \quad + \sum_{j=1}^{n-1} \mathbb{E}\left[ \frac{t - U_j^\ast}{U_{j+1}^\ast - U_j^\ast} \left( \varphi(U_{j+1}^\ast) - \varphi(U_j^\ast) \right) \mathds{1}_{[U_j^\ast, U_{j+1}^\ast)}(t) \right].
\end{align*} 
The term $\sum_{j=1}^{n-1} \mathbb{E}[((t - U_j^\ast)/ (U_{j+1}^\ast - U_j^\ast)) (\varphi(U_{j+1}^\ast) - \varphi(U_j^\ast)) \mathds{1}_{[U_j^\ast, U_{j+1}^\ast)}(t)]$ is positive. Moreover, we have 
\begin{align*}
    &\sum_{j=1}^{n-1} \mathbb{E}\left[ \frac{t - U_j^\ast}{U_{j+1}^\ast - U_j^\ast} \left( \varphi(U_{j+1}^\ast) - \varphi(U_j^\ast) \right) \mathds{1}_{[U_j^\ast, U_{j+1}^\ast)}(t) \right]\\
    &\leq \sum_{j=1}^{n-1} \mathbb{E}\left[ \left( \varphi\left( U_{j+1}^\ast \right) - \varphi\left( U_j^\ast\right) \right) \mathds{1}_{[U_j^\ast, U_{j+1}^\ast)}(t)   \right] \\
    &= \varphi(t) - \varphi(t)(1-t)^n - t^n + \int_t^1 \varphi'(v) \left( (1-v+t)^n - (1-v)^n \right)\mathrm{d}v  \\
    & \hspace{60pt} - \left( \varphi(t) - \varphi(t)t^n - \int_0^t \varphi'(u) \left((1-t+u)^n - u^n \right) \mathrm{d}u \right) \\
    &= \varphi(t)t^n -t^n - \varphi(t)(1-t)^n + \int_t^1 \varphi'(v) \left( (1-v+t)^n - (1-v)^n \right)\mathrm{d}v \\
    & \hspace{60pt}+ \int_0^t \varphi'(u) \left((1-t+u)^n - u^n \right) \mathrm{d}u . 
\end{align*}
The following two inequalities can be deduced: 
\begin{align*}
    & R_n^{\BK}(t;\theta) \geq - \varphi(t)t^n - \int_0^t \varphi'(u) \left( (1 - t+ u)^n - u^n \right) \mathrm{d}u ,\\
    & R_n^{\BK}(t;\theta) \leq (1-t)^n (1 - \varphi(t)) + \int_t^1 \varphi'(v) \left( (1-v+t)^n - (1-v)^n \right)\mathrm{d}v .
\end{align*} 
We set 
\begin{align*}
    &m_n(t) = - \varphi(t)t^n - \int_0^t \varphi'(u) \left( (1 - t+ u)^n - u^n \right) \mathrm{d}u, \\
    &M_n(t) = (1-t)^n (1 - \varphi(t)) + \int_t^1 \varphi'(v) \left( (1-v+t)^n - (1-v)^n \right)\mathrm{d}v.
\end{align*}
Since $\varphi \in \Tilde{\Gamma}$ and using the notation introduced in \eqref{sup}, we obtain 
\begin{align*}
    |m_n(t)|&=   \varphi(t)t^n + \int_0^t \varphi'(u) \left( (1 - t+ u)^n - u^n \right) \mathrm{d}u\\
    &\leq \varphi(t)t^n + \| \varphi' \|_{\infty} \int_0^t \left( (1 - t+ u)^n - u^n \right) \mathrm{d}u\\
    &= \varphi(t)t^n + \| \varphi' \|_{\infty} \frac{1}{n+1} \left( 1 - t^{n+1} - (1-t)^{n+1} \right).
\end{align*}
We deduce, as $n \to \infty,$  
\begin{equation*}
    m_n(t) = O\left( \frac{1}{n} \right). 
\end{equation*} 
Similar arguments give 
\begin{align*}
    |M_n(t)| &= (1-t)^n (1 - \varphi(t)) + \int_t^1 \varphi'(v) \left( (1-v+t)^n - (1-v)^n \right)\mathrm{d}v \\
    &\leq (1-t)^n (1 - \varphi(t)) + \| \varphi' \|_{\infty} \int_t^1 \left( (1-v+t)^n - (1-v)^n \right)\mathrm{d}v \\
    &= (1-t)^n (1 - \varphi(t)) + \| \varphi' \|_{\infty} \frac{1}{n+1} \left[1  -t^{n+1} - (1-t)^{n+1} \right].
\end{align*}
Thus, as $n \to \infty,$ 
\begin{equation*}
    M_n(t) = O\left( \frac{1}{n} \right) .
\end{equation*}
Since $R_n^{\BK}(t;\theta)$ is bounded by $\max\{|m_n(t)|\;;\;|M_n(t)| \},$ we conclude, as $n \to \infty,$ \begin{equation*}
    R_n^{\BK}(t;\theta)  = O\left( \frac{1}{n} \right). 
\end{equation*}
We finally get, as $n \to \infty,$ 
\begin{equation*}
    \mathbb{E}\left[w_n^{\BK}(t;\theta)\right] - \varphi(t)  =  R_n^{\BK}(t;\theta)  = O\left( \frac{1}{n} \right).
\end{equation*}

\subsection{Proof of Proposition~\ref{prop:warping}(ii) for BK Algorithm}

In this subsection, we prove the second statement of Proposition~\ref{prop:warping} for the BK case:  we derive an explicit expression for the variance of $w_n^\BK(t;\theta)$ at every point $t \in (0,1)$ and obtain the asymptotic convergence, as $n \to \infty,$ of this term to $\varphi(t)(1 - \varphi(t)) / (1+\theta),$ with a bound on the associated convergence rate.

\subsubsection{Explicit expression of the variance term}

Let $t \in (0,1).$ Since the indexed family $\{[U_j^\ast, U_{j+1}^\ast]\;;\;j \in \llbracket0, n \rrbracket \}$ is a partition of $[0,1],$ the variance term can be written in the following manner: \[\Var[w_n^{\BK}(t;\theta)] = \Var[A_0(t)] + \sum_{j=1}^{n-1}\Var[A_j(t)] + \Var[A_n(t)].\]
$\bullet$ \textit{First term.} 

\noindent We express the variance term $\Var[A_0(t)]$ using conditional expectation: 
\begin{align*}
    \Var[A_0(t)] &= \mathbb{E}\left[ A_0(t)^2 \right] - \mathbb{E}\left[A_0(t) \right]^2 \\
    &= \mathbb{E}\left[ \mathbb{E}\left[ A_0(t)^2 \; | \; U^\ast \right] \right] - \mathbb{E}\left[ A_0(t) \right]^2.
\end{align*}
The definition of $A_0(t)$ leads to 
\begin{align*}
    \mathbb{E}\left[ \mathbb{E}\left[ A_0(t)^2 \;|\; U^\ast \right] \right] 
    &= \mathbb{E}\left[ \mathbb{E}\left[ \left( \frac{t}{U_1^\ast} \right)^2 \alpha_1^2 \mathds{1}_{[0,U_1^\ast)}(t) \; | \; U^\ast \right] \right] \\
    &= \mathbb{E}\left[ \left( \frac{t}{U_1^\ast} \right)^2 \mathbb{E}\left[ \alpha_1^2 \; | \; U^\ast \right] \mathds{1}_{[0, U_1^\ast)}(t) \right]. \end{align*} Using Equation~\eqref{mom_ord_2_alpha}, we deduce \begin{align*}
    &\mathbb{E}\left[ \mathbb{E}\left[ A_0(t)^2 \;|\; U^\ast \right] \right] \\
    &= \mathbb{E}\left[ \left( \frac{t}{U_1^\ast} \right)^2 \left( \frac{1}{1 + \theta} \varphi(U_1^\ast) + \left( 1 - \frac{1}{1 + \theta} \right) \varphi(U_1^\ast)^2 \right) \mathds{1}_{[0, U_1^\ast)}(t) \right]  \\
    &= \frac{1}{1 + \theta} \mathbb{E}\left[ \left(\frac{t}{U_1^\ast} \right)^2 \varphi(U_1^\ast) \mathds{1}_{[0, U_1^\ast)}(t) \right] + \left( 1 - \frac{1}{1 + \theta} \right) \mathbb{E}\left[ \left( \frac{t}{U_1^\ast} \right)^2 \varphi(U_1^\ast)^2 \mathds{1}_{[0, U_1^\ast)}(t) \right].
\end{align*} 
Consequently, 
\begin{align*}
        \Var[A_0(t)] &= \frac{1}{1 + \theta} \mathbb{E}\left[ \left(\frac{t}{U_1^\ast} \right)^2 \varphi(U_1^\ast) \mathds{1}_{[0, U_1^\ast)}(t) \right] + \left( 1 - \frac{1}{1 + \theta} \right) \mathbb{E}\left[ \left( \frac{t}{U_1^\ast} \right)^2 \varphi(U_1^\ast)^2 \mathds{1}_{[0, U_1^\ast)}(t) \right] \\
        & \quad - \mathbb{E}\left[ \frac{t}{U_1^\ast} \varphi(U_1^\ast) \mathds{1}_{[0, U_1^\ast)}(t) \right]^2.
\end{align*}
$\bullet$ \textit{Term} $\Var[A_n(t)].$ 

\noindent Similar arguments lead to the following sequence of equations: 
\begin{align*}
    \Var[A_n(t)]  &= \mathbb{E}\left[ A_n(t)^2 \right] - \mathbb{E}\left[A_n(t)\right]^2 \\
    &= \mathbb{E}\left[ \mathbb{E}\left[ A_n(t)^2 \; | \; U^\ast \right] \right] - \mathbb{E}\left[ A_n(t)\right]^2 \\
    &= \mathbb{E}\left[ \mathbb{E}\left[ \left( 1 - \alpha_{n+1} \frac{1-t}{1 - U_n^\ast} \right)^2 \mathds{1}_{[U_n^\ast, 1)}(t) \; | \; U^\ast\right] \right] - \mathbb{E}\left[ A_n(t) \right]^2 \\
    &= \mathbb{E}\left[ \mathds{1}_{[U_n^\ast, 1)}(t) \right] - 2 \mathbb{E}\left[ \frac{1-t}{1 - U_n^\ast} \mathbb{E}\left[ \alpha_{n+1} \; | \; U^\ast  \right] \mathds{1}_{[U_n^\ast,1)}(t) \right] \\
    & \quad \quad + \mathbb{E}\left[ \left( \frac{1-t}{1-U_n^\ast}\right)^2 \mathbb{E}\left[ \alpha_{n+1} ^2 \; | \; U^\ast \right] \mathds{1}_{[U_n^\ast,1)}(t)\right]  - \mathbb{E}\left[A_n(t) \right]^2 .
\end{align*}
Thus, 
\begin{align*}
    \Var[A_n(t)] &= \mathbb{E}\left[ \mathds{1}_{[U_n^\ast, 1)}(t) \right] -2 \mathbb{E}\left[ \frac{1-t}{1 - U_n^\ast} \left( 1 - \varphi(U_n^\ast) \right) \mathds{1}_{[U_n^\ast, 1)}(t) \right] \\
    & \quad \quad + \frac{1}{1 + \theta}\mathbb{E}\left[ \left(\frac{1-t}{1 - U_n^\ast} \right)^2 \left(1 - \varphi(U_n^\ast) \right) \mathds{1}_{[U_n^\ast, 1)}(t)   \right] \\
    &\quad \quad + \left( 1 - \frac{1}{1 + \theta} \right) \mathbb{E}\left[ \left( \frac{1-t}{1 - U_n^\ast} \right)^2 \left( 1 - \varphi(U_n^\ast) \right)^2 \mathds{1}_{[U_n^\ast, 1)}(t) \right] - \mathbb{E}\left[ A_n(t) \right]^2 .
\end{align*} 
As a result, 
\begin{align*}
    \Var[A_n(t)] &= \mathbb{E}\left[ \mathds{1}_{[U_n^\ast, 1)}(t) \right] -2 \mathbb{E}\left[ \frac{1-t}{1 - U_n^\ast} \left( 1 - \varphi(U_n^\ast) \right) \mathds{1}_{[U_n^\ast, 1)}(t) \right] \\
    & \quad \quad + \frac{1}{1 + \theta}\mathbb{E}\left[ \left(\frac{1-t}{1 - U_n^\ast} \right)^2 \left(1 - \varphi(U_n^\ast) \right) \mathds{1}_{[U_n^\ast, 1)}(t)   \right] \\
    &\quad \quad + \left( 1 - \frac{1}{1 + \theta} \right) \mathbb{E}\left[ \left( \frac{1-t}{1 - U_n^\ast} \right)^2 \left( 1 - \varphi(U_n^\ast) \right)^2 \mathds{1}_{[U_n^\ast, 1)}(t) \right] \\
    &\quad \quad  - \mathbb{E}\left[ \left(1 - (1 - \varphi(U_n^\ast)) \frac{1 - t}{1 - U_n^\ast} \right) \mathds{1}_{[U_n^\ast, 1)}(t) \right]^2.
\end{align*} 
Now, we express each term separately. Using the known densities, we have the following equations: 
\begin{align*}
    &\mathbb{E}\left[ \mathds{1}_{[U_n^\ast, 1)}(t) \right] = \int_0^1 \mathds{1}_{[v,1)}(t) n v^{n-1} \mathrm{d}v = \int_0^t n v^{n-1} \mathrm{d}v = t^n \;;\\
    &\mathbb{E}\left[ \frac{1-t}{1 - U_n^\ast} \left( 1 - \varphi(U_n^\ast) \right) \mathds{1}_{[U_n^\ast,1)}(t) \right] = \int_0^t \frac{1 - t}{1-v} \left( 1 - \varphi(v) \right)n v^{n-1} \mathrm{d}v\;;\\
    &\mathbb{E}\left[ \left( \frac{1 - t}{1 - U_n^\ast} \right)^2 (1 - \varphi(U_n^\ast)) \mathds{1}_{[U_n^\ast, 1)}(t)  \right] = \int_0^t \left( \frac{1 - t}{1 - v} \right)^2 (1 - \varphi(v)) n v^{n-1} \mathrm{d}v.
\end{align*}
Moreover, 
\begin{align*}
    \mathbb{E}\left[ \left( \frac{1 - t}{1  -U_n^\ast} \right)^2 (1- \varphi(U_n^\ast))^2 \mathds{1}_{[U_n^\ast, 1)}(t)  \right] &= \int_0^t \left( \frac{1 - t}{1 - v} \right)^2 (1-\varphi(v))^2 n v^{n-1} \mathrm{d}v .
\end{align*}
We deduce the following explicit expression: 
\begin{align*}
    &\Var[A_n(t)]\\
    &= t^n -2 \int_0^t \frac{1 - t}{1-v} \left( 1 - \varphi(v) \right)n v^{n-1} \mathrm{d}v   + \frac{1}{1 + \theta} \int_0^t \left( \frac{1 - t}{1 - v} \right)^2 (1 - \varphi(v)) n v^{n-1} \mathrm{d}v\\
    &\quad  + \left( 1 - \frac{1}{1 + \theta} \right) \int_0^t \left( \frac{1 - t}{1 - v} \right)^2 (1-\varphi(v))^2 n v^{n-1} \mathrm{d}v  \\
    &\quad - \left(  \int_0^t \left( 1 - (1 - \varphi(v)) \frac{1-t}{1-v} \right) n v^{n-1} \mathrm{d}v \right)^2.
\end{align*}
$\bullet$ \textit{Remaining terms.} \newline 
The term $ \sum_{j=1}^{n-1} \Var[A_j(t)]$ is expressed as follows: 
\[
    \sum_{j=1}^{n-1} \Var[A_j(t)] = \sum_{j=1}^{n-1} \left( \mathbb{E}\left[A_j(t) \right]^2 - \mathbb{E}\left[A_j(t) \right]^2 \right).
\]
Expanding the term $\mathbb{E}\left[A_j(t)^2 \right],$ we derive  
\begin{align*}
    &\mathbb{E}\left[A_j(t)^2 \right]\\
    &= \mathbb{E}\left[ \mathbb{E}\left[A_j(t)^2 \; | \; U^\ast \right] \right] \\
    &= \mathbb{E}\left[ \mathbb{E}\left[ \left( \Tilde{\alpha}_j + \alpha_{j+1} \frac{t - U_j^\ast}{U_{j+1}^\ast - U_j^\ast} \right)^2 \mathds{1}_{[U_j^\ast, U_{j+1}^\ast)}(t) \; | \; U^\ast \right] \right]\\
    &= \mathbb{E}\left[ \mathbb{E}\left[ \left( \Tilde{\alpha}_j ^2 + 2 \Tilde{\alpha}_j\alpha_{j+1} \frac{t - U_j^\ast}{U_{j+1}^\ast - U_j^\ast} + \alpha_{j+1}^2\left(\frac{t-U_j^\ast }{U_{j+1}^\ast - U_j^\ast }\right)^2 \right) \mathds{1}_{[U_j^\ast, U_{j+1}^\ast)}(t) \; | \; U^\ast \right] \right].
\end{align*} 
The results of Lemma~\ref{lem_exp_mom_phi_BK} then give  
\begin{align*}
    &\sum_{j=1}^{n-1 }\mathbb{E}\left[A_j(t)^2 \right] \\
    &= \frac{1}{1 + \theta} \sum_{j=1}^{n-1} \mathbb{E}\left[ \varphi(U_j^\ast) \mathds{1}_{[U_j^\ast, U_{j+1}^\ast)}(t) \right] + \left( 1 - \frac{1}{1 + \theta} \right) \sum_{j=1}^{n-1} \mathbb{E}\left[ \varphi(U_j^\ast)^2 \mathds{1}_{[U_j^\ast, U_{j+1}^\ast)}(t) \right] \\
    & \quad+2 \left( 1 - \frac{1}{1 + \theta} \right)\sum_{j=1}^{n-1}\mathbb{E}\left[ \frac{t - U_j^\ast}{U_{j+1}^\ast - U_j^\ast}\left( \varphi(U_{j+1}^\ast) - \varphi(U_j^\ast) \right) \varphi(U_j^\ast)  \mathds{1}_{[U_j^\ast, U_{j+1}^\ast)}(t) \right] \\
    &\quad  + \frac{1}{1 + \theta} \sum_{j=1}^{n-1} \mathbb{E}\left[ \left( \varphi(U_{j+1}^\ast) - \varphi(U_j^\ast) \right) \left( \frac{t - U_j^\ast}{U_{j+1}^\ast - U_j^\ast} \right)^2 \mathds{1}_{[U_j^\ast, U_{j+1}^\ast)}(t) \right] \\
    &\quad  + \left( 1 - \frac{1}{1 + \theta} \right) \sum_{j=1}^{n-1}\mathbb{E}\left[ \left( \varphi(U_{j+1}^\ast) - \varphi(U_j^\ast) \right)^2 \left( \frac{t - U_j^\ast}{U_{j+1}^\ast - U_j^\ast} \right)^2 \mathds{1}_{[U_j^\ast, U_{j+1}^\ast)}(t) \right].
\end{align*} 
Since the densities of the pairs $(U_j^\ast, U_{j+1}^\ast)$, and those of $U_1^\ast$ and $U_n^\ast$ are known, this formula is an explicit expression of $\Var[w_n^{\BK}(t;\theta)],$ that can be rewritten as in Equation~\eqref{var_explicite_BK}: 

\begin{align} \label{var_explicite_BK}
   \begin{split} &\mathcal V_n^{\BK}(t;\theta) \\
    &= \frac{1}{1 + \theta} \int_t^1 \left(\frac{t}{u}\right)^2 \varphi(u) n (1-u)^{n-1} \mathrm{d}u + \left( 1 - \frac{1}{1 + \theta} \right) \int_t^1 \left(\frac{t}{u}\right)^2 \varphi(u)^2 n (1-u)^{n-1} \mathrm{d}u\\
    &\quad - \left( \int_t^1 \frac{t}{u} \varphi(u) n (1-u)^{n-1} \mathrm{d}u\right)^2 \\
    & \quad + \frac{1}{1 + \theta} \sum_{j=1}^{n-1} \int_t^1 \int_0^t \varphi(u) \frac{n!}{(j-1)!(n-j-1)!} u^{j-1} (1-v)^{n-j-1} \mathrm{d}u \mathrm{d}v \\
    &\quad + \left(1 -\frac{1}{1+\theta} \right) \sum_{j=1}^{n-1} \int_t^1 \int_0^t \varphi(u)^2 \frac{n!}{(j-1)!(n-j-1)!} u^{j-1} (1-v)^{n-j-1} \mathrm{d}u \mathrm{d}v \\
    &\quad + 2 \left( 1 - \frac{1}{1 + \theta} \right) \sum_{j=1}^{n-1} \int_t^1 \int_0^t \frac{t - u}{v-u} \left(\varphi(v) - \varphi(u)\right) \varphi(u) n!\frac{u^{j-1} (1-v)^{n-j-1}}{(j-1)!(n-j-1)!}  \mathrm{d}u \mathrm{d}v \\
    & \quad + \frac{1}{1 + \theta} \sum_{j=1}^{n-1} \int_t^1 \int_0^t \left( \varphi(v) - \varphi(u) \right) \left( \frac{t- u}{v-u} \right)^2 \frac{n!}{(j-1)!(n-j-1)!} u^{j-1} (1-v)^{n-j-1} \mathrm{d}u \mathrm{d}v \\
    &\quad + \left( 1 - \frac{1}{1 + \theta} \right) \int_t^1 \int_0^t \left( \varphi(v) - \varphi(u) \right)^2 \left( \frac{t-u}{v-u} \right)^2 n!\frac{  u^{j-1} (1-v)^{n-j-1}}{(j-1)!(n-j-1)!} \mathrm{d}u \mathrm{d}v \\
    & \quad -  \sum_{j=1}^{n-1} \left( \int_t^1 \int_0^t \left(\varphi(u) + \frac{t - u}{v-u} (\varphi(v) - \varphi(u)) \right) n!\frac{ u^{j-1} (1-v)^{n-j-1}}{(j-1)!(n-j-1)!} \mathrm{d}u \mathrm{d}v \right)^2 \\
    &\quad + t^n -2 \int_0^t \frac{1 - t}{1-v} \left( 1 - \varphi(v) \right)n v^{n-1} \mathrm{d}v   + \frac{1}{1 + \theta} \int_0^t \left( \frac{1 - t}{1 - v} \right)^2 (1 - \varphi(v)) n v^{n-1} \mathrm{d}v\\
    &\quad  + \left( 1 - \frac{1}{1 + \theta} \right) \int_0^t \left( \frac{1 - t}{1 - v} \right)^2 (1-\varphi(v))^2 n v^{n-1} \mathrm{d}v  \\
    &\quad - \left(  \int_0^t \left( 1 - (1 - \varphi(v)) \frac{1-t}{1-v} \right) n v^{n-1} \mathrm{d}v \right)^2. \end{split}
\end{align}

\subsubsection{Asymptotic behaviour of the variance term}

Using the notations of the proof for the convergence result of the expectation term, we get: 
\[
    \Var[w_n^{\BK}(t;\theta)] = \mathbb{E}\left[A_0(t)^2\right] + \mathbb{E}\left[A_n(t)^2 \right] + \sum_{j=1}^{n-1} \mathbb{E}\left[A_j(t)^2\right] - \left( \varphi(t) + R_n^{\BK}(t;\theta) \right)^2.
\]
$\bullet$\textbf{ Highlighting the term } $\varphi(t) (1 - \varphi(t)) / ( 1 + \theta).$ 

We break down the term $\sum_{j=1}^{n-1} \mathbb{E}\left[A_j(t)^2\right] $ into two sums: 
\begin{align*}
    \sum_{j=1}^{n-1} \mathbb{E}\left[A_j(t)^2\right] &= \sum_{j=1}^{n-1} \mathbb{E}\left[ \mathbb{E}\left[ \Tilde{\alpha}_j ^2 |U^\ast \right] \mathds{1}_{[U_j^\ast, U_{j+1}^\ast)}(t) \right] \\
    & \quad  + \sum_{j=1}^{n-1} \mathbb{E}\left[ \mathbb{E}\left[ 2 \Tilde{\alpha}_j \alpha_{j+1} \frac{t - U_j^\ast}{U_{j+1}^\ast - U_j^\ast} + \alpha_{j+1}^2 \left( \frac{t - U_j^\ast}{U_{j+1}^\ast - U_j^\ast} \right)^2 |U^\ast \right] \mathds{1}_{[U_j^\ast, U_{j+1}^\ast)}(t) \right].
\end{align*}
Then, Lemmas~\ref{lem_exp_coeff_BK}-\ref{lem_exp_mom_phi_BK} give 
\begin{align*}
    &\sum_{j=1}^{n-1} \mathbb{E}\left[ \mathbb{E}\left[ \Tilde{\alpha}_j ^2 |U^\ast \right] \mathds{1}_{[U_j^\ast, U_{j+1}^\ast)}(t) \right]\\
    &= \frac{1}{1 + \theta} \sum_{j=1}^{n-1} \mathbb{E}\left[ \varphi(U_j^\ast) \mathds{1}_{[U_j^\ast, U_{j+1}^\ast)}(t) \right]
     + \left( 1 - \frac{1}{1 + \theta}\right) \sum_{j=1}^{n-1} \mathbb{E}\left[ \varphi(U_j^\ast)^2 \mathds{1}_{[U_j^\ast, U_{j+1}^\ast)}(t) \right] \\
     &= \frac{1 }{1 + \theta} \sum_{j=1}^{n-1} \mathbb{E}\left[ (\varphi(U_j^\ast) - \varphi(U_j^\ast)^2) \mathds{1}_{[U_j^\ast, U_{j+1}^\ast)}(t) \right] + \sum_{j=1}^{n-1} \mathbb{E}\left[ \varphi(U_j^\ast)^2 \mathds{1}_{[U_j^\ast, U_{j+1}^\ast)}(t) \right] \\
     &= \frac{1}{1 + \theta} \left[ \varphi(t) - \varphi(t) t^n - \int_0^t \varphi'(u) \left( (1-t+u)^n -u^n \right) \mathrm{d}u \right] \\
     &\quad - \frac{1}{1 + \theta}\left[ \varphi(t)^2 - \varphi(t)^2t^n  - 2 \int_0^t \varphi'(u) \varphi(u) \left( (1-t+u)^n - u^n \right) \mathrm{d}u\right] \\
     &\quad + \varphi(t)^2 - \varphi(t)^2t^n  - 2 \int_0^t \varphi'(u) \varphi(u) \left( (1-t+u)^n - u^n \right) \mathrm{d}u \\
     &= \frac{1}{1 + \theta} \varphi(t) - \frac{1}{1 + \theta} \varphi(t)^2 + \varphi(t)^2 - \frac{1}{1 + \theta} \varphi(t) (1- \varphi(t))t^n \\
     &\quad + \frac{1}{1 + \theta} \left[ 2 \int_0^t \varphi'(u) \varphi(u) \left( (1-t + u)^n - u^n \right) \mathrm{d}u - \int_0^t \varphi'(u)\left( (1-t + u)^n - u^n \right) \mathrm{d}u\right] \\
     &\quad - \varphi(t)^2 t^n -2\int_0^t \varphi'(u) \varphi(u) \left( (1-t + u)^n - u^n \right) \mathrm{d}u.
\end{align*}
We can thus rewrite the variance term as follows: 
\begin{align*}
    \Var[w_n^{\BK}(t;\theta)] = \frac{1}{1 + \theta} \varphi(t) - \frac{1}{1 + \theta} \varphi(t)^2 + \Tilde{R}_n^{\BK}(t;\theta),
\end{align*} where we set 
\begin{align*}
    \Tilde{R}_n^{\BK}(t;\theta) &=  \mathbb{E}\left[A_0(t)^2\right] + \mathbb{E}\left[A_n(t)^2 \right]\\
    &\quad  + \frac{1}{1 + \theta} \left[ - \varphi(t)(1 - \varphi(t))t^n + 2 \int_0^t \varphi'(u) \varphi(u) \left( (1-t + u)^n - u^n \right) \mathrm{d}u \right] \\
    &\quad - \frac{1}{1 + \theta} \int_0^t \varphi'(u)\left( (1-t + u)^n - u^n \right) \mathrm{d}u \\
     &\quad - \varphi(t)^2 t^n -2\int_0^t \varphi'(u) \varphi(u) \left( (1-t + u)^n - u^n \right) \mathrm{d}u\\
      & \quad  + \sum_{j=1}^{n-1} \mathbb{E}\left[ \mathbb{E}\left[ 2 \Tilde{\alpha}_j \alpha_{j+1} \frac{t - U_j^\ast}{U_{j+1}^\ast - U_j^\ast} + \alpha_{j+1}^2 \left( \frac{t - U_j^\ast}{U_{j+1}^\ast - U_j^\ast} \right)^2 |U^\ast \right] \mathds{1}_{[U_j^\ast, U_{j+1}^\ast)}(t) \right]\\
     &\quad - 2 \varphi(t)R_n^{\BK}(t;\theta)  - R_n^{\BK}(t;\theta)^2. 
\end{align*} 
$\bullet$ \textbf{Term} $\mathbb{E}\left[A_0(t)^2\right].$

We note that 
\begin{align*}
    \mathbb{E}\left[ A_0(t)^2 \right] &= \frac{1}{1 + \theta} \mathbb{E}\left[ \left( \frac{t}{U_1^\ast} \right)^2 \varphi(U_1^\ast) \mathds{1}_{[0, U_1^\ast)}(t) \right] + \left( 1 - \frac{1}{1 + \theta} \right) \mathbb{E}\left[ \left(\frac{t}{U_1^\ast}\right)^2 \varphi(U_1^\ast)^2 \mathds{1}_{[0, U_1^\ast)}(t) \right] \\
    &= \frac{1}{1 + \theta} \mathbb{E}\left[ \left(\frac{t}{U_1^\ast} \right)^2 \left( \varphi(U_1^\ast) - \varphi(U_1^\ast)^2\right) \mathds{1}_{[0, U_1^\ast)}(t) \right] +  \mathbb{E}\left[ \left(\frac{t}{U_1^\ast}\right)^2 \varphi(U_1^\ast)^2 \mathds{1}_{[0, U_1^\ast)}(t) \right]\\
    &\leq \frac{1}{1 + \theta} \mathbb{E}\left[ \varphi(U_1^\ast)(1 - \varphi(U_1^\ast)) \mathds{1}_{[0, U_1^\ast)}(t) \right] + \mathbb{E}\left[ \varphi(U_1^\ast)^2 \mathds{1}_{[0, U_1^\ast)}(t) \right].\end{align*} Using that the function $x \in [0,1] \mapsto x(1-x)$ is bounded by $1/4,$ we derive the following inequality: \begin{align*}
    \mathbb{E}\left[ A_0(t)^2 \right]
    &\leq \frac{1}{1+\theta} \frac{1}{4} \mathbb{E}\left[\mathds{1}_{[0, U_1^\ast)}(t) \right] + \mathbb{E}\left[\mathds{1}_{[0, U_1^\ast)}(t)\right] \\
    &= \left( \frac{1}{4} \frac{1}{1 +  \theta} + 1 \right) (1-t)^n.
\end{align*}
We have proved 
\begin{align*}
    0 \leq \mathbb{E}\left[A_0(t)^2 \right] \leq \left( \frac{1}{4} \frac{1}{1 +  \theta} + 1 \right) (1-t)^n.
\end{align*} 
As a result, as $n \to \infty,$ 
\begin{align*}
    \mathbb{E}\left[A_0(t)^2 \right] = O\left( \frac{1}{n} \right).
\end{align*}
$\bullet$ \textbf{Term} $\mathbb{E}\left[A_n(t)^2\right].$ 

The expression of  $\mathbb{E}\left[A_n(t)^2\right]$ is given by 
\begin{align*}
    \mathbb{E}\left[A_n(t)^2\right] &= \mathbb{E}\left[ \mathds{1}_{[U_n^\ast, 1)}(t) \right] - 2 \mathbb{E}\left[ \frac{1 - t}{1 - U_n^\ast} (1 - \varphi(U_n^\ast)) \mathds{1}_{[U_n^\ast, 1)}(t) \right]\\ 
    &\quad +\frac{1}{1 + \theta} \mathbb{E}\left[ \left( \frac{1-t}{1 - U_n^\ast} \right)^2 \left[ ( 1 - \varphi(U_n^\ast)) - (1-\varphi(U_n^\ast))^2 \right]  \mathds{1}_{[U_n^\ast, 1)}(t)  \right] \\
    & \quad + \mathbb{E}\left[ \left( \frac{1 -t}{1-U_n^\ast}\right)^2 (1-\varphi(U_n^\ast))^2  \mathds{1}_{[U_n^\ast, 1)}(t) \right]. 
\end{align*} 
By the argument mentioned above, we get 
\begin{align*} 
    \mathbb{E}\left[A_n(t)^2\right] 
    & \leq \left(1 +  \frac{1}{4}\frac{1}{1 + \theta} + 1 \right) \mathbb{E}\left[ \mathds{1}_{[U_n^\ast, 1)}(t) \right] \\
    &= \left(2 + \frac{1}{4} \frac{1}{1 + \theta} \right) t^n.
\end{align*}
We have derived 
\[0 \leq \mathbb{E}\left[ A_n(t)^2 \right] \leq \left(2 + \frac{1}{4} \frac{1}{1 + \theta} \right) t^n.\]
Thus, as $n \to \infty,$ we have 
\[\mathbb{E}\left[A_n(t)^2\right] = O\left( \frac{1}{n} \right).\]
Before continuing the proof, we introduce some notations. Let 
\begin{align*}
    T_{1,n}^{\BK}(t;\theta)&= \frac{1}{1 + \theta} \left[ - \varphi(t)(1 - \varphi(t))t^n + 2 \int_0^t \varphi'(u) \varphi(u) \left( (1-t + u)^n - u^n \right) \mathrm{d}u \right] \\
    & \quad \quad - \frac{1}{1 + \theta} \int_0^t \varphi'(u)\left( (1-t + u)^n - u^n \right) \mathrm{d}u,\\
     T_{2,n}^{\BK}(t;\theta) &= - \varphi(t)^2 t^n -2\int_0^t \varphi'(u) \varphi(u) \left( (1-t + u)^n - u^n \right) \mathrm{d}u, \\
     T_{3,n}^{\BK}(t;\theta) &= \sum_{j=1}^{n-1} \mathbb{E}\left[ \mathbb{E}\left[ 2 \Tilde{\alpha}_j \alpha_{j+1} \frac{t - U_j^\ast}{U_{j+1}^\ast - U_j^\ast} + \alpha_{j+1}^2 \left( \frac{t - U_j^\ast}{U_{j+1}^\ast - U_j^\ast} \right)^2 |U^\ast \right] \mathds{1}_{[U_j^\ast, U_{j+1}^\ast)}(t) \right],
\end{align*}
such that 
\begin{align*}
   \Tilde{R}_n^{\BK}(t;\theta) &=  \mathbb{E}\left[A_0(t)^2\right] + \mathbb{E}\left[A_n(t)^2 \right]\\
   &\quad  \quad + T_{1,n}^{\BK}(t;\theta) +  T_{2,n}^{\BK}(t;\theta) +   T_{3,n}^{\BK}(t;\theta) - 2 \varphi(t)R_n^{\BK}(t;\theta)  - R_n^{\BK}(t;\theta)^2.
\end{align*}
Using the bounds $\varphi (1 - \varphi) \leq \frac{1}{4}$ and $  \varphi \leq 1,$ we get 
\begin{align*}
    | T_{1,n}^{\BK}(t;\theta)| &\leq \frac{1}{1 + \theta} \left \lvert- \varphi(t)(1 - \varphi(t))t^n \right \rvert + \frac{2}{1 + \theta} \left\lvert \int_0^t \varphi'(u) \varphi(u) \left( (1-t + u)^n - u^n \right) \mathrm{d}u \right\rvert  \\
    &\quad+ \frac{1}{1 + \theta} \left\lvert \int_0^t \varphi'(u)\left( (1-t + u)^n - u^n \right) \mathrm{d}u \right\rvert \\
    &\leq \frac{1}{4} \frac{1}{1 + \theta} t^n + \frac{2}{1 + \theta} \| \varphi'\|_{\infty} \int_0^t \left( (1-t + u)^n - u^n \right) \mathrm{d}u \\
    &\quad +\frac{1}{1 + \theta} \| \varphi' \|_{\infty} \int_0^t \left( (1-t + u)^n - u^n \right) \mathrm{d}u \\
    &= \frac{1}{4} \frac{1}{1 + \theta} t^n + \frac{3}{1+\theta}\frac{1}{n+1} (1 - t^{n+1} - (1-t)^{n+1}).
\end{align*} 
We have obtained, as $n \to \infty,$  
\[ T_{1,n}^{\BK}(t;\theta) = O\left(\frac{1}{n} \right).\]
Similar arguments to those above lead to the following sequence of inequalities: 
\begin{align*}
    | T_{2,n}^{\BK}(t;\theta)| &\leq \left \lvert- \varphi(t)^2 t^n \right \rvert + \left \lvert -2\int_0^t \varphi'(u) \varphi(u) \left( (1-t + u)^n - u^n \right) \mathrm{d}u\right \rvert\\
    &\leq t^n + 2 \|\varphi' \|_{\infty} \int_0^t \left( (1-t + u)^n - u^n \right) \mathrm{d}u\\
    &= t^n + 2 \| \varphi' \|_{\infty} \frac{1}{n+1} (1 - t^{n+1} - (1-t)^{n+1}). 
\end{align*} 
Consequently, as $n \to \infty,$ 
\[ T_{2,n}^{\BK}(t;\theta) = O\left( \frac{1}{n} \right).\] 
We can notice that $ T_{3,n}^{\BK}(t;\theta)$ is a non-negative term and  
\begin{align*}
     T_{3,n}^{\BK}(t;\theta)&\leq \sum_{j=1}^{n-1} \mathbb{E}\left[ \mathbb{E}\left[ 2 \Tilde{\alpha}_j \alpha_{j+1} + \alpha_{j+1}^2 | U^\ast \right] \mathds{1}_{[U_j^\ast, U_{j+1}^\ast)}(t) \right] .
\end{align*} 
Using the results of Lemma~\ref{lem_exp_coeff_BK}, we derive 
\begin{align*}   
     &T_{3,n}^{\BK}(t;\theta) \\
    &= \sum_{j=1}^{n-1} \mathbb{E}\left[ 2\left(1- \frac{1}{1 + \theta} \right) ( \varphi(U_{j+1}^\ast) - \varphi(U_j^\ast))\varphi(U_j^\ast) \mathds{1}_{[U_j^\ast, U_{j+1}^\ast)}(t) \right] \\
    &\quad + \frac{1}{1 + \theta} \sum_{j=1}^{n-1} \mathbb{E}\left[ (\varphi(U_{j+1}^\ast) - \varphi(U_j^\ast)) \mathds{1}_{[U_j^\ast, U_{j+1}^\ast)}(t) \right] \\
    &\quad + \left( 1 - \frac{1}{1 + \theta} \right) \sum_{j=1}^{n-1} \mathbb{E}\left[ (\varphi(U_{j+1}^\ast) - \varphi(U_j^\ast))^2 \mathds{1}_{[U_j^\ast, U_{j+1}^\ast)}(t) \right] \\
    &= \left( 1- \frac{1}{1 + \theta} \right) \sum_{j=1}^{n-1} \mathbb{E}\left[ \left( \varphi(U_{j+1}^\ast) - \varphi(U_j^\ast) \right) \left( 2 \varphi(U_j^\ast) + \varphi(U_{j+1}^\ast) - \varphi(U_j^\ast) \right) \mathds{1}_{[U_j^\ast, U_{j+1}^\ast)}(t) \right] \\
    &\quad + \frac{1}{1 + \theta} \sum_{j=1}^{n-1} \mathbb{E}\left[ \left( \varphi(U_{j+1}^\ast) - \varphi(U_j^\ast) \right) \mathds{1}_{[U_j^\ast, U_{j+1}^\ast)}(t) \right] \\
    &=\left( 1 - \frac{1}{ 1 + \theta} \right) \sum_{j=1}^{n-1} \mathbb{E}\left[ \left(\varphi(U_{j+1}^\ast)^2 - \varphi(U_j^\ast)^2 \right) \mathds{1}_{[U_j^\ast, U_{j+1}^\ast)}(t) \right] \\
    &\quad + \frac{1}{1 + \theta} \sum_{j=1}^{n-1} \mathbb{E}\left[ \left( \varphi(U_{j+1}^\ast) - \varphi(U_j^\ast) \right) \mathds{1}_{[U_j^\ast, U_{j+1}^\ast)}(t) \right].
\end{align*}
The results obtained in Lemma~\ref{lem_exp_mom_phi_BK} give the two following expressions: 
\begin{align*}
    &\sum_{j=1}^{n-1} \mathbb{E}\left[\varphi(U_{j+1}^\ast)^2 \mathds{1}_{[U_j^\ast, U_{j+1}^\ast)}(t) \right] - \sum_{j=1}^{n-1} \mathbb{E}\left[\varphi(U_j^\ast)^2 \mathds{1}_{[U_j^\ast, U_{j+1}^\ast)}(t)\right]\\
    &= -t^n + \varphi(t)^2 (1 - (1-t)^n) + 2 \int_t^1 \varphi'(v) \varphi(v) \left( (1-v+t)^n - (1-v)^n \right) \mathrm{d}v  \\
    &\quad - \left[ \varphi(t)^2 - \varphi(t)^2t^n  - 2 \int_0^t \varphi'(u) \varphi(u) \left( (1-t+u)^n - u^n \right) \mathrm{d}u \right] \\
    &= -t^n + \varphi(t)^2t^n - \varphi(t)^2 (1-t)^n + 2 \int_t^1 \varphi'(v) \varphi(v) ((1-v+t)^n - (1-v)^n) \mathrm{d}v \\
    &\quad + 2 \int_0^t \varphi'(u) \varphi(u) \left( (1-t+u)^n - u^n \right) \mathrm{d}u,
\end{align*}
and 
\begin{align*}
    &\sum_{j=1}^{n-1} \mathbb{E}\left[ \left( \varphi(U_{j+1}^\ast) - \varphi(U_j^\ast) \right) \mathds{1}_{[U_j^\ast, U_{j+1}^\ast)}(t) \right] \\
    &=  -t^n + \varphi(t) - (1-t)^n\varphi(t) + \int_t^1 \varphi'(v) \left( (1-v+t)^n - (1-v)^n \right) \mathrm{d}v \\
    &\quad -\left[ \varphi(t) - \varphi(t) t^n - \int_0^t \varphi'(u) \left( (1-t+u)^n -u^n \right) \mathrm{d}u \right] \\
    &= -t^n -(1-t)^n\varphi(t) + \int_t^1 \varphi'(v)((1-v+t)^n -(1-v)^n)\mathrm{d}v + \varphi(t) t^n \\
    &\quad + \int_0^t \varphi'(u) ((1-t+u)^n - u^n ) \mathrm{d}u.
\end{align*}
Consequently, 
\begin{align*}
     T_{3,n}^{\BK}(t;\theta) &\leq  -t^n + \varphi(t)^2t^n - \varphi(t)^2 (1-t)^n + 2 \int_t^1 \varphi'(v) \varphi(v) ((1-v+t)^n - (1-v)^n) \mathrm{d}v \\
    &\quad + 2 \int_0^t \varphi'(u) \varphi(u) \left( (1-t+u)^n - u^n \right) \mathrm{d}u \\
    &\quad + \frac{1}{1+\theta} \left( -t^n -(1-t)^n\varphi(t) + \int_t^1 \varphi'(v)((1-v+t)^n -(1-v)^n)\mathrm{d}v + \varphi(t) t^n  \right) \\
    &\quad + \frac{1}{1+\theta}\int_0^t \varphi'(u) ((1-t+u)^n - u^n ) \mathrm{d}u  \\
    & \quad - \frac{1}{1 + \theta}  \left( -t^n + \varphi(t)^2t^n - \varphi(t)^2 (1-t)^n\right)\\
    & \quad - \frac{2}{1  + \theta} \int_t^1 \varphi'(v) \varphi(v) ((1-v+t)^n - (1-v)^n) \mathrm{d}v \\
    &\quad -  2 \frac{1}{1+\theta} \int_0^t \varphi'(u) \varphi(u) \left( (1-t+u)^n - u^n \right) \mathrm{d}u \\
    &= -t^n + \varphi(t)^2t^n - \varphi(t)^2 (1-t)^n + 2 \int_t^1 \varphi'(v) \varphi(v) ((1-v+t)^n - (1-v)^n) \mathrm{d}v \\
    &\quad + 2 \int_0^t \varphi'(u) \varphi(u) \left( (1-t+u)^n - u^n \right) \mathrm{d}u \\
    &\quad +\frac{1}{1 + \theta} \left[ -(1-t)^n\varphi(t)(1 - \varphi(t)) + \int_t^1 \varphi'(v)((1-v+t)^n -(1-v)^n)\mathrm{d}v \right] \\
    &\quad + \frac{1}{1 + \theta} \left[ \varphi(t)(1 - \varphi(t)) t^n + \int_0^t \varphi'(u) ((1-t+u)^n - u^n ) \mathrm{d}u  \right] \\
    &\quad - \frac{2}{1 + \theta} \int_t^1 \varphi'(v) \varphi(v) ((1-v+t)^n - (1-v)^n) \mathrm{d}v \\
    &\quad - \frac{2}{1+\theta} \int_0^t \varphi'(u) \varphi(u) \left( (1-t+u)^n - u^n \right) \mathrm{d}u.
\end{align*}
Noting that the two terms 
\begin{equation*}
    - \frac{2}{1 + \theta} \int_t^1 \varphi'(v) \varphi(v) ((1+t-v)^n - (1-v)^n) \mathrm{d}v
\end{equation*} 
and 
\begin{equation*}
    - \frac{2}{1+\theta} \int_0^t \varphi'(u) \varphi(u) \left( (1-t+u)^n - u^n \right) \mathrm{d}u
\end{equation*} 
are non positive, we derive the inequality 
\begin{align*}
     T_{3,n}^{\BK}(t;\theta) &\leq  \varphi(t)^2t^n  + 2 \int_t^1 \varphi'(v) \varphi(v) ((1-v+t)^n - (1-v)^n) \mathrm{d}v \\
    &\quad + 2 \int_0^t \varphi'(u) \varphi(u) \left( (1-t+u)^n - u^n \right) \mathrm{d}u \\
    &\quad +\frac{1}{1 + \theta} \int_t^1 \varphi'(v)((1-v+t)^n -(1-v)^n)\mathrm{d}v \\
    &\quad + \frac{1}{1 + \theta} \left[ \varphi(t)(1 - \varphi(t)) t^n + \int_0^t \varphi'(u) ((1-t+u)^n - u^n ) \mathrm{d}u  \right] .
\end{align*}
We deduce that 
\begin{align*}
         T_{3,n}^{\BK}(t;\theta) &\leq  t^n  + 4 \|\varphi' \|_{\infty}\frac{1}{n+1} (1-t^{n+1} -(1-t)^{n+1}) \\
    &\quad + \frac{1}{1 + \theta}\frac{1}{4}t^n + 2 \|\varphi'\|_{\infty} \frac{1}{1 + \theta} \frac{1}{n+1} (1 - t^{n+1} - (1-t)^{n+1} ).
\end{align*}
We can conclude, as $n \to  \infty,$  
\[ T_{3,n}^{\BK}(t;\theta) = O\left( \frac{1}{n} \right).\]
Since we proved, as $n \to \infty,$ that $R_n^{\BK}(t;\theta) = O\left( \frac{1}{n} \right),$ we also get  
\[-2 \varphi(t) R_n^{\BK}(t;\theta) = O\left( \frac{1}{n} \right).\]
In addition, we can notice that  
\begin{align*}
    m_n(t)^2 &= \left( \varphi(t) t^n + \int_0^t \varphi'(u) \left( (1-t+u)^n - u^n \right) \mathrm{d}u \right)^2 \\
    &= \varphi(t)^2 t^{2n} + 2\varphi(t) t^n \int_0^t \varphi'(u) \left( (1-t+u)^n - u^n \right) \mathrm{d}u \\
    &\quad + \left( \int_0^t \varphi'(u) \left( (1-t+u)^n - u^n \right) \mathrm{d}u \right)^2 \\
    &\leq \varphi(t)^2 t^{2n} + 2 \varphi(t) t^n \|\varphi'\|_{\infty} \frac{1}{n+1} \left( 1 - t^{n+1} - (1-t)^{n+1} \right) \\
    & \quad + \left(\frac{1}{n+1}\right)^2 \|\varphi'\|_{\infty}^2 \left( 1 - t^{n+1} - (1-t)^{n+1} \right)^2.
\end{align*} 
We can conclude that, as $n \to  \infty,$  
\begin{align*}
    m_n(t)^2 = O\left( \frac{1}{n} \right).
\end{align*}
We can also derive the following expression  
\begin{align*}
    M_n(t)^2 &= \left( (1-t)^n (1-\varphi(t)) + \int_t^1 \varphi'(v) \left( (1-v+t)^n - (1-v)^n \right) \mathrm{d}v\right)^2 \\
    &= (1-t)^{2n} (1-\varphi(t))^2 \\
    &\quad + 2 (1-t)^n (1-\varphi(t)) \int_t^1 \varphi'(v) \left( (1-v+t)^n - (1-v)^n \right) \mathrm{d}v \\
    & \quad + \left( \int_t^1 \varphi'(v) \left( (1-v+t)^n - (1-v)^n \right) \mathrm{d}v \right)^2.
\end{align*} 
Then  
\begin{align*}
    M_n(t)^2
    &\leq (1-t)^{2n} (1 - \varphi(t))^2 + 2(1-t)^n (1-\varphi(t)) \| \varphi'\|_{\infty} \frac{1}{n+1} \left( 1- t^{n+1} - (1-t)^{n+1} \right) \\
    &\quad + \| \varphi' \|_{\infty}^2 \frac{1}{(n+1)^2} \left(1 - (1-t)^{n+1} -t^{n+1} \right)^2.
\end{align*}
We can conclude, as $n \to  \infty,$ that
\begin{align*}
    M_n(t)^2 = O\left( \frac{1}{n} \right).
\end{align*}
We deduce that, as $n \to  \infty,$ 
\[R_n^{\BK}(t;\theta)^2 = O\left( \frac{1}{n} \right).\]
Finally, we get, as $n \to  \infty,$  
\[\Var\left[w_n^{\BK}(t;\theta)\right] - \frac{1}{1 + \theta} \varphi(t) (1 - \varphi(t)) = O\left( \frac{1}{n} \right).\]

\section{Proof of Proposition~\ref{prop:warping} in CDF case and of Proposition~\ref{prop:CDF_comp_asymp_var_precise}}\label{appendix::proof_CDF}

In this appendix, we aim at proving \rev{Proposition~\ref{prop:warping} and Proposition~\ref{prop:CDF_comp_asymp_var_precise} }for the CDF Algorithm. This appendix is divided into \rev{four subsections corresponding respectively to the two statements of Proposition~\ref{prop:warping}, technical lemmas for the proof of Proposition~\ref{prop:CDF_comp_asymp_var_precise} and the proof of Proposition~\ref{prop:CDF_comp_asymp_var_precise}.} We recall that we denote by $w_n^\CDF(\cdot;\theta,p)$ the path obtained with the CDF Algorithm. 

\subsection{Proof of Proposition~\ref{prop:warping}(i) for CDF Algorithm} 

In this subsection, we prove the first statement of Proposition~\ref{prop:warping} for the CDF case:  we derive an explicit expression for the expectation of $w_n^\CDF(t;\theta,p)$ at every point $t \in (0,1)$ and obtain the asymptotic convergence, as $n \to \infty,$ of this term to $\varphi(t),$ with a bound on the associated convergence rate.

\subsubsection{Explicit expression of the expectation term}

Let $t \in (0,1).$ For all $j \in \llbracket 0, n \rrbracket,$ we set  
\begin{equation*}
    \rev{
    Q_j(t) = \frac{t - \varphi^{-1}(U_j^\ast)}{\varphi^{-1}(U^\ast_{j+1}) - \varphi^{-1}(U_j^\ast) }  
    } \; \; \text{ and } 
    B_j(t) = \left( \Tilde{\gamma}_{j,p} + \gamma_{j+1, p} \rev{ Q_j(t) } \right) \mathds{1}_{[ \varphi^{-1}(U_j^\ast), \varphi^{-1}(U^\ast_{j+1}) )}(t).
\end{equation*}   
Hence, the expectation of $w_n^{\CDF}(t;\theta,p)$ is expressed as 
\begin{align*}
    \mathbb{E}\left[ w_n^{\CDF}(t;\theta,p) \right]=  \sum_{j=0}^{n} \mathbb{E}\left[ B_j(t) \right] .
\end{align*}
Since the first and last terms of the partition $\{[\varphi^{-1}(U_j^\ast), \varphi^{-1}(U_{j+1}^\ast)] )\;; j \in \llbracket0, n\rrbracket \}$ only have one random endpoint, we perform the calculation of $\mathbb{E}\left[ B_0(t) \right]$ and $\mathbb{E}\left[ B_n(t) \right]$ separately. We have 
\begin{align*}
    \mathbb{E}\left[B_0(t)\right] &= \mathbb{E}\left[ \gamma_{1,p} \frac{t}{\varphi^{-1}(U_1^\ast)} \mathds{1}_{[0, \varphi^{-1}(U_1^\ast))}(t) \right] \\
    &= \mathbb{E}\left[ \gamma_{1,p}\right] \mathbb{E}\left[ \frac{t}{\varphi^{-1}(U_1^\ast)} \mathds{1}_{[0, \varphi^{-1}(U_1^\ast))}(t) \right]\\
    &= \frac{1-p}{n} \mathbb{E}\left[ \frac{t}{\varphi^{-1}(U_1^\ast)} \mathds{1}_{[0, \varphi^{-1}(U_1^\ast))}(t) \right] \quad  \text{(using Equation~\eqref{esp_gamma_1})}.
\end{align*}
Similarly, 
\begin{align*}
    \mathbb{E}\left[B_n(t)\right] &= \mathbb{E}\left[ \left(1 - p \beta_n \frac{1-t}{1 - \varphi^{-1}(U_n^\ast)} \right) \mathds{1}_{[\varphi^{-1}(U_n^\ast),1)}(t) \right] \\
    &= \mathbb{E}\left[ \left( 1 - p \mathbb{E}\left[ \beta_n \right] \frac{1 - t}{1 - \varphi^{-1}(U_n^\ast)} \right) \mathds{1}_{[\varphi^{-1}(U_n^\ast), 1)}(t) \right] \\
    &= \mathbb{E}\left[ \left(1 - \frac{p}{n} \frac{1 - t}{1 - \varphi^{-1}(U_n^\ast)} \right) \mathds{1}_{[\varphi^{-1}(U_n^\ast),1)}(t) \right] \quad  \text{(using Equation~\eqref{esp_beta})}.
\end{align*}
Since 
\begin{align*}
    &\sum_{j=1}^{n-1}\mathbb{E}\left[ B_j(t) \right] = \sum_{j=1}^{n-1} \mathbb{E}\left[ \mathbb{E}\left[ B_j(t) \; | \; U^\ast \right] \right],
\end{align*}
Equations~\eqref{esp_gamma_j} and \eqref{esp_tilde_gamma_j} lead to  
\begin{align*}
    & \sum_{j=1}^{n-1} \mathbb{E}\left[ B_j(t) \right] \\ 
    &= \sum_{j=1}^{n-1} \mathbb{E}\left[ \mathbb{E}\left[\Tilde{\gamma}_{j,p} \right] \mathds{1}_{[\varphi^{-1}(U_j^\ast), \varphi^{-1}(U_{j+1}^\ast))}(t) \right] 
    + \sum_{j=1}^{n-1} \mathbb{E}\left[ \mathbb{E}\left[\gamma_{j+1,p} \right] \rev{Q_j(t)} \mathds{1}_{[\varphi^{-1}(U_j^\ast), \varphi^{-1}(U_{j+1}^\ast))}(t) \right]\\
    &= \sum_{j=1}^{n-1} \frac{j-p}{n} \mathbb{E}\left[ \mathds{1}_{[\varphi^{-1}(U_j^\ast), \varphi^{-1}(U_{j+1}^\ast))}(t)  \right] 
    + \sum_{j=1}^{n-1} \frac{1}{n} \mathbb{E}\left[ \rev{ Q_j(t) } \mathds{1}_{[\varphi^{-1}(U_j^\ast), \varphi^{-1}(U_{j+1}^\ast))}(t) \right].
\end{align*}
We get the announced formula by combining the previous three equations. Using the expressions of the known densities of $U_0^\ast, U_n^\ast$ and of the pairs $(U_j^\ast, U_{j+1}^\ast),$ we can see that this formula is an explicit expression of the expectation of $w_n^{\CDF}(t;\theta,p).$ More precisely, using Lemma~\ref{lem:esp_ind_CDF}, we get 
\begin{align} \label{esp_explicite_CDF}
\begin{split}
\mathcal E_n^{\CDF}(t;p) &= \frac{1-p}{n}\int_0^1 \frac{t}{ \varphi^{-1}(u)} \mathds{1}_{[0, \varphi^{-1}(u))}(t) n (1-u)^{n-1} \mathrm{d}u  \\
&\quad \quad + \frac{1}{n} \left[ n \varphi(t) - n \varphi(t)^n - p \left( 1 - \varphi(t)^n - (1 - \varphi(t))^n \right) \right] \\
&\quad +\frac{1}{n} \sum_{j=1}^{n-1} \int_0^{\varphi(t)} \int_{\varphi(t)}^1 \frac{t - \varphi^{-1}(u)}{\varphi^{-1}(v) - \varphi^{-1}(u)} n!\frac{u^{j-1}}{(j-1)!}\frac{(1-v)^{n-j-1}}{(n-j-1)!} \mathrm{d}v \mathrm{d}u \\
&\quad + \int_{0}^{\varphi(t)} \left( 1 - \frac{p}{n} \frac{1-t}{1 - \varphi^{-1}(v)}\right) nv^{n-1}\mathrm{d}v.
\end{split}
\end{align}

\subsubsection{Asymptotic behaviour of the expectation term}

\textbf{Highlighting the term} $\varphi(t).$ \newline We previously showed that 
\begin{align*}
    & \sum_{j=1}^{n-1} \mathbb{E}\left[ B_j(t) \right] \\ 
    &= \frac{1}{n} \left[ n \varphi(t) - n \varphi(t)^n - p\left( 1 - \varphi(t)^n - (1-\varphi(t))^n \right) \right] + \frac{1}{n} \sum_{j=1}^{n-1} \mathbb{E}\left[ \rev{Q_j(t)} \mathds{1}_{[\varphi^{-1}(U_j^\ast), \varphi^{-1}(U_{j+1}^\ast))}(t) \right] \\
    &= \varphi(t) - \varphi(t)^n - \frac{p}{n} \left( 1 - \varphi(t)^n - (1 - \varphi(t))^n \right) + \frac{1}{n} \sum_{j=1}^{n-1} \mathbb{E}\left[ \rev{Q_j(t) } \mathds{1}_{[\varphi^{-1}(U_j^\ast), \varphi^{-1}(U_{j+1}^\ast))}(t) \right].
\end{align*}
Thus, we can rewrite $\mathbb{E}\left[w_n^{\CDF}(t;\theta,p) \right]$ as follows: 
\begin{align} \label{CDF_esp_expression} 
     \mathbb{E}\left[w_n^{\CDF}(t;\theta,p) \right] = \varphi(t) + \rev{R_n^{\CDF}(t;p)},
\end{align} 
where \rev{$R_n^{\CDF}(t;p)$} is given by the following expression   
\begin{align}
\begin{split}
    \rev{R_n^{\CDF}(t;p)} &= \rev{\mathbb{E}[B_0(t)] + \mathbb{E}[B_n(t)]} - \varphi(t)^n - \frac{p}{n} \left( 1 - \varphi(t)^n - (1 - \varphi(t))^n \right)\\
    & \qquad + \frac{1}{n} \sum_{j=1}^{n-1} \mathbb{E}\left[ \rev{Q_j(t)} \mathds{1}_{[ \varphi^{-1}(U_j^\ast), \varphi^{-1}(U_{j+1}^\ast))}(t) \right].
\end{split} \label{CDF_expression_Rn}
\end{align}

\textbf{Study of the term } $\rev{ \mathbf{\mathbb{E}[B_0(t)].}}$

Since 
\[ 
\mathbb{E}[B_0(t)] = \frac{1-p}{n} \mathbb{E}\left[ \frac{t}{\varphi^{-1}(U_1^\ast)} \mathds{1}_{[0, \varphi^{-1}(U_1^\ast))}(t) \right], 
\]
$\mathbb{E}[B_0(t)]$ is a positive term. Moreover, we have 
\begin{align*}
    \frac{1-p}{n} \mathbb{E}\left[ \frac{t}{\varphi^{-1}(U_1^\ast)} \mathds{1}_{[0, \varphi^{-1}(U_1^\ast))}(t) \right] &\leq \frac{1-p}{n} \mathbb{E}\left[ \mathds{1}_{[0, \varphi^{-1}(U_1^\ast))}(t) \right] \\
    &= \frac{1-p}{n} (1 - \varphi(t))^n.
\end{align*}
By definition of a warping function, we have $\varphi(t) \in (0,1)$ when $t$ belongs to $(0,1).$ Thus, as $n \to  \infty,$ 
\begin{equation*}
    \rev{
    \mathbb{E}[B_0(t)] = O\left( \frac{1}{n} \right).
    }
\end{equation*} 

\noindent \textbf{Study of the term: } $  \mathbf{- \varphi(t)^n - (p/n) ( 1 - \varphi(t)^n - (1 - \varphi(t))^n).}$ 

\rev{Equation~\eqref{esp_indicatrice} highlights that $ 1 - \varphi(t)^n - (1 - \varphi(t))^n $ is a positive term.} Hence, we have 
\begin{align*}
    \lvert - \varphi(t)^n - \frac{p}{n} \left( 1 - \varphi(t)^n - (1 - \varphi(t))^n \right) \rvert = \varphi(t)^n + \frac{p}{n} (1 - \varphi(t)^n + (1-\varphi(t))^n).
\end{align*} 
Consequently, as $n \to \infty,$ we get 
\[
- \varphi(t)^n - \frac{p}{n} \left( 1 - \varphi(t)^n - (1 - \varphi(t))^n \right) = O\left( \frac{1}{n} \right).
\] 
\textbf{Study of the term } $ \mathbf{ \sum_{j=1}^{n-1} \mathbb{E}[ \rev{Q_j(t)} \mathds{1}_{ [ \varphi^{-1}(U_j^\ast), \varphi^{-1}(U_{j+1}^\ast) )}(t) ] / n.}$ 

\rev{
Since $Q_j(t)$ is a positive term, we deduce that $\sum_{j=1}^{n-1} \mathbb{E}[ \rev{Q_j(t)} \mathds{1}_{ [ \varphi^{-1}(U_j^\ast), \varphi^{-1}(U_{j+1}^\ast) )}(t) ] / n $ is a positive term.
}
Moreover, using \rev{$Q_j(t) \leq 1$ and Equation~\eqref{esp_indicatrice},} we get  
\begin{align*}
    \frac{1}{n} \sum_{j=1}^{n-1} \mathbb{E}\left[ \rev{Q_j(t)} \mathds{1}_{[ \varphi^{-1}(U_j^\ast), \varphi^{-1}(U_{j+1}^\ast))}(t) \right] &\leq \frac{1}{n} \sum_{j=1}^{n-1} \mathbb{E}\left[\mathds{1}_{[ \varphi^{-1}(U_j^\ast), \varphi^{-1}(U_{j+1}^\ast))}(t) \right] \\
    &= \frac{1}{n} \left( 1 - \varphi(t)^n - (1-\varphi(t))^n\right).
\end{align*}
Thus, as $n \to \infty,$ we have  
\begin{align*}
   \frac{1}{n} \sum_{j=1}^{n-1} \mathbb{E}\left[ \rev{Q_j(t)} \mathds{1}_{[ \varphi^{-1}(U_j^\ast), \varphi^{-1}(U_{j+1}^\ast))}(t) \right] = O\left( \frac{1}{n} \right).
\end{align*} 

\noindent\textbf{Study of the term } $\mathbf{\mathbb{E}[B_n(t)].}$ 

This is a positive term. Moreover, we have 
\begin{align*}
    \mathbb{E}[B_n(t)] = \mathbb{E}\left[ \left(1 - \frac{p}{n} \frac{1 - t}{1 - \varphi^{-1}(U_n^\ast)} \right) \mathds{1}_{[\varphi^{-1}(U_n^\ast),1)}(t) \right] &\leq \mathbb{E}\left[\mathds{1}_{[\varphi^{-1}(U_n^\ast),1)}(t) \right]  \\
    &= \varphi(t)^n . 
\end{align*}
Hence, as $n \to \infty,$ 
\begin{align*}
    \rev{\mathbb{E}[B_n(t)] = O\left( \frac{1}{n} \right).}
\end{align*}
\rev{
Finally, we showed that }
\begin{equation*}
    \rev{R_n^\CDF(t;p) = O\left(\frac{1}{n} \right).}
\end{equation*}
\rev{\textbf{Conclusion.}} 
In conclusion, we have as $n \to \infty,$  
\begin{align*}
    \mathbb{E}\left[ w_n^{\CDF}(t;\theta,p)\right] = \varphi(t) + O\left( \frac{1}{n} \right).
\end{align*}

\subsection{Proof of Proposition~\ref{prop:warping}(ii) for CDF Algorithm}

In this subsection, we prove the second statement of Proposition~\ref{prop:warping} for the CDF case:  we derive an explicit expression for the variance of $w_n^\CDF(t;\theta,p)$ at every point $t \in (0,1)$ and get the asymptotic convergence, as $n \to \infty,$ of this term to $\varphi(t)(1 - \varphi(t)) / (1+\theta),$ with a bound on the associated convergence rate.

\subsubsection{Explicit expression of the variance term}

Let $t \in (0,1).$  Since $\{[\varphi^{-1}(U_j^\ast), \varphi^{-1}(U_{j+1}^\ast))\;;j \in \llbracket0, n\rrbracket \}$ is a partition of $[0,1),$ we get  \[\Var[w_n^{\CDF}(t;\theta,p)] = \sum_{j=0}^{n} \Var[B_j(t)].\]
We have  \begin{align*}
    \Var[B_0(t)] &= \mathbb{E}\left[ B_0(t)^2 \right] - \mathbb{E}\left[ B_0(t) \right]^2 \\
    &= \mathbb{E}\left[ \gamma_{1,p}^2 \left(\frac{t}{\varphi^{-1}(U_1^\ast)}\right)^2 \mathds{1}_{[0, \varphi^{-1}(U_1^\ast))}(t)\right] - \mathbb{E}\left[ \gamma_{1,p} \frac{t}{\varphi^{-1}(U_1^\ast)} \mathds{1}_{[0, \varphi^{-1}(U_1^\ast))}(t) \right]^2 .
\end{align*}
Using Equation~\eqref{mom_ord_2_gamma_1}, we deduce \begin{align*}
    &\mathbb{E}\left[ \gamma_{1,p}^2 \left(\frac{t}{\varphi^{-1}(U_1^\ast)}\right)^2 \mathds{1}_{[0, \varphi^{-1}(U_1^\ast))}(t)\right]\\
    &= \mathbb{E}\left[ \gamma_{1,p}^2 \right] \mathbb{E}\left[ \left(\frac{t}{\varphi^{-1}(U_1^\ast)} \right)^2 \mathds{1}_{[0, \varphi^{-1}(U_1^\ast))}(t) \right] \\
    &= (1-p)^2 \left( \frac{1}{n} \frac{1}{1 + \theta} + \frac{1}{n^2} \left(1 - \frac{1}{1+\theta} \right) \right) \mathbb{E}\left[ \frac{t^2}{\varphi^{-1}(U_1^\ast)^2} \mathds{1}_{[0, \varphi^{-1}(U_1^\ast))}(t) \right].
\end{align*}
From Equation~\eqref{esp_gamma_1}, we derive  \begin{align*}
    \mathbb{E}\left[ \gamma_{1,p} \frac{t}{\varphi^{-1}(U_1^\ast)} \mathds{1}_{[0, \varphi^{-1}(U_1^\ast))}(t)\right]^2 &= \mathbb{E}\left[ \gamma_{1,p} \right]^2 \mathbb{E}\left[ \frac{t}{\varphi^{-1}(U_1^\ast)} \mathds{1}_{[0, \varphi^{-1}(U_1^\ast))}(t)\right]^2  \\
    &= \left( \frac{1-p}{n} \right)^2 \mathbb{E}\left[ \frac{t}{\varphi^{-1}(U_1^\ast)} \mathds{1}_{[0, \varphi^{-1}(U_1^\ast))}(t)\right]^2.
\end{align*}
Thus, \begin{align*}
    \Var\left[B_0(t)\right] &= (1-p)^2 \frac{1}{n} \frac{1}{1 + \theta} \mathbb{E}\left[ \frac{t^2}{\varphi^{-1}(U_1^\ast)^2} \mathds{1}_{[0, \varphi^{-1}(U_1^\ast))}(t) \right] \\
    & \quad + \left( \frac{1-p}{n} \right)^2 \left(1 - \frac{1}{1 + \theta} \right) \mathbb{E}\left[ \frac{t^2}{\varphi^{-1}(U_1^\ast)^2} \mathds{1}_{[0, \varphi^{-1}(U_1^\ast))}(t) \right] \\
    & \quad - \left( \frac{1 - p}{n} \right)^2  \mathbb{E}\left[ \frac{t}{\varphi^{-1}(U_1^\ast)} \mathds{1}_{[0, \varphi^{-1}(U_1^\ast))}(t)\right]^2.
\end{align*}
We also have the following relation: \begin{align*}
    & \mathbb{E}\left[ B_n(t)^2 \right] \\
    &=\mathbb{E}\left[ \left( 1 - p \beta_n \frac{1 - t}{1 - \varphi^{-1}(U_n^\ast)} \right)^2 \mathds{1}_{[\varphi^{-1}(U_n^\ast), 1)}(t) \right] \\
    &= \mathbb{E}\left[ \left( 1 - 2 p \beta_n \frac{1-t}{1 - \varphi^{-1}(U_n^\ast)} + p^2 \beta_n^2 \left( \frac{1-t}{1 - \varphi^{-1}(U_n^\ast)} \right)^2 \right) \mathds{1}_{[\varphi^{-1}(U_n^\ast), 1)}(t) \right] \\
    &= \mathbb{E}\left[ \mathds{1}_{[\varphi^{-1}(U_n^\ast), 1)}(t) \right] -2p \mathbb{E}\left[ \beta_n \right] \mathbb{E}\left[ \frac{1-t}{1 - \varphi^{-1}(U_n^\ast)} \mathds{1}_{[\varphi^{-1}(U_n^\ast), 1)}(t) \right] \\
    & \quad  + p^2 \mathbb{E}\left[ \beta_n^2 \right] \mathbb{E}\left[ \left( \frac{1 - t}{1 - \varphi^{-1}(U_n^\ast)} \right)^2 \mathds{1}_{[\varphi^{-1}(U_n^\ast), 1)}(t) \right] .
    \end{align*} Using Equations~\eqref{esp_beta} and \eqref{mom_ord_2_beta}, we derive the following result:  \begin{align*}
     \mathbb{E}\left[ B_n(t)^2 \right]      &=  \mathbb{E}\left[ \mathds{1}_{[\varphi^{-1}(U_n^\ast), 1)}(t) \right] -2p \frac{1}{n} \mathbb{E}\left[ \frac{1-t}{1 - \varphi^{-1}(U_n^\ast)} \mathds{1}_{[\varphi^{-1}(U_n^\ast), 1)}(t) \right]\\
    & \quad + p^2 \left( \frac{1}{n} \frac{1}{1 +\theta} + \frac{1}{n^2} \left( 1 - \frac{1}{1 + \theta} \right) \right)  \mathbb{E}\left[ \left( \frac{1 - t}{1 - \varphi^{-1}(U_n^\ast)} \right)^2 \mathds{1}_{[\varphi^{-1}(U_n^\ast), 1)}(t) \right].
\end{align*}
\textbf{Remaining terms.} 

We transform the expression of $\mathbb{E}\left[B_j(t)^2\right],$ for all $j\in \llbracket 1, n-1 \rrbracket,$ as follows: \begin{align*}
    &\mathbb{E}\left[ B_j(t)^2 \right] \\
    &= \mathbb{E}\left[  \Tilde{\gamma}_{j,p}^2 \mathds{1}_{[\varphi^{-1}(U_j^\ast), \varphi^{-1}(U_{j+1}^\ast))}(t) \right]  + 2 \mathbb{E}\left[  \gamma_{j+1,p} \Tilde{\gamma}_{j,p} \rev{Q_j(t)} \mathds{1}_{[\varphi^{-1}(U_j^\ast), \varphi^{-1}(U_{j+1}^\ast))}(t)  \right] \\
    & \qquad + \mathbb{E}\left[ \gamma_{j+1,p}^2 \rev{Q_j(t)}^2 \mathds{1}_{[\varphi^{-1}(U_j^\ast), \varphi^{-1}(U_{j+1}^\ast))}(t) \right] \\
    &= \mathbb{E}\left[ \Tilde{\gamma}_{j,p}^2 \right] \mathbb{E}\left[ \mathds{1}_{[\varphi^{-1}(U_j^\ast), \varphi^{-1}(U_{j+1}^\ast))}(t) \right] 
    + 2 \mathbb{E}\left[ \gamma_{j+1,p} \Tilde{\gamma}_{j,p} \right] \mathbb{E}\left[ \rev{Q_j(t)} \mathds{1}_{[\varphi^{-1}(U_j^\ast), \varphi^{-1}(U_{j+1}^\ast))}(t) \right] \\
    & \qquad + \mathbb{E}\left[ \gamma_{j+1,p}^2 \right] \mathbb{E}\left[\rev{Q_j(t)}^2 \mathds{1}_{[\varphi^{-1}(U_j^\ast), \varphi^{-1}(U_{j+1}^\ast))}(t) \right]. 
    \end{align*}
Using Equations~\eqref{mom_ord_2_gamma_j}, \eqref{mom_ord_2_tilde_gamma_j} and \eqref{esp_cross_gamma_tilde_gamma} we get\begin{align*}
    &\sum_{j=1}^{n-1} \Var[B_j(t)] \\
    &=  \frac{1}{n} \frac{1}{ 1 + \theta} \sum_{j=1}^{n-1} (j-2p+p^2) \mathbb{E}\left[ \mathds{1}_{[\varphi^{-1}(U_j^\ast), \varphi^{-1}(U_{j+1}^\ast))}(t)\right] \\
    & \quad + \frac{1}{n^2 } \left(1 - \frac{1}{1 + \theta} \right)\sum_{j=1}^{n-1} (j-p)^2 \mathbb{E}\left[ \mathds{1}_{[\varphi^{-1}(U_j^\ast), \varphi^{-1}(U_{j+1}^\ast))}(t)\right] \\
    & \quad + \rev{2} \frac{1}{n} \frac{1}{1 + \theta} (p-p^2)\sum_{j=1}^{n-1}  \mathbb{E}\left[ \frac{t - \varphi^{-1}(U_j^\ast)}{\varphi^{-1}(U_{j+1}^\ast) - \varphi^{-1}(U_j^\ast)} \mathds{1}_{[\varphi^{-1}(U_j^\ast), \varphi^{-1}(U_{j+1}^\ast))}(t) \right] \\
    & \quad + \rev{2} \frac{1}{n^2} \left(1 - \frac{1}{1 + \theta} \right) \sum_{j=1}^{n-1} (j-p) \mathbb{E}\left[ \frac{t - \varphi^{-1}(U_j^\ast)}{\varphi^{-1}(U_{j+1}^\ast) - \varphi^{-1}(U_j^\ast)} \mathds{1}_{[\varphi^{-1}(U_j^\ast), \varphi^{-1}(U_{j+1}^\ast))}(t) \right] \\
    &  \quad  + \frac{1}{n} \frac{1}{1 + \theta} (1 - 2p + 2p^2) \sum_{j=1}^{n-1}\mathbb{E}\left[ \left( \frac{t - \varphi^{-1}(U_j^\ast)}{\varphi^{-1}(U_{j+1}^\ast) - \varphi^{-1}(U_j^\ast)} \right)^2 \mathds{1}_{[\varphi^{-1}(U_j^\ast), \varphi^{-1}(U_{j+1}^\ast))}(t) \right]\\
    & \quad + \frac{1}{n^2} \left(1 - \frac{1}{1 + \theta} \right) \sum_{j=1}^{n-1}  \mathbb{E}\left[ \left( \frac{t - \varphi^{-1}(U_j^\ast)}{\varphi^{-1}(U_{j+1}^\ast) - \varphi^{-1}(U_j^\ast)} \right)^2 \mathds{1}_{[\varphi^{-1}(U_j^\ast), \varphi^{-1}(U_{j+1}^\ast))}(t) \right] \\
    & \quad  - \sum_{j=1}^{n-1} \mathbb{E}\left[ B_j(t) \right]^2.
\end{align*}
Finally, we get the following expression for $\Var[w_n^{\CDF}(t;\theta,p)]: $
\begin{align*}
    &\Var[w_n^{\CDF}(t;\theta,p)] \\
    &= (1-p)^2 \frac{1}{n} \frac{1}{1 + \theta} \mathbb{E}\left[ \frac{t^2}{\varphi^{-1}(U_1^\ast)^2} \mathds{1}_{[0, \varphi^{-1}(U_1^\ast))}(t) \right] \\
    & \quad + \left( \frac{1-p}{n} \right)^2 \left(1 - \frac{1}{1 + \theta} \right) \mathbb{E}\left[ \frac{t^2}{\varphi^{-1}(U_1^\ast)^2} \mathds{1}_{[0, \varphi^{-1}(U_1^\ast))}(t) \right] \\
    & \quad +  \frac{1}{n} \frac{1}{ 1 + \theta} \sum_{j=1}^{n-1} (j-2p+p^2) \mathbb{E}\left[ \mathds{1}_{[\varphi^{-1}(U_j^\ast), \varphi^{-1}(U_{j+1}^\ast))}(t)\right] \\
    & \quad + \frac{1}{n^2 } \left(1 - \frac{1}{1 + \theta} \right)\sum_{j=1}^{n-1} (j-p)^2 \mathbb{E}\left[ \mathds{1}_{[\varphi^{-1}(U_j^\ast), \varphi^{-1}(U_{j+1}^\ast))}(t)\right] \\
    & \quad + \rev{2} \frac{1}{n} \frac{1}{1 + \theta} (p-p^2)\sum_{j=1}^{n-1}  \mathbb{E}\left[ \frac{t - \varphi^{-1}(U_j^\ast)}{\varphi^{-1}(U_{j+1}^\ast) - \varphi^{-1}(U_j^\ast)} \mathds{1}_{[\varphi^{-1}(U_j^\ast), \varphi^{-1}(U_{j+1}^\ast))}(t) \right] \\
    & \quad + \rev{2} \frac{1}{n^2} \left(1 - \frac{1}{1 + \theta} \right) \sum_{j=1}^{n-1} (j-p) \mathbb{E}\left[ \frac{t - \varphi^{-1}(U_j^\ast)}{\varphi^{-1}(U_{j+1}^\ast) - \varphi^{-1}(U_j^\ast)} \mathds{1}_{[\varphi^{-1}(U_j^\ast), \varphi^{-1}(U_{j+1}^\ast))}(t) \right] \\
    &  \quad  + \frac{1}{n} \frac{1}{1 + \theta} (1 - 2p + 2p^2) \sum_{j=1}^{n-1}\mathbb{E}\left[ \left( \frac{t - \varphi^{-1}(U_j^\ast)}{\varphi^{-1}(U_{j+1}^\ast) - \varphi^{-1}(U_j^\ast)} \right)^2 \mathds{1}_{[\varphi^{-1}(U_j^\ast), \varphi^{-1}(U_{j+1}^\ast))}(t) \right]\\
    & \quad + \frac{1}{n^2} \left(1 - \frac{1}{1 + \theta} \right) \sum_{j=1}^{n-1}  \mathbb{E}\left[ \left( \frac{t - \varphi^{-1}(U_j^\ast)}{\varphi^{-1}(U_{j+1}^\ast) - \varphi^{-1}(U_j^\ast)} \right)^2 \mathds{1}_{[\varphi^{-1}(U_j^\ast), \varphi^{-1}(U_{j+1}^\ast))}(t) \right] \\
    & \quad + \mathbb{E}\left[ \mathds{1}_{[\varphi^{-1}(U_n^\ast), 1)}(t) \right] -2p \frac{1}{n} \mathbb{E}\left[ \frac{1-t}{1 - \varphi^{-1}(U_n^\ast)} \mathds{1}_{[\varphi^{-1}(U_n^\ast), 1)}(t) \right] \\
    & \quad + p^2 \left( \frac{1}{n} \frac{1}{1 +\theta} + \frac{1}{n^2} \left( 1 - \frac{1}{1 + \theta} \right) \right)  \mathbb{E}\left[ \left( \frac{1 - t}{1 - \varphi^{-1}(U_n^\ast)} \right)^2 \mathds{1}_{[\varphi^{-1}(U_n^\ast), 1)}(t) \right]\\
    & \quad - \sum_{j=0}^{n} \mathbb{E}\left[ B_j(t) \right]^2.
\end{align*} This is an explicit expression that can be written as follows: \begin{align} \label{var_explicite_CDF}
    \begin{split}&\mathcal V_n^{\CDF}(t;\theta,p) \\
    &= (1-p)^2 \frac{1}{n} \frac{1}{1 + \theta} \int_{\varphi(t)}^1 \frac{t^2}{\varphi^{-1}(u)^2} n(1-u)^{n-1}\mathrm{d}u \\
    &\quad + \left( \frac{1-p}{n} \right)^2 \left(1 - \frac{1}{1 + \theta} \right) \int_{\varphi(t)}^1 \frac{t^2}{\varphi^{-1}(u)^2} n(1-u)^{n-1}\mathrm{d}u \\
    &\quad +  \frac{1}{n} \frac{1}{ 1 + \theta}\left[ n \varphi(t) - n \varphi(t)^n - (2p-p^2)\left(1 - \varphi(t)^n - (1-\varphi(t))^n\right) \right]\\
    & + \frac{1}{n^2 } \left(1 - \frac{1}{1 + \theta} \right)\left[ n(n-1) \varphi(t)^2 - (n-p)^2 \varphi(t)^n + n (1-2p) \varphi(t) + p^2 - p^2 (1 - \varphi(t))^n \right] \\
    & \quad + \rev{2} \frac{1}{n} \frac{1}{1 + \theta} (p-p^2)\sum_{j=1}^{n-1}  \int_0^{\varphi(t)} \int_{\varphi(t)}^1 \frac{t - \varphi^{-1}(u)}{\varphi^{-1}(v) - \varphi^{-1}(u)} n! \frac{u^{j-1}}{(j-1)!} \frac{(1-v)^{n-j-1}}{(n-j-1)!} \mathrm{d}v \mathrm{d}u \\
    & \quad + \rev{2} \frac{1}{n^2} \left(1 - \frac{1}{1 + \theta} \right) \sum_{j=1}^{n-1} (j-p) \int_0^{\varphi(t)} \int_{\varphi(t)}^1 \frac{t - \varphi^{-1}(u)}{\varphi^{-1}(v) - \varphi^{-1}(u)} n! \frac{u^{j-1}}{(j-1)!} \frac{(1-v)^{n-j-1}}{(n-j-1)!} \mathrm{d}v \mathrm{d}u \\
    &  \quad  + \frac{1}{n} \frac{1}{1 + \theta} (1 - 2p + 2p^2) \sum_{j=1}^{n-1}\int_0^{\varphi(t)} \int_{\varphi(t)}^1 \left( \frac{t - \varphi^{-1}(u)}{\varphi^{-1}(v) - \varphi^{-1}(u)} \right)^2 \frac{ n! u^{j-1} (1-v)^{n-j-1} }{(j-1)! (n-j-1)!} \mathrm{d}v \mathrm{d}u\\
    & \quad + \frac{1}{n^2} \left(1 - \frac{1}{1 + \theta} \right) \sum_{j=1}^{n-1} \int_0^{\varphi(t)} \int_{\varphi(t)}^1 \left( \frac{t - \varphi^{-1}(u)}{\varphi^{-1}(v) - \varphi^{-1}(u)} \right)^2 n! \frac{u^{j-1}}{(j-1)!} \frac{(1-v)^{n-j-1}}{(n-j-1)!}  \mathrm{d}v \mathrm{d}u \\
    & \quad + \varphi(t)^n -2p \frac{1}{n} \int_0^{\varphi(t)} \frac{1-t}{1 - \varphi^{-1}(v)} nv^{n-1} \mathrm{d}v \\
    & \quad + p^2 \left( \frac{1}{n} \frac{1}{1 +\theta} + \frac{1}{n^2} \left( 1 - \frac{1}{1 + \theta} \right) \right)  \int_0^{\varphi(t)} \left( \frac{1 - t}{1 - \varphi^{-1}(v)} \right)^2 n v^{n-1} \mathrm{d}v\\
    & - \left( \frac{1-p}{n}\int_{\varphi(t)}^1 \frac{t}{ \varphi^{-1}(u)} \mathds{1}_{[0, \varphi^{-1}(u))}(t) n (1-u)^{n-1} \mathrm{d}u\right)^2 \\
    & - \left( \frac{1}{n} \left[ n \varphi(t) - n \varphi(t)^n - p \left( 1 - \varphi(t)^n - (1 - \varphi(t))^n \right) \right]\right. \\
    &\quad \left. +\frac{1}{n} \sum_{j=1}^{n-1} \int_0^{\varphi(t)} \int_{\varphi(t)}^1 \frac{t - \varphi^{-1}(u)}{\varphi^{-1}(v) - \varphi^{-1}(u)} n!\frac{u^{j-1}}{(j-1)!}\frac{(1-v)^{n-j-1}}{(n-j-1)!} \mathrm{d}v \mathrm{d}u \right)^2  \\
    & - \left( \int_{0}^{\varphi(t)} \left( 1 - \frac{p}{n} \frac{1-t}{1 - \varphi^{-1}(v)}\right) nv^{n-1}\mathrm{d}v \right)^2.\end{split}
\end{align}

\subsubsection{Asymptotic behaviour of the variance term}

\textbf{Highlighting the term } $ \varphi(t) (1 - \varphi(t))/(1 + \theta).$ 

We have  
\begin{align*}
    \mathbb{E}\left[w_n^{\CDF}(t;\theta,p) \right]^2 &= (\varphi(t) + \rev{R_n^{\CDF}(t;p)})^2 = \varphi(t)^2 \rev{+} 2 \rev{R_n^{\CDF}(t;p)} \varphi(t) + \rev{R_n^{\CDF}(t;p)}^2.
\end{align*} 
Using the previous calculations, we can remark that  
\begin{align*}
    &\sum_{j=1}^{n-1} \mathbb{E}\left[ B_j(t)^2 \right] \\
    &= \sum_{j=1}^{n-1}\mathbb{E}\left[ \Tilde{\gamma}_{j,p}^2 \right] \mathbb{E}\left[ \mathds{1}_{[\varphi^{-1}(U_j^\ast), \varphi^{-1}(U_{j+1}^\ast))}(t) \right] \\
    &+ 2 \sum_{j=1}^{n-1} \mathbb{E}\left[ \gamma_{j+1,p} \Tilde{\gamma}_{j,p} \right] \mathbb{E}\left[ \rev{Q_j(t)} \mathds{1}_{[\varphi^{-1}(U_j^\ast), \varphi^{-1}(U_{j+1}^\ast))}(t) \right] \\
    & \qquad + \sum_{j=1}^{n-1} \mathbb{E}\left[ \gamma_{j+1,p}^2 \right] \mathbb{E}\left[\rev{Q_j(t)}^2 \mathds{1}_{[\varphi^{-1}(U_j^\ast), \varphi^{-1}(U_{j+1}^\ast))}(t) \right] \\
    &=  \frac{1}{n} \frac{1}{ 1 + \theta} \sum_{j=1}^{n-1} (j-2p+p^2) \mathbb{E}\left[ \mathds{1}_{[\varphi^{-1}(U_j^\ast), \varphi^{-1}(U_{j+1}^\ast))}(t)\right] \\
    & \quad + \frac{1}{n^2 } \left(1 - \frac{1}{1 + \theta} \right)\sum_{j=1}^{n-1} (j-p)^2 \mathbb{E}\left[ \mathds{1}_{[\varphi^{-1}(U_j^\ast), \varphi^{-1}(U_{j+1}^\ast))}(t)\right] \\
    &\quad + 2 \sum_{j=1}^{n-1} \mathbb{E}\left[ \gamma_{j+1,p} \Tilde{\gamma}_{j,p} \right] \mathbb{E}\left[ \rev{Q_j(t)} \mathds{1}_{[\varphi^{-1}(U_j^\ast), \varphi^{-1}(U_{j+1}^\ast))}(t) \right] \\
    & \quad + \sum_{j=1}^{n-1} \mathbb{E}\left[ \gamma_{j+1,p}^2 \right] \mathbb{E}\left[ \rev{Q_j(t)}^2 \mathds{1}_{[\varphi^{-1}(U_j^\ast), \varphi^{-1}(U_{j+1}^\ast))}(t) \right]. \end{align*} 
    Then,  \begin{align*}
    &\sum_{j=1}^{n-1} \mathbb{E}\left[ B_j(t)^2 \right] \\
    &= \frac{1}{n} \frac{1}{ 1 + \theta}[ n \varphi(t) - n \varphi(t)^n - (2p-p^2)\left(1 - \varphi(t)^n - (1-\varphi(t))^n\right) ]\\
    &  + \frac{1}{n^2 } \left(1 - \frac{1}{1 + \theta} \right)[ n(n-1) \varphi(t)^2 - (n-p)^2 \varphi(t)^n + n (1-2p) \varphi(t) + p^2 - p^2 (1 - \varphi(t))^n ] \\
    &\quad + 2 \sum_{j=1}^{n-1} \mathbb{E}[ \gamma_{j+1,p} \Tilde{\gamma}_{j,p}] \mathbb{E}\left[ \rev{Q_j(t)} \mathds{1}_{[\varphi^{-1}(U_j^\ast), \varphi^{-1}(U_{j+1}^\ast))}(t) \right] \\
    & \quad + \sum_{j=1}^{n-1} \mathbb{E}[ \gamma_{j+1,p}^2] \mathbb{E}\left[ \rev{Q_j(t)}^2 \mathds{1}_{[\varphi^{-1}(U_j^\ast), \varphi^{-1}(U_{j+1}^\ast))}(t) \right].
\end{align*}
We can rewrite this term as
\begin{align*}
    &\sum_{j=1}^{n-1} \mathbb{E}\left[ B_j(t)^2 \right] \\
    &= \frac{1}{1 + \theta} \varphi(t) + \frac{1}{n^2} \left(1 - \frac{1}{1 + \theta} \right) [n(n-1) \varphi(t)^2 + n (1 - 2p) \varphi(t)] \\
    & \qquad + \frac{1}{n} \frac{1}{1+\theta} [- n \varphi(t)^n - (2p-p^2)\left(1 - \varphi(t)^n - (1-\varphi(t))^n\right)] \\
    & \qquad + \frac{1}{n^2} \left( 1 - \frac{1}{1 + \theta} \right) [-(n-p)^2 \varphi(t)^n + p^2 - p^2 (1 - \varphi(t))^n] \\
    &\qquad + 2 \sum_{j=1}^{n-1} \mathbb{E}[ \gamma_{j+1,p} \Tilde{\gamma}_{j,p}] \mathbb{E}\left[ \rev{Q_j(t)} \mathds{1}_{[\varphi^{-1}(U_j^\ast), \varphi^{-1}(U_{j+1}^\ast))}(t) \right] \\
    & \qquad + \sum_{j=1}^{n-1} \mathbb{E}[ \gamma_{j+1,p}^2] \mathbb{E}\left[ \rev{Q_j(t)}^2 \mathds{1}_{[\varphi^{-1}(U_j^\ast), \varphi^{-1}(U_{j+1}^\ast))}(t) \right],
\end{align*}
and then, 
\begin{align*}
    &\sum_{j=1}^{n-1} \mathbb{E}\left[ B_j(t)^2 \right] \\
    &= \rev{ \frac{1}{1 + \theta} \varphi(t) + \varphi(t)^2 - \frac{1}{n} \varphi(t)^2 + \frac{1}{n} \varphi(t) - \frac{2p}{n} \varphi(t) - \frac{1}{ 1 + \theta} \varphi(t)^2 } \\
    & \qquad \rev{ + \frac{1}{n} \frac{1}{1+\theta} \left[ - n \varphi(t)^n -(2p-p^2)(1 - \varphi(t)^n - (1 - \varphi(t))^n) \right] }\\
    & \qquad \rev{ + \frac{1}{n^2} \left[ - (n-p)^2 \varphi(t)^n + p^2 - p^2 (1 - \varphi(t))^n \right] }\\
    & \qquad \rev{ - \frac{1}{n^2} \frac{1}{1+\theta} \left[ -n \varphi(t)^2 - (n-p)^2 \varphi(t)^n + n \varphi(t) -2np \varphi(t) + p^2 - p^2 (1 - \varphi(t))^n \right] } \\
    &\qquad \rev{ + 2 \sum_{j=1}^{n-1} \mathbb{E}[ \gamma_{j+1,p} \Tilde{\gamma}_{j,p}] \mathbb{E}\left[ Q_j(t) \mathds{1}_{[\varphi^{-1}(U_j^\ast), \varphi^{-1}(U_{j+1}^\ast))}(t) \right] } \\
    & \qquad \rev{ + \sum_{j=1}^{n-1} \mathbb{E}[ \gamma_{j+1,p}^2] \mathbb{E}\left[ Q_j(t)^2 \mathds{1}_{[\varphi^{-1}(U_j^\ast), \varphi^{-1}(U_{j+1}^\ast))}(t) \right].}
\end{align*}

Thus, \begin{align*} 
    \Var[w_n^{\CDF}(t;\theta,p)] &= \rev{ \left(\frac{1}{1+\theta} + \frac{1}{n} \right)\varphi(t)(1 - \varphi(t)) - \frac{2p}{n} \varphi(t) - 2R_n^\CDF(t;p)\varphi(t) + \Tilde{R}_n^\CDF(t;\theta,p), }
\end{align*} 
where 
\begin{align}
\begin{split}
    \Tilde{R}_n^{\text{CDF}}(t;\theta,p) &= \frac{1}{1 + \theta} \left[ - \varphi(t)^n - \frac{2p-p^2}{n} \left(1 - \varphi(t)^n - (1 - \varphi(t))^n\right) \right] \\
    &\quad \rev{+\frac{1}{n^2}\left(1- \frac{1}{1 + \theta} \right) \left[ -(n-p)^2 \varphi(t)^n + p^2 - p^2 ( 1 - \varphi(t))^n \right] }\\
    & \quad \rev{ - \frac{1}{n^2} \frac{1}{1 + \theta}\left[ - n \varphi(t)^2 + n(1-2p) \varphi(t)  \right] }  \\
    &\quad + 2 \sum_{j=1}^{n-1} \mathbb{E}\left[ \gamma_{j+1,p} \Tilde{\gamma}_{j,p} \right] \mathbb{E}\left[ \rev{Q_j(t)} \mathds{1}_{[\varphi^{-1}(U_j^\ast), \varphi^{-1}(U_{j+1}^\ast))}(t) \right] \\
    & \quad + \sum_{j=1}^{n-1} \mathbb{E}\left[ \gamma_{j+1,p}^2 \right] \mathbb{E}\left[ \rev{Q_j(t)}^2 \mathds{1}_{[\varphi^{-1}(U_j^\ast), \varphi^{-1}(U_{j+1}^\ast))}(t) \right] \\
    &\quad + \mathbb{E}\left[B_0(t)^2 \right] + \mathbb{E}\left[B_n(t)^2\right] \rev{ - R_n^{\CDF}(t;p)^2.} \end{split} \label{tilde_R_n_expression}
\end{align} 

\noindent\rev{\textbf{Term} $\mathbf{{R_n^{\CDF}(t;p)^2}}.$ }
Since $R_n^\CDF(t;p) = O\left( \frac{1}{n} \right),$ we get $R_n^\CDF(t;p)^2 = o\left( \frac{1}{n} \right). $ \newline 

\noindent\rev{\textbf{Term} $\mathbf{ (1/n) \varphi(t)(1 - \varphi(t)) - (2p/n) \varphi(t) - 2R_n^\CDF(t;p)\varphi(t)}.$ \newline 
Since $R_n(t) = O\left( \frac{1}{n}\right),$ we get } 
\begin{equation*}
    \rev{ \frac{1}{n} \varphi(t)(1 - \varphi(t)) - \frac{2p}{n} \varphi(t) - 2R_n^\CDF(t;p)\varphi(t) = O\left( \frac{1}{n} \right). }
\end{equation*}

\noindent\textbf{Term} $\mathbf{\mathbb{E}\left[B_0(t)^2 \right].}$ \newline 
\rev{
Since $B_0(t)$ belongs to $[0,1],$ we get 
\[
\mathbb{E}[B_0(t)^2] \leq \mathbb{E}[B_0(t)] \leq (1 - \varphi(t))^n.
\]
}
Thus, as $n \to \infty,$  
\[
\mathbb{E}\left[B_0(t)^2\right] = O\left( \frac{1}{n} \right).
\]
\textbf{Term} $\mathbf{\mathbb{E}\left[B_n(t)^2\right].}$  \newline 
\rev{
Since $B_n(t)$ belongs to $[0,1],$ we get 
\[
\mathbb{E}[B_n(t)^2] \leq \mathbb{E}[B_n(t)] \leq \varphi(t)^n.
\]
}
We deduce that, as $n \to  \infty,$  
\[
\mathbb{E}\left[ B_n(t)^2 \right]  = O\left( \frac{1}{n} \right).
\]
\textbf{Term} $\displaystyle (1/(1 + \theta)) \left[ - \varphi(t)^n - \frac{2p-p^2}{n} \left(1 - \varphi(t)^n - (1 - \varphi(t))^n\right) \right].$ \newline 
\rev{Since $2p-p^2, 1 - \varphi(t)^n - (1-\varphi(t))^n$ and $\varphi(t)^n$ are positive terms, we have}
\begin{align*}
    &\left|\frac{1}{1 + \theta} \left[ - \varphi(t)^n - \frac{2p-p^2}{n} \left(1 - \varphi(t)^n - (1 - \varphi(t))^n\right) \right]\right|\\
    &\rev{=} \frac{1}{1 + \theta} \varphi(t)^n + \frac{2p-p^2}{n} \frac{1}{1 + \theta} \left( 1 - \varphi(t)^n - (1 - \varphi(t))^n \right).
\end{align*}
We deduce that, as $n \to \infty,$ 
\[
\frac{1}{1 + \theta} \left[ - \varphi(t)^n - \frac{2p-p^2}{n} \left(1 - \varphi(t)^n - (1 - \varphi(t))^n\right) \right] = O\left( \frac{1}{n} \right).
\] 

\noindent\textbf{Term} \rev{$n^{-2}\left(1- (1 + \theta)^{-1} \right) \left[ -(n-p)^2 \varphi(t)^n + p^2 - p^2 ( 1 - \varphi(t))^n \right].$} \newline
\rev{
We note that, as $n \to \infty,$ 
}
\begin{equation*}
   \rev{\frac{1}{n^2}\left(1- \frac{1}{1 + \theta} \right) \left[ -(n-p)^2 \varphi(t)^n + p^2 - p^2 ( 1 - \varphi(t))^n \right] = O\left( \frac{1}{n} \right). }
\end{equation*}

\noindent\textbf{Term} \rev{$- n^{-2} (1 + \theta)^{-1}\left[ - n \varphi(t)^2 + n(1-2p) \varphi(t)  \right].$}

\rev{We remark that, as $n \to \infty,$ }
\begin{equation*}
   \rev{ - \frac{1}{n^2} \frac{1}{1 + \theta}\left[ - n \varphi(t)^2 + n(1-2p) \varphi(t)  \right] = O\left( \frac{1}{n} \right).}
\end{equation*} 

\noindent \rev{\textbf{Remaining terms.} } \newline 
\rev{Let $S_n^{\CDF}(t;\theta,p)$ be the quantity} 
\begin{align*}
    &2 \sum_{j=1}^{n-1} \mathbb{E}[ \gamma_{j+1,p} \Tilde{\gamma}_{j,p}] \mathbb{E}[ \rev{Q_j(t)} \mathds{1}_{[\varphi^{-1}(U_j^\ast), \varphi^{-1}(U_{j+1}^\ast))}(t)] \\
    & \qquad +\sum_{j=1}^{n-1} \mathbb{E}\left[ \gamma_{j+1,p}^2 \right] \mathbb{E}\left[ \rev{Q_j(t)}^2 \mathds{1}_{[\varphi^{-1}(U_j^\ast), \varphi^{-1}(U_{j+1}^\ast))}(t) \right].
\end{align*}
\rev{ Combining $Q_j(t) \in [0,1]$ and Equation~\eqref{esp_cross_gamma_tilde_gamma} we get}
\begin{align*}
    &2 \sum_{j=1}^{n-1} \mathbb{E}\left[ \gamma_{j+1,p} \Tilde{\gamma}_{j,p} \right] \mathbb{E}\left[ \rev{Q_j(t)} \mathds{1}_{[\varphi^{-1}(U_j^\ast), \varphi^{-1}(U_{j+1}^\ast))}(t) \right] \\
    &\rev{ \leq } 2 \sum_{j=1}^{n-1} \left( \frac{1}{n} \frac{1}{1 + \theta} p(1-p) + \frac{1}{n^2} \left(1 - \frac{1}{1 + \theta} \right)(j-p) \right) \mathbb{E}\left[ \mathds{1}_{[\varphi^{-1}(U_j^\ast), \varphi^{-1}(U_{j+1}^\ast))}(t) \right].
\end{align*}  
\rev{ From $Q_j(t)^2 \in [0,1]$ and Equation~\eqref{mom_ord_2_gamma_j}, we also have}
\begin{align*}
&\sum_{j=1}^{n-1} \mathbb{E}\left[ \gamma_{j+1,p}^2 \right] \mathbb{E}\left[ \rev{Q_j(t)}^2 \mathds{1}_{[\varphi^{-1}(U_j^\ast), \varphi^{-1}(U_{j+1}^\ast))}(t) \right] \\
&\rev{\leq} \sum_{j=1}^{n-1} \left( \frac{1}{n} \frac{1}{1 + \theta} (1-2p+2p^2) + \frac{1}{n^2} \left( 1 - \frac{1}{ 1 + \theta} \right) \right) \mathbb{E}\left[\mathds{1}_{[\varphi^{-1}(U_j^\ast), \varphi^{-1}(U_{j+1}^\ast))}(t) \right] .
\end{align*}
Then \begin{align*}
    S_n^{\CDF}(t;\theta,p) &\rev{\leq} \frac{1}{n} \frac{1}{1+\theta} \sum_{j=1}^{n-1} \mathbb{E}\left[ \mathds{1}_{[\varphi^{-1}(U_j^\ast), \varphi^{-1}(U_{j+1}^\ast))}(t) \right] \\
    &\quad + \rev{2}\frac{1}{n^2} \left( 1 - \frac{1}{1 + \theta} \right) \sum_{j=1}^{n-1} \rev{(j -(p-\frac{1}{2}))} \mathbb{E}\left[ \mathds{1}_{[\varphi^{-1}(U_j^\ast), \varphi^{-1}(U_{j+1}^\ast))}(t) \right].
\end{align*} 
Using Equation~\eqref{esp_indicatrice} and Equation~\eqref{j-k_esp}, we get  
\begin{align*}
     S_n^{\CDF}(t;\theta,p) &\rev{\leq} \frac{1}{n} \frac{1}{1 + \theta} \left( 1 - \varphi(t)^n - (1 - \varphi(t))^n \right) \\
     &\quad + \rev{2} \frac{1}{n^2} \left(1 - \frac{1}{1 + \theta} \right) \left[ n \varphi(t) - n \varphi(t)^n - \rev{ (p-\frac{1}{2}) } \left(1 - \varphi(t)^n - (1-\varphi(t))^n\right) \right].
\end{align*}
Consequently, as $n \to \infty,$  
\[
     S_n^{\CDF}(t;\theta,p) = O\left( \frac{1}{n} \right).
\]
\rev{\textbf{Conclusion.}}

Finally, we have the following result: \[ \Var[w_n^{\CDF}(t;\theta,p)] - \frac{1}{1 + \theta} \varphi(t) (1 - \varphi(t)) = O\left( \frac{1}{n} \right) .\]

\subsection{Technical lemmas for the proof of Proposition~\ref{prop:CDF_comp_asymp_var_precise}.}
\rev{This subsection aims at presenting useful lemmas for the proof of Proposition~\ref{prop:CDF_comp_asymp_var_precise}. Let $a \in (0,1/2).$ Let $t \in (a, 1-a). $ We start by introducing some notation. Let}

\rev{\begin{align*}
    & C_n(t) = \frac{1}{n} \sum_{k=1}^{n-1} (k-p) \mathds{1}_{[\varphi^{-1}(U_k^\ast), \varphi^{-1}(U_{k+1}^\ast))}(t), \\
    & D_n(t) = \frac{1}{n} \sum_{j = 1}^{n-1} Q_j(t) \mathds{1}_{[\varphi^{-1}(U_j^\ast), \varphi^{-1}(U_{j+1}^\ast))}(t).
\end{align*}}

\rev{\begin{lem}\label{lem:remarks}
Let $a \in (0,1/2)$ and $\varphi \in \Gamma.$ Let $i \in \{1,2\}.$ As $n \xrightarrow{} \infty,$ we have for any $k\ge 1$
\begin{align*}
 \underset{t \in (a, 1-a)}{\sup} \varphi(t)^n &= o\left( \frac{1}{n^k} \right)
%, \text{ for k }\in \mathbb{N}^\ast 
\\
 \underset{t \in (a, 1-a)}{\sup} (1 - \varphi(t))^n &= o \left( \frac{1}{n^k} \right).
%, \text{ for k } \in \mathbb{N}^\ast.
\end{align*}
In particular, 
\begin{align*}
& \underset{t \in (a, 1-a)}{\sup} \mathbb{E}[B_0(t)^i] = o\left( \frac{1}{n^k} \right)
%, \text{ for k }\in \mathbb{N}^\ast 
\\
& \underset{t \in (a, 1-a)}{\sup} \mathbb{E}[B_n(t)^i] = o \left( \frac{1}{n^k} \right).
%    , \text{ for k } \in \mathbb{N}^\ast.
\end{align*}
\end{lem}}

\rev{\begin{proof}
Since $\varphi$ is a warping function, $t \in (a, 1-a)$ leads to $\varphi(t) \in (0,1)$ and the first two results directly follow. Since $B_0(t), B_n(t)$ belong to $[0,1],$ we have 
\begin{align*}
    \mathbb{E}[B_0(t)^2] \leq \mathbb{E}[B_0(t)] &\leq \mathbb{E}[\mathds{1}_{[0, \varphi^{-1}(U_1^\ast) )}(t)] \\
    &= (1 - \varphi(t))^n.
\end{align*}
and
\begin{align*}
    \mathbb{E}[B_n(t)^2] \leq \mathbb{E}[B_n(t)] &\leq \mathbb{E}[\mathds{1}_{[\varphi^{-1}(U_n^\ast), 1)}(t)] \\
    &= \varphi(t)^n,
\end{align*}
which end the proof.
\end{proof}}

\rev{Lemma~\ref{lem:lemme_du_lemme_cov} gives an asymptotic result for the covariance of $C_n(t)$ and $D_n(t).$
\begin{lem} \label{lem:lemme_du_lemme_cov}
Let $a \in (0, 1/2).$ As $ n \xrightarrow{} \infty,$ we have 
\begin{equation*}
    \underset{t \in (a, 1-a)}{\sup}\left|\Cov[C_n(t), D_n(t)] \right|  = o\left( \frac{1}{n} \right)
\end{equation*}
\end{lem}}

\rev{\begin{proof}
Let $a \in (0, 1/2)$ and $t \in (a, 1-a).$ We begin with the study of the term $\Var[C_n(t)].$ On the one hand, Equation~\eqref{(j-k)^2_esp} gives
\begin{align*}
    & \mathbb{E}[C_n(t)^2]\\
    &= \frac{1}{n^2} \mathbb{E}\left[ \left( \sum_{k=1}^{n-1} (k-p) \mathds{1}_{[\varphi^{-1}(U_k^\ast), \varphi^{-1}(U_{k+1}^\ast))}(t) \right) \left( \sum_{l=1}^{n-1} (l-p) \mathds{1}_{[\varphi^{-1}(U_l^\ast), \varphi^{-1}(U_{l+1}^\ast))}(t) \right) \right] \\
    &= \frac{1}{n^2} \mathbb{E}\left[ \sum_{k=1}^{n-1} (k-p)^2 \mathds{1}_{[\varphi^{-1}(U_k^\ast), \varphi^{-1}(U_{k+1}^\ast))}(t) \right] \\
    &= \frac{1}{n^2} \left[ n(n-1) \varphi(t)^2 - (n-p)^2 \varphi(t)^n + n(1-2p) \varphi(t) + p^2 - p^2(1 - \varphi(t))^n \right] \\
    &= \varphi(t)^2 - \frac{\varphi(t)^2}{n} + \frac{1-2p}{n} \varphi(t) + o \left( \frac{1}{n} \right).
\end{align*}
On the other hand, using Equation~\eqref{j-k_esp}, we have 
\begin{align*}
    \mathbb{E}[C_n(t)] &= \frac{1}{n} \sum_{k=1}^{n-1} (k-p) \mathbb{E}[\mathds{1}_{[\varphi^{-1}(U_k^\ast), \varphi^{-1}(U_{k+1}^\ast))}(t)] \\
    &= \frac{1}{n} \left[ n \varphi(t) - n \varphi(t)^n - p (1 - \varphi(t)^n - (1-\varphi(t))^n ) \right] \\
    &= \varphi(t) - \varphi(t)^n - \frac{p}{n} \left[ 1 - \varphi(t)^n - (1 - \varphi(t))^n \right] \\
    &= \varphi(t) - \frac{p}{n} + o \left(\frac{1}{n} \right).
\end{align*}
Thus, we get 
\begin{align*}
    \mathbb{E}[C_n(t)]^2 = \varphi(t)^2 - 2p \frac{\varphi(t)}{n} + o \left( \frac{1}{n} \right),
\end{align*}
leading to 
\begin{align*}
    \Var[C_n(t)] = \frac{1}{n} \varphi(t) (1 - \varphi(t)) + o \left( \frac{1}{n} \right).
\end{align*}
Now, we focus on $D_n(t).$ We remark that $D_n(t)$ is positive and bounded by $1/n.$ Thus 
\begin{align*}
    \left| D_n(t) - \mathbb{E}[D_n(t)] \right| &\leq |D_n(t)| + | \mathbb{E}[D_n(t)] | \\
    &= D_n(t) + \mathbb{E}[D_n(t)] \\
    &\leq \frac{1}{n} + \frac{1}{n}.
\end{align*}
Consequently 
\begin{align*}
    \left| \Cov[C_n(t), D_n(t)] \right| & \leq \mathbb{E}[|C_n(t) - \mathbb{E}[C_n(t)]| |D_n(t) - \mathbb{E}[D_n(t)]| ] \\
    & \leq \frac{2}{n} \mathbb{E}[|C_n(t) - \mathbb{E}[C_n(t)]| ] \\
    & \leq \frac{2}{n} \Var[C_n(t)]^{1/2} \text{ using Cauchy-Schwarz inequality} \\
    &= \frac{2}{n} \left( \frac{1}{n} \varphi(t) (1 - \varphi(t)) + o \left(\frac{1}{n} \right) \right)^{1/2} \\
    &= O\left( \frac{1}{n^{3/2}} \right).
\end{align*}
We have proved that, for all $t \in (a, 1-a),$ as $n \xrightarrow{} \infty$
\[
\left|\Cov[C_n(t), D_n(t)] \right| = o \left( \frac{1}{n} \right).
\]
\end{proof}}

\rev{\begin{lem} \label{lem:term_in_precise_var}
Let $a \in (0, 1/2).$ As $ n \xrightarrow{} \infty,$ we have  
\begin{equation*} 
    \underset{t \in (a, 1-a)}{\sup}\left|\frac{2}{n^2} \sum_{j = 1}^{n-1} (j-p) \mathbb{E}[Q_j(t) \mathds{1}_{[\varphi^{-1}(U_j^\ast), \varphi^{-1}(U_{j+1}^\ast))}(t) ] - 2 \varphi(t)  \mathbb{E}[D_n(t)] \right|= o \left( \frac{1}{n} \right).
\end{equation*}
\end{lem}}

\rev{\begin{proof}
First, remark that 
\begin{align*}
    2 \varphi(t)  \mathbb{E}[D_n(t)] & = \frac{2}{n} (n\varphi(t)) \mathbb{E}[D_n(t)].
\end{align*}
Using Equation~\eqref{j-k_esp}, we get 
\begin{align*}
    n \varphi(t) = \mathbb{E}\left[ \sum_{k=1}^{n-1} (k-p) \mathds{1}_{[\varphi^{-1}(U_k^\ast), \varphi^{-1}(U_{k+1}^\ast))}(t) \right] + n \varphi(t)^n + p [1 - \varphi(t)^n - (1 - \varphi(t))^n].
\end{align*}
This equation can be rewritten as: 
\begin{align*}
    n \varphi(t) = n \mathbb{E}[C_n(t)] + n \varphi(t)^n + p [1 - \varphi(t)^n - (1 - \varphi(t))^n].
\end{align*}
Combining the two previous equations lead to
\begin{align*}
    & 2 \varphi(t) \mathbb{E}[D_n(t)] \\
    &= \frac{2}{n} \mathbb{E}[D_n(t)] n \mathbb{E}[C_n(t)] \\
    & \qquad  + \frac{2}{n} \left(n \varphi(t)^n + p (1 - \varphi(t)^n - (1 - \varphi(t))^n)\right)  \mathbb{E}[D_n(t)] \\
    &= 2  \mathbb{E}[D_n(t) C_n(t)] - 2 \Cov[D_n(t), C_n(t)] \\
    & \qquad + \frac{2}{n} \left(n \varphi(t)^n + p (1 - \varphi(t)^n - (1 - \varphi(t))^n)\right)  \mathbb{E}[D_n(t)]. 
\end{align*}
Since $\left\{[\varphi^{-1}(U_j^\ast), \varphi^{-1}(U_{j+1}^\ast)) ; j \in \llbracket0 ,n \rrbracket \right\}$ forms a partition of $[0,1),$ we have the following: 
\begin{align*}
    \mathbb{E}[C_n(t) D_n(t)] 
    &= \frac{1}{n^2} \mathbb{E}\left[ \sum_{j=1}^{n-1} Q_j(t) \mathds{1}_{[\varphi^{-1}(U_j^\ast), \varphi^{-1}(U_{j+1}^\ast))}(t) \sum_{k=1}^{n-1} (k-p) \mathds{1}_{[\varphi^{-1}(U_k^\ast), \varphi^{-1}(U_{k+1}^\ast))}(t) \right] \\
    &= \frac{1}{n^2} \mathbb{E}\left[ \sum_{j=1}^{n-1} \sum_{k=1}^{n-1} (k-p) Q_j(t) \mathds{1}_{[\varphi^{-1}(U_j^\ast), \varphi^{-1}(U_{j+1}^\ast))}(t)  \mathds{1}_{[\varphi^{-1}(U_k^\ast), \varphi^{-1}(U_{k+1}^\ast))}(t)    \right] \\
    &= \frac{1}{n^2} \mathbb{E}\left[ \sum_{j=1}^{n-1} (j-p) Q_j(t) \mathds{1}_{[\varphi^{-1}(U_j^\ast), \varphi^{-1}(U_{j+1}^\ast))}(t)\right].
\end{align*}
Hence, 
\begin{align*}
    2 \varphi(t) \mathbb{E}[D_n(t)] 
    &= \frac{2}{n^2} \mathbb{E}\left[ \sum_{j=1}^{n-1} (j-p) Q_j(t) \mathds{1}_{[\varphi^{-1}(U_j^\ast), \varphi^{-1}(U_{j+1}^\ast))}(t)\right] - 2 \Cov[C_n(t), D_n(t)] \\
    & \qquad + \frac{2}{n} \left(n \varphi(t)^n + p (1 - \varphi(t)^n - (1 - \varphi(t))^n)\right)  \mathbb{E}[D_n(t)].
\end{align*} 
Consequently, we get 
\begin{align*}
    & \frac{2}{n^2} \sum_{j = 1}^{n-1} (j-p) \mathbb{E}[Q_j(t) \mathds{1}_{[\varphi^{-1}(U_j^\ast), \varphi^{-1}(U_{j+1}^\ast))}(t) ] - 2 \varphi(t) \mathbb{E}[D_n(t)] \\
    &= 2  \Cov[C_n(t), D_n(t)] \\
    & \qquad - \frac{2}{n} \left(n \varphi(t)^n + p (1 - \varphi(t)^n - (1 - \varphi(t))^n)\right) \mathbb{E}[D_n(t)]. \\
    &= 2 \Cov[C_n(t), D_n(t)] \\
    &\qquad - \frac{2}{n^2} \left(n \varphi(t)^n + p (1 - \varphi(t)^n - (1 - \varphi(t))^n)\right)  \sum_{j=1}^{n-1} \mathbb{E}[Q_j(t)  \mathds{1}_{[\varphi^{-1}(U_j^\ast), \varphi^{-1}(U_{j+1}^\ast))}(t)]
\end{align*}
Using $ \left(n \varphi(t)^n + p (1 - \varphi(t)^n - (1 - \varphi(t))^n)\right)  \sum_{j=1}^{n-1} \mathbb{E}[Q_j(t)  \mathds{1}_{[\varphi^{-1}(U_j^\ast), \varphi^{-1}(U_{j+1}^\ast))}(t)] = O(1)$ and Lemma~\ref{lem:lemme_du_lemme_cov} end the proof.
\end{proof}}

\subsection{Proof of Proposition~\ref{prop:CDF_comp_asymp_var_precise}}
\rev{Now, we prove Proposition~\ref{prop:CDF_comp_asymp_var_precise}.}

\rev{Let $a \in (0, 1/2)$ and $t \in (a, 1-a)$. First, we rewrite the exact expression of $2 \varphi(t) R_n^\CDF(t;p).$ Using Equation~\eqref{CDF_expression_Rn}, we get:
\begin{align*}
    &  2 \varphi(t) R_n^\CDF(t;p) \\
    &= 2\varphi(t) \mathbb{E}[B_0(t)] + 2\varphi(t) \mathbb{E}[B_n(t)] - 2\varphi(t)^{n+1} - 2\frac{p}{n} \varphi(t) \left( 1 - \varphi(t)^n - (1 - \varphi(t))^n \right)  \\
    & \qquad + 2\frac{\varphi(t)}{n} \sum_{j=1}^{n-1} \mathbb{E}[Q_j(t) \mathds{1}_{[\varphi^{-1}(U_j^\ast), \varphi^{-1}(U_{j+1}^\ast))}(t)].
\end{align*}
Thus,
\begin{align*}
    &  - \frac{2p}{n} \varphi(t) - 2 R_n^\CDF(t;p) \varphi(t)  \\
    &=  - 2\varphi(t) \mathbb{E}[B_0(t)] - 2\varphi(t) \mathbb{E}[B_n(t)] + 2\varphi(t)^{n+1} - 2\frac{p}{n} \varphi(t) \left( \varphi(t)^n + (1 - \varphi(t))^n \right)  \\
    & \qquad  - 2\varphi(t) \mathbb{E}[D_n(t)]. 
\end{align*}
Consequently, 
\begin{align*}
    &  \left| \Var[w_n^{\CDF}(t;\theta,p)] - \left(\frac{1}{1+\theta} + \frac{1}{n} \right)\varphi(t)(1 - \varphi(t)) \right| \\
    &= \left| - \frac{2p}{n} \varphi(t) - 2R_n^\CDF(t;p)\varphi(t) + \Tilde{R}_{n,\theta,p}^{\text{CDF}}(t) \right|  \\
    &= \Big|  - 2\varphi(t) \mathbb{E}[B_0(t)] - 2\varphi(t) \mathbb{E}[B_n(t)] + 2\varphi(t)^{n+1} - 2\frac{p}{n} \varphi(t) \left( \varphi(t)^n + (1 - \varphi(t))^n \right)  \\
    &\phantom{=}  \qquad - 2\varphi(t) \mathbb{E}[D_n(t)] + \Tilde{R}_n^\CDF(t;\theta,p) \Big|.
\end{align*}
Using Lemma~\ref{lem:remarks}, we can write: 
\begin{align*}
    \underset{t \in (a, 1-a)}{\sup}
    \Big| - 2 \varphi(t) \mathbb{E}[B_0(t)] - 2 \varphi(t) \mathbb{E}[B_n(t)] + 2 \varphi(t)^{n+1} - 2 \frac{p}{n} \varphi(t) (\varphi(t)^n + &(1 - \varphi(t))^n) \Big| \\
    &= o\left(\frac{1}{n} \right). 
\end{align*} 
Thus, 
\begin{align*}
     & \underset{t \in (a, 1-a)}{\sup} \left| \Var[w_n^{\CDF}(t;\theta,p)] - \left(\frac{1}{1+\theta} + \frac{1}{n} \right)\varphi(t)(1 - \varphi(t)) \right| \\
     &\leq \underset{t \in (a, 1-a)}{\sup} \left| \tilde{R}_n^\CDF(t ; \theta, p) - 2 \varphi(t) \mathbb{E}[D_n(t)] \right|  + o\left(\frac{1}{n} \right)
\end{align*}
Now, from the definition of $\Tilde{R}_n^\CDF(t;\theta,p)$ given by Equation~\eqref{tilde_R_n_expression} and Lemma~\ref{lem:remarks}, we can write: 
\begin{align*}
    & \underset{t \in (a, 1-a)}{\sup} \left| \Tilde{R}_n^\CDF(t;\theta,p) - 2 \varphi(t) \mathbb{E}[D_n(t)] \right| \\
    &\leq \underset{t \in (a, 1-a)}{\sup} \left| 2 \sum_{j=1}^{n-1} \mathbb{E}[\gamma_{j+1,p} \tilde{\gamma}_{j,p}] \mathbb{E}[Q_j(t) \mathds{1}_{[\varphi^{-1}(U_j^\ast), \varphi^{-1}(U_{j+1}^\ast))}(t) ] - 2 \varphi(t)\mathbb{E}[D_n(t)] \right| \\
    & \qquad + \underset{t \in (a, 1-a)}{\sup} \sum_{j=1}^{n-1} \mathbb{E}[\gamma_{j+1,p}^2] \mathbb{E}[Q_j(t)^2 \mathds{1}_{[\varphi^{-1}(U_j^\ast), \varphi^{-1}(U_{j+1}^\ast))}(t)] \\
    & \qquad + \underset{t \in (a, 1-a)}{\sup} R_n^\CDF(t;p)^2 + o\left( \frac{1}{n} \right) + O\left( \frac{1}{n \theta} \right). 
\end{align*}
From Equation~\eqref{mom_ord_2_gamma_j}, we get: 
\begin{align*}
    &\sum_{j=1}^{n-1} \mathbb{E}[\gamma_{j+1,p}^2] \mathbb{E}[Q_j(t)^2 \mathds{1}_{[\varphi^{-1}(U_j^\ast), \varphi^{-1}(U_{j+1}^\ast))}(t)] \\
    &= \sum_{j=1}^{n-1} \frac{1}{n} \frac{1}{1 + \theta} (1 - (2p-2p^2)) \mathbb{E}[Q_j(t)^2 \mathds{1}_{[\varphi^{-1}(U_j^\ast), \varphi^{-1}(U_{j+1}^\ast))}(t)] \\
    & \qquad + \frac{1}{n^2}\left( 1 - \frac{1}{1 + \theta} \right) \sum_{j=1}^{n-1} \mathbb{E}[Q_j(t)^2 \mathds{1}_{[\varphi^{-1}(U_j^\ast), \varphi^{-1}(U_{j+1}^\ast))}(t)].
\end{align*}
Since $\sum_{j=1}^{n-1} \mathbb{E}[Q_j(t)^2 \mathds{1}_{[\varphi^{-1}(U_j^\ast), \varphi^{-1}(U_{j+1}^\ast))}(t)] $ is bounded by $1,$ we deduce that 
\begin{equation*}
    \underset{t \in (a, 1-a)}{\sup} \sum_{j=1}^{n-1} \mathbb{E}[\gamma_{j+1,p}^2] \mathbb{E}[Q_j(t)^2 \mathds{1}_{[\varphi^{-1}(U_j^\ast), \varphi^{-1}(U_{j+1}^\ast))}(t)] = O\left( \frac{1}{n \theta} \right) + o \left( \frac{1}{n} \right).
\end{equation*}
Moreover, Equation~\eqref{esp_cross_gamma_tilde_gamma} leads to 
\begin{align*}
    & 2 \sum_{j=1}^{n-1} \mathbb{E}\left[ \gamma_{j+1,p} \Tilde{\gamma}_{j,p} \right] \mathbb{E}\left[ Q_j(t) \mathds{1}_{[\varphi^{-1}(U_j^\ast), \varphi^{-1}(U_{j+1}^\ast))}(t) \right]  \\
    &= 2 \frac{1}{n} \frac{1}{1 + \theta} (p - p^2) \sum_{j=1}^{n-1} \mathbb{E}[Q_j(t) \mathds{1}_{[\varphi^{-1}(U_j^\ast), \varphi^{-1}(U_{j+1}^\ast))}(t) ] \\
    & \qquad + \frac{2}{n^2} \left( 1 - \frac{1}{1 + \theta} \right) \sum_{j=1}^{n-1} (j-p) \mathbb{E}[Q_j(t) \mathds{1}_{[\varphi^{-1}(U_j^\ast), \varphi^{-1}(U_{j+1}^\ast))}(t)]  \\
    &= \frac{2}{n^2} \sum_{j=1}^{n-1} (j-p)\mathbb{E}[Q_j(t) \mathds{1}_{[\varphi^{-1}(U_j^\ast), \varphi^{-1}(U_{j+1}^\ast))}(t)] + O\left( \frac{1}{n \theta} \right) + o \left( \frac{1}{n} \right).
\end{align*}
We continue with
\begin{align*}
    & \underset{t \in (a, 1-a)}{\sup} \left| \tilde{R}_n^\CDF(t;\theta,p) - 2 \varphi(t) \mathbb{E}[D_n(t)] \right| \\
    &\leq \underset{t \in (a, 1-a)}{\sup} \left|\frac{2}{n^2} \sum_{j=1}^{n-1} (j-p)\mathbb{E}[Q_j(t) \mathds{1}_{[\varphi^{-1}(U_j^\ast), \varphi^{-1}(U_{j+1}^\ast))}(t)]  - 2 \varphi(t)\mathbb{E}[D_n(t)] \right| \\
    & \qquad + \underset{t \in (a, 1-a)}{\sup} R_n^\CDF(t;p)^2 + o\left( \frac{1}{n} \right) + O\left( \frac{1}{n \theta} \right). 
\end{align*}
Using that, for all $t \in [0,1],\;R_n^\CDF(t;p) = O\left( \frac{1}{n} \right),$ we obtain 
\[
\underset{t \in (a, 1-a)}{\sup} R_n^\CDF(t;p)^2 = o \left( \frac{1}{n} \right).
\]
Hence, 
\begin{align*}
    & \underset{t \in (a, 1-a)}{\sup}\left| \Var[w_n^{\CDF}(t;\theta,p)] - \left(\frac{1}{1+\theta} + \frac{1}{n} \right)\varphi(t)(1 - \varphi(t)) \right| \\
    & \leq \left|  \frac{2}{n^2} \sum_{j=1}^{n-1} (j-p)\mathbb{E}[Q_j(t) \mathds{1}_{[\varphi^{-1}(U_j^\ast), \varphi^{-1}(U_{j+1}^\ast))}(t)]  - 2\varphi(t) \mathbb{E}[D_n(t)]  \right|  + o\left( \frac{1}{n} \right) + O\left( \frac{1}{n \theta} \right).
\end{align*}
Lemma~\ref{lem:term_in_precise_var} leads to the result.}

\section{Reminder and additional results for MZW Algorithm} \label{app:MZW}

\subsection{Inner product space structure}\label{rappel_MZW}

We begin this appendix by recalling the main sets involved in the MZW Algorithm. \rev{We remind that $\Gamma$ introduced in Equation~\eqref{expression_Gamma_set} stands for the following set: 
\[ \Gamma = \{ \gamma:[0,1] \to [0,1] \mid \gamma(0) = 0, \gamma(1) =1, \gamma \text{ is a strictly increasing continuous function} \}.\]}
For given functions $f,g$ belonging to $\Gamma,$ we can notice that the difference $f -g$ is generally not in $\Gamma.$ Consequently, \citet{ma2024stochastic} consider the more restrictive set denoted by $\Gamma_1$ given by
\[\Gamma_1~=\{\gamma \in \Gamma \cap \mathcal{D}([0,1])\;|\; \exists\;(m_\gamma, M_\gamma)\in (\mathbb{R}_+^\ast)^2, 0< m_\gamma< \gamma^\prime <M_\gamma <\infty\}, \] 
where we recall that $ \mathcal{D}([0,1])$ refers to as the set of functions that admits a derivative at each point of the line segment $[0,1]$ and on which it is possible to add an inner product space structure. We recall the definitions of the two operations $\oplus_{\Gamma_1}$ and $\odot_{\Gamma_1}$ such that $(\Gamma_1, \oplus_{\Gamma_1}, \odot_{\Gamma_1})$ owns a vector space structure and the inner product $\langle \cdot, \cdot \rangle_{\Gamma_1}$ on $\Gamma_1.$ \begin{definition}\citep[Definition 1]{ma2024stochastic} Let $f,g \in \Gamma_1$ and $\alpha \in \mathbb{R},$ the perturbation operation with operator $\oplus_{\Gamma_1}:\Gamma_1 \times \Gamma_1 \to \Gamma_1$ is given by \[ [f \oplus_{\Gamma_1} g] = \frac{\int_0^\cdot f^\prime(s)g^\prime(s) \mathrm{d}s}{\int_0^1 f^\prime(\tau)g^\prime(\tau) \mathrm{d}\tau}. 
\] The power operation with operator $\odot_{\Gamma_1}:\mathbb{R} \times \Gamma_1 \to \Gamma_1$ is given by \[[\alpha \odot_{\Gamma_1} f] = \frac{\int_0^\cdot f^\prime(s)^\alpha \mathrm{d}s}{\int_0^1 f^\prime(\tau)^\alpha \mathrm{d}\tau}. \]
\end{definition}

In particular, we notice that the identity element for the perturbation operation is given by the identity function and the inverse element of a given $f\in \Gamma_1$ is given by $(-1) \odot_{\Gamma_1} f$, that can be rewritten as $(\int_0^\cdot (1/f^\prime(s))\mathrm{d}s) / (\int_0^1 (1 / f^\prime(\tau)) \mathrm{d}\tau).$ We then naturally introduce the following notation. Let $f,g \in \Gamma_1.$ We denote by $f \ominus_{\Gamma_1} g$ the element of $\Gamma_1$ given by $f \oplus_{\Gamma_1} [(-1) \odot_{\Gamma_1} g].$ Finally, we recall the inner product on $\Gamma_1.$ 

\begin{definition} \citep[Definition 2]{ma2024stochastic} Let $f,g \in \Gamma_1,$ the inner product is defined as the functional $\langle \cdot, \cdot \rangle_{\Gamma_1}: \Gamma_1 \times \Gamma_1 \to \mathbb{R}$ of the following form: 
\[ \langle f, g \rangle_{\Gamma_1} = \int_0^1 \log(f^\prime(t)) \log(g^\prime(t)) \mathrm{d}t - \int_0^1 \log(f^\prime(s)) \mathrm{d}s \int_0^1 \log(g^\prime(t)) \mathrm{d}t. \] 
\end{definition}

\subsection{Proof of Proposition~\ref{prop_MZW}}

We begin by recalling some useful notation. The function $\varphi$ belongs to the set $\Gamma_1$ defined in Appendix~\ref{rappel_MZW}. Let $m\ge 0$. The family of functions $(\phi_i)_{i \geq 1}$ denotes the Fourier basis with the constant function equal to 1 removed. The set $H(0,1)$ is defined as $H(0,1)= \{ h \in L^2([0,1]) \;;\; \int_0^1 h = 0, -\infty < m_h < h < M_h < \infty \}.$ The application $\psi$ is an isometric isomorphism between $ \Gamma_1$ and $H(0,1)$ defined as follows $\psi \;:\; \varphi \in \Gamma_1 \mapsto \psi(\varphi) =  \log(\varphi') - \int_0^1 \log(\varphi') \in H(0,1).$ For a fixed $m,$ $X_m(\cdot;\theta)$ is the function $ \psi(\varphi) + \sum_{i=1}^m G_i \phi_i,$ where the coefficients $G_i$ follow a centered stochastic process with variance term given by $v_i / (1 + \theta),$ for $i \in \llbracket 1,m \rrbracket.$ The expression of the simulated path is given by $w_m^{\MZW}(\cdot;\theta) = (\int_0^\cdot \exp{(X_m(\cdot;\theta))}) / (\int_0^1 \exp{(X_m(\cdot;\theta))})$ which corresponds to $\psi^{-1}(X_m(\cdot;\theta)).$

\begin{proof}

Let $\Check{\varphi} \in \Gamma_1$ and $t \in (0,1).$ The expression of $X_m(\cdot;\theta)$ and the definition of $w_m^{\MZW}(\cdot;\theta)$ lead to \[ w_m^{\MZW}(\cdot;\theta) = \frac{\int_0^\cdot \varphi' \exp{(\sum_{i=1}^m G_i \phi_i)}}{ \int_0^1 \varphi' \exp{(\sum_{i=1}^m G_i \phi_i)}}.\] In particular, the derivative of $w_m^{\MZW}(\cdot;\theta)$ can be expressed as follows: \[ (w_m^{\MZW})'(\cdot;\theta) = \frac{\varphi' \exp{(\sum_{i=1}^m G_i \phi_i)}}{\int_0^1 \varphi' \exp{(\sum_{i=1}^m G_i \phi_i)}}.\] Consequently, 
\[ w_m^{\MZW}(\cdot;\theta) \ominus_{\Gamma_1} \Check{\varphi} = \frac{\int_0^\cdot \exp{(\sum_{i=1}^m G_i \phi_i)} \frac{\varphi'}{\Check{\varphi}'}}{\int_0^1 \exp{(\sum_{i=1}^m G_i \phi_i)} \frac{\varphi'}{\Check{\varphi}'}}. \] 
This expression leads to 
\[ \log((w_m^\MZW(\cdot;\theta) \ominus_{\Gamma_1} \Check{\varphi} )') = \sum_{i=1}^m G_i \phi_i + \log \left( \frac{\varphi'}{\Check{\varphi}'} \right) - \log\left( \int_0^1 \exp{(\sum_{i=1}^m G_i \phi_i) \frac{\varphi'}{\Check{\varphi}'}} \right) \]
and 
\begin{align*}\log((w_m^\MZW(\cdot;\theta) \ominus_{\Gamma_1} \Check{\varphi} )')^2 &= \left( \sum_{i=1}^m G_i \phi_i \right)^2 + \log\left(  \frac{\varphi'}{\Check{\varphi}'}\right)^2 + \log\left( \int_0^1 \exp{(\sum_{i=1}^m G_i \phi_i) } \frac{\varphi'}{\Check{\varphi}'} \right)^2 \\
& + 2 \sum_{i=1}^m G_i \phi_i \log\left( \frac{\varphi'}{\Check{\varphi}'} \right) - 2 \sum_{i=1}^m G_i \phi_i \log\left( \int_0^1\exp{(\sum_{i=1}^m G_i \phi_i)} \frac{\varphi'}{\Check{\varphi}'}\right) \\
& - 2 \log\left( \frac{\varphi'}{\Check{\varphi}'} \right) \log\left( \int_0^1 \exp{(\sum_{i=1}^m G_i \phi_i)} \frac{\varphi'}{\Check{\varphi}'} \right).
\end{align*} 
Using $\int_0^1 \phi_i = 0,$ for $i \in \llbracket1,m\rrbracket,$ we get 
\[ \int_0^1 \log((w_m^\MZW(\cdot;\theta) \ominus_{\Gamma_1} \Check{\varphi} )') = \int_0^1 \log\left( \frac{\varphi'}{\Check{\varphi}'} \right) - \log\left( \int_0^1 \exp{(\sum_{i=1}^m G_i \phi_i)} \frac{\varphi'}{\Check{\varphi}'}\right)  \] 
and 
\begin{align*}
\int_0^1 \log((w_m^\MZW(\cdot;\theta) \ominus_{\Gamma_1} \Check{\varphi} )')^2 &= \int_0^1 ( \sum_{i=1}^m G_i \phi_i)^2 + \int_0^1 \log\left( \frac{\varphi'}{\Check{\varphi}'} \right)^2 \\
&+ \log\left( \int_0^1 \exp{\left( \sum_{i=1}^m G_i \phi_i\right)} \frac{\varphi'}{ \Check{\varphi}'}\right)^2 \\
& - 2 \log\left( \int_0^1 \exp{\left( \sum_{i=1}^m G_i \phi_i \right)} \frac{\varphi'}{\Check{\varphi}'} \right) \int_0^1 \log\left( \frac{\varphi'}{\Check{\varphi}'} \right).
\end{align*}
Then we deduce 
\begin{align*}
&\mathbb{E}\left[ \int_0^1 \log((w_m^\MZW(\cdot;\theta) \ominus_{\Gamma_1} \Check{\varphi} )')^2  \right]\\
&= \sum_{i=1}^m \Var[G_i] + \int_0^1 \log\left( \frac{\varphi'}{\Check{\varphi}'} \right)^2 + \mathbb{E}\left[ \log\left( \int_0^1 \exp{\left( \sum_{i=1}^m G_i \phi_i\right)} \frac{\varphi'}{\Check{\varphi}'} \right)^2 \right] \\
& \quad  - 2 \int_0^1 \log\left( \frac{\varphi'}{\Check{\varphi}'} \right) \mathbb{E}\left[ \log\left(  \int_0^1 \exp{\left( \sum_{i=1}^m G_i \phi_i\right)} \frac{\varphi'}{\Check{\varphi}'}\right) \right]
\end{align*} 
and 
\begin{align*}
&\mathbb{E}\left[ \left(\int_0^1 \log((w_m^\MZW(\cdot;\theta) \ominus_{\Gamma_1} \Check{\varphi} )') \right)^2  \right] \\
&= \left( \int_0^1 \log\left( \frac{\varphi'}{\Check{\varphi}'} \right) \right)^2 - 2 \left( \int_0^1 \log\left( \frac{\varphi'}{\Check{\varphi}'} \right) \right) \mathbb{E}\left[ \log\left(  \int_0^1 \exp{\left( \sum_{i=1}^m G_i \phi_i\right)} \frac{\varphi'}{\Check{\varphi}'}\right) \right] \\
&+ \mathbb{E}\left[ \log\left( \int_0^1 \exp{\left( \sum_{i=1}^m G_i \phi_i \right)} \frac{\varphi'}{\Check{\varphi}'} \right)^2 \right].
\end{align*} Since \begin{equation*} \| w_m^\MZW(\cdot;\theta) \ominus_{\Gamma_1} \Check{\varphi} \|_{\Gamma_1}^2 = \int_0^1 \log\left( (w_m^\MZW(\cdot;\theta) \ominus_{\Gamma_1} \Check{\varphi})' \right)^2 - \left( \int_0^1 \log\left((w_m^\MZW(\cdot;\theta) \ominus_{\Gamma_1} \Check{\varphi})' \right) \right) ^2, \end{equation*} we obtain 
\begin{align*}
\mathbb{E}\left[ \| w_m^\MZW(\cdot;\theta) \ominus_{\Gamma_1} \Check{\varphi} \|_{\Gamma_1}^2\right] &= \sum_{i=1}^m \Var[G_i] + \int_0^1 \log\left( \frac{\varphi'}{\Check{\varphi}'}\right)^2 - \left( \int_0^1 \log\left( \frac{\varphi'}{\Check{\varphi}'} \right) \right)^2.
\end{align*} 
Now, we remark the following: 
\begin{align*}
\int_0^1 \log\left( \frac{\varphi'}{\Check{\varphi}'} \right)^2 &= \int_0^1 \left( \log(\varphi') - \log(\Check{\varphi}) \right)^2 \\
&= \int_0^1 \log(\varphi')^2 + \int_0^1 \log(\Check{\varphi}') - 2 \int_0^1 \log(\varphi') \log(\Check{\varphi}'). 
\end{align*} 
Moreover, 
\begin{align*}
\left( \int_0^1 \log\left( \frac{\varphi'}{\Check{\varphi}'} \right) \right)^2 &= \left[ \int_0^1 \log(\varphi') - \int_0^1 \log(\Check{\varphi}') \right]^2 \\
&= \left( \int_0^1 \log(\varphi') \right)^2 + \left( \int_0^1 \log(\Check{\varphi}') \right)^2 - 2 \left( \int_0^1 \log(\varphi') \right) \left( \int_0^1 \log( \Check{\varphi}') \right).
\end{align*}
Thus, 
\begin{align*}
&\int_0^1 \log\left( \frac{\varphi'}{ \Check{\varphi}'} \right)^2 - \left( \int_0^1 \log\left( \frac{\varphi'}{\Check{\varphi}'}\right)  \right)^2\\ 
&= \int_0^1 \log(\varphi')^2 - \left( \int_0^1 \log(\varphi') \right)^2 + \int_0^1 \log(\Check{\varphi}')^2 - \left( \int_0^1 \log(\Check{\varphi}') \right)^2 \\
& - 2 \left[ \int_0^1 \log(\varphi') \log(\Check{\varphi}') - \left( \int_0^1 \log(\varphi') \right) \left( \int_0^1 \log(\Check{\varphi}') \right) \right] \\
&= \| \varphi \|_{\Gamma_1}^2 + \| \Check{\varphi}' \|_{\Gamma_1}^2- 2 \langle \varphi, \Check{\varphi} \rangle_{\Gamma_1} \\
&= \| \varphi \ominus_{\Gamma_1} \Check{\varphi} \|_{\Gamma_1}^2.
\end{align*}
Thus, we justified Equation~\eqref{decompo_biais_variance}: 
\[ \mathbb{E}\left[ \| w_m^\MZW(\cdot;\theta) \ominus_{\Gamma_1} \Check{\varphi} \|_{\Gamma_1}^2 \right]   = \sum_{i=1}^m \Var[G_i] + \| \varphi \ominus_{\Gamma_1} \Check{\varphi} \|_{\Gamma_1}^2 .\] 
Hence, $\Check{\varphi}$ minimizes $\mathbb{E}\left[ \| \varphi \ominus_{\Gamma_1} \cdot \|_{\Gamma_1}^2 \right]$ if and only if $ \Check{\varphi} \ominus_{\Gamma_1} \varphi$ is equal to the identity function if and only if $ \Check{\varphi} = \varphi.$
\end{proof}

\subsection{Proof of the behaviour of the first two moments in a particular case }\label{appendix::cv_moments_MZW}

Numerical illustrations presented in Subsection~\ref{ssec:num_study} highlight that considering large values of the concentration parameter $\theta$ in the modified version of MZW Algorithm allows to simulate a warping process that can be centered at a target warping function $\varphi$ belonging to $\Gamma_1$ and presenting no variability. In this subsection, we give a proof of this observation in the case where $(G_i)_{i \geq 1}$ is a sequence of independent normally distributed random variables with mean 0 and variance term equals to $v_i / (1 + \theta),$ for $i \geq 1.$

\begin{proposition}
Let $m$ be a positive integer and $(v_i)_{i \geq 1}$ be a sequence of real values satisfying $\sum_{i = 1}^\infty v_i < \infty.$  We assume that $(G_i)_{i \geq 1}$ is a sequence of independent normally distributed random variables with mean 0 and variance term $v_i/(1+\theta),$ for $i \in \llbracket1,m\rrbracket.$ Then, for all $ \varphi \in \Gamma_1$ and all $t \in [0,1],$ we have, as $\theta \to \infty$
\begin{equation*}
\mathbb{E}\left[ w_m^{\MZW}(t;\theta)\right] \to \varphi(t) 
\quad \text{ and } \quad     
    \Var[w_m^{\MZW}(t;\theta)] \to  0.
\end{equation*}

\end{proposition}

\begin{proof}
Let $m \in \mathbb{N}^\ast.$ For a given $i \in \llbracket 1,m \rrbracket,$ the definition of $\phi_i$ implies that $- \sqrt{2} \leq \phi_i \leq \sqrt{2}.$ Consequently, we have the following inequalities: 
\begin{equation}\label{eq:bornes_phi} - \sqrt{2} \sum_{i=1}^m |G_i | \leq \sum_{i=1}^m G_i \phi_i \leq \sqrt{2} \sum_{i=1}^m |G_i|. 
\end{equation} 
Then, using the definition of $X_m(\cdot, \theta),$ we get 
\begin{equation*}
    \exp{\left(- \sqrt{2} \sum_{i=1}^m |G_i | \right)} \int_0^t \varphi' \leq \int_0^t \exp{(X_m(\cdot, \theta))} \leq \exp{\left(\sqrt{2} \sum_{i=1}^m |G_i|\right)}\int_0^t \varphi', 
\end{equation*} which can be rewritten as 
\begin{equation*}
    \exp{\left(- \sqrt{2} \sum_{i=1}^m |G_i | \right)} \varphi(t) \leq \int_0^t \exp{(X_m(\cdot, \theta))} \leq  \exp{\left(\sqrt{2} \sum_{i=1}^m |G_i|\right)} \varphi(t).
\end{equation*}
In particular, using the equation $\varphi(1) = 1,$ we obtain the following bound
\begin{equation*}
    \int_0^1 \exp{(X_m(\cdot, \theta))} \leq \exp{\left(\sqrt{2} \sum_{i=1}^m |G_i|\right)}.
\end{equation*}
Let $t \in [0,1].$ We deduce that \begin{equation*}
    \frac{\varphi(t) \exp{\left(- \sqrt{2} \sum_{i=1}^m |G_i|\right)}}{\exp\left(\sqrt{2} \sum_{i=1}^m |G_i|\right)}\leq w_m^{\MZW}(t;\theta)
\end{equation*}
and
\begin{equation*}
    \mathds{E}[ w_m^{\MZW}(t;\theta)] \geq \varphi(t) \mathds{E}\left[ \exp{\left( -2 \sqrt{2} \sum_{i=1}^m |G_i|\right)}\right]. 
\end{equation*}
Using Jensen's inequality to $x \mapsto \exp{(-x)},$ we get \begin{equation*}
    \mathds{E}[w_m^{\MZW}(t;\theta)] \geq \varphi(t) \exp{\left( - \mathds{E}\left[ 2 \sqrt{2} \sum_{i=1}^m |G_i| \right] \right)}.
\end{equation*}
Since $|G_i|,$ for $i \in \llbracket 1, m\rrbracket,$ is Folded Normal distributed with parameters $0$ and $v_i / (1 +\theta),$ the expectation of $|G_i|$ is given by $\sqrt{2 v_i}/ \sqrt{ \pi ( 1 + \theta)}.$ Consequently: \begin{equation*}
    \mathds{E}[ w_m^{\MZW}(t;\theta)] \geq \varphi(t) \exp{\left(- \frac{4}{\sqrt{\pi}} \frac{1}{\sqrt{1 + \theta}} \sum_{i=1}^m \sqrt{v_i}\right)}.
\end{equation*}
Now, we derive an upper bound of the term $ \mathbb{E}[w_m^{\MZW}(t;\theta)].$ Since the inequality $\exp{(x)} \geq 1 + x$ holds for $x \in \mathbb{R},$ we get \begin{equation*}
    \int_0^1 \exp{(X_m(\cdot; \theta))} \geq \int_0^1 (1  + X_m(\cdot; \theta)).
\end{equation*}
Since $X_m(\cdot; \theta)$ belongs to $H(0,1),$ we deduce that \begin{equation*}
     \int_0^1 \exp{(X_m(\cdot; \theta))} \geq 1
\end{equation*}
and then that
\begin{equation*}
    w_m^{\MZW}(t;\theta) \leq \int_0^t \varphi' \exp{\left( \sum_{i=1}^m G_i \phi_i \right)}.
\end{equation*} 
Then, we have 
\begin{equation*}
    \mathds{E}\left[ w_m^{\MZW}(t;\theta) \right] \leq \int_0^t \varphi'(u) M_{\sum_{i=1}^m G_i \phi_i(u)}(1) \mathrm{d}u,
\end{equation*} 
where, for $u \in [0,1],$ $M_{\sum_{i=1}^m G_i \phi_i(u)}$ refers to the moment generating function of the random variable $\sum_{i=1}^m G_i \phi_i(u).$ Since the random variables $G_i,$ for $i \in \llbracket 1, m\rrbracket,$ are independent and follow a Normal distribution, the random variable $ \sum_{i=1}^m G_i \phi_i(u) $ is normally distributed with mean 0 and variance parameter $(\sum_{i=1}^m \phi_i(u)^2 v_i) / (1 + \theta)$. Consequently, 
\begin{align*}
    M_{\sum_{i=1}^m G_i \phi_i}(1) &= \exp\left( \frac{1}{2} \frac{1}{1 + \theta} \sum_{i=1}^m \phi_i^2 v_i \right) \\
    &\leq \exp \left( \frac{1}{1 + \theta} \sum_{i=1}^m v_i \right). 
\end{align*}
Thus, we get 
\begin{align*}
    \mathds{E}\left[ w_m^\MZW(t;\theta) \right] &\leq \exp\left( \frac{1}{1 + \theta} \sum_{i=1}^m v_i \right) \int_0^t \varphi'\\
    &=  \exp\left( \frac{1}{1 + \theta} \sum_{i=1}^m v_i \right) \varphi(t).
\end{align*}

Since, for a fixed $m,$ the limit term, as $\theta\to\infty,$ of the upper bound and lower bound of $\mathds{E}\left[ w_m^{\MZW}(t;\theta) \right]$ is given by $\varphi(t),$ we deduce that, as $\theta\to\infty,$ 
\begin{equation}\label{eq:cv_esp_MZW}
    \mathds{E}\left[ w_m^{\MZW}(t;\theta) \right] \xrightarrow{} \varphi(t).
\end{equation}
We now focus on the variance term. Let $m \in \mathbb{N}^\ast.$ We deduce from~\eqref{eq:cv_esp_MZW} that \begin{equation*}
    \underset{\theta \to \infty}{ \lim} \mathds{E}[w_m^{\MZW}(t;\theta)]^2 = \varphi(t)^2.
\end{equation*}
Moreover, we have  
\begin{align*}
    \mathds{E}[w_m^\MZW(t;\theta)^2] &= \mathds{E}\left[ \left( \frac{\int_0^t \exp{\left( \sum_{i=1}^m G_i \phi_i \right)} \varphi'}{\int_0^1 \exp{\left( \sum_{i=1}^m G_i \phi_i \right)} \varphi'} \right)^2 \right] \\
    &\leq \mathds{E}\left[ \left( \int_0^t \exp{\left( \sum_{i=1}^m G_i \phi_i \right)} \varphi' \right)^2 \right] \\
    &= \mathds{E}\left[ \int_0^t \int_0^t \varphi'(u) \varphi'(v) \exp\left( \sum_{i=1}^m G_i (\phi_i(u) + \phi_i(v)) \right) \mathrm{d}u \mathrm{d}v \right] \\
    &= \int_0^t \int_0^t \varphi'(u) \varphi'(v) M_{\sum_{i=1}^m G_i (\phi_i(u) + \phi_i(v))}(1) \mathrm{d}u \mathrm{d}v . 
\end{align*}    
Since, for a given $(u,v) \in [0,t]^2,$ the random variable $\sum_{i=1}^m G_i (\phi_i(u) + \phi_i(v))$ is normally distributed with mean $0$ and variance $\sum_{i=1}^m v_i (\phi_i(u) + \phi_i(v))^2 / (1 + \theta) ,$ which is upper-bounded by $8 \sum_{i=1}^m v_i / (1 + \theta)$, its moment generating function at $1$ is therefore upper-bounded by $ \exp(4 \sum_{i=1}^m v_i / ( 1 + \theta)).$ Consequently, \begin{align*}
     \mathds{E}[w_m^\MZW(t;\theta)^2] & \leq \left( \int_0^t \varphi'(u) \mathrm{d}u \right) \left( \int_0^t \varphi'(v) \mathrm{d}v \right) \exp\left( \frac{4}{1 + \theta} \sum_{i=1}^m v_i \right) \\
     &= \varphi(t)^2 \exp\left( \frac{4}{1 + \theta} \sum_{i=1}^m v_i \right).
\end{align*}
Thus, 
\begin{equation*}
    \mathds{E}[w_m^{\MZW}(t;\theta)]^2 \leq \mathds{E}\left[ w_m^{\MZW}(t;\theta)^2 \right] \leq \varphi(t)^2 \exp\left( \frac{4}{1 + \theta} \sum_{i=1}^m v_i \right).
\end{equation*}
Hence
\begin{equation*}
    \underset{\theta \to \infty}{\lim} \mathds{E}\left[ w_m^{\MZW}(t;\theta)^2 \right] = \varphi(t)^2,
\end{equation*} and \begin{equation*}
    \underset{\theta \to \infty}{\lim}  \Var[ w_m^{\MZW}(t;\theta)] = 0.
\end{equation*}
\end{proof}
\end{appendix}

%%%%%%%%%%%%%%%%%%%%%%%%%%%%%%%%%%%%%%%%%%%%%%
%% Acknowledgements                         %%
%% should be provided in the                %%
%% Acknowledgements section.                %%
%%%%%%%%%%%%%%%%%%%%%%%%%%%%%%%%%%%%%%%%%%%%%%
\begin{acks}[Acknowledgments]
\rev{The authors are grateful to the referees for their comments and suggestions, which have helped improve the first version of this manuscript.} The authors also thank Julie Mireille Thériault (UQÀM) for providing the temperature data and for her help in obtaining research funding. \rev{Finally,} the authors are grateful to Sebastian Kurtek for discussions about the BK Algorithm.
\end{acks}

%%%%%%%%%%%%%%%%%%%%%%%%%%%%%%%%%%%%%%%%%%%%%%
%% Funding information, if any,             %%
%% should be provided in the                %%
%% funding section.                         %%
%%%%%%%%%%%%%%%%%%%%%%%%%%%%%%%%%%%%%%%%%%%%%%
\begin{funding}
This project has received financial support from the CNRS through the MITI interdisciplinary programs and Institut des Mathématiques pour la Planète Terre.
\end{funding}

\medskip
\noindent{\bf Code and Data Availability}\\

\noindent\rev{Temperature data are extracted from the Copernicus Climate Change Service (C3S) Climate Data Store (CDS)~\citep{hersbach2023era5}. The algorithms, simulations and figures are implemented and produced in the Julia language. The simulation algorithms and the code producing all figures of Section~\ref{ssec:num_study} can be found in \url{https://gricad-gitlab.univ-grenoble-alpes.fr/lemehaun/warping-simulation}.}

%%  MR numbers will be added by VTeX.                      %%
%%                                                         %%
%%  Use \cite{...} to cite references in text.             %%
%%                                                         %%
%%%%%%%%%%%%%%%%%%%%%%%%%%%%%%%%%%%%%%%%%%%%%%%%%%%%%%%%%%%%%

%% if your bibliography is in bibtex format, uncomment commands:
\bibliographystyle{imsart-nameyear} % Style BST file (imsart-number.bst or imsart-nameyear.bst)
\bibliography{refs}       % Bibliography file (usually '*.bib')

\end{document}